\documentclass[prx,aps,superscriptaddress,showpacs,floatfix,tightenlines,reprint]{revtex4-2}
\usepackage[utf8]{inputenc}
\usepackage[T1]{fontenc}
\usepackage{graphicx}
\usepackage{dcolumn}
\usepackage{bm}
\usepackage{amssymb}
\usepackage{bbm}
\usepackage{dsfont}
\usepackage{bbold}
\usepackage{mathbbol}
\usepackage{amsmath}
\usepackage{url}
\usepackage{amsthm}
\usepackage{layout}
\usepackage{epsfig}
\usepackage{epstopdf}
\usepackage{booktabs}
\usepackage{float}
\usepackage{subfigure}
\usepackage{lipsum}
\usepackage{mathrsfs}
\usepackage{blindtext}
\usepackage{appendix}
\usepackage[hyperindex,breaklinks]{hyperref}
\usepackage{enumerate}
\usepackage{enumitem}
\usepackage{makecell}
\usepackage{multirow}
\usepackage{tabularx}
\usepackage{array}
\usepackage{subfigure}
\usepackage{booktabs}   
\usepackage{threeparttable}
\allowdisplaybreaks[4]

\theoremstyle{definition}
\newtheorem{theorem}{\normalfont Theorem}
\newtheorem{lemma}{\normalfont Lemma}

\makeatletter
\renewenvironment{proof}[1][\proofname]{\par
  \pushQED{\qed}
  \normalfont\topsep6\p@\@plus6\p@\relax
  \trivlist
  \item[\hskip\labelsep
        \normalfont
        #1\@addpunct{.}]\ignorespaces
}{%
  \popQED\endtrivlist\@endpefalse
}
\makeatother

\allowdisplaybreaks[3]

\usepackage{xcolor}

\newcommand{\yaqi}[1]{\textcolor{black}{#1}}

\begin{document}

\title{Negativity Percolation in Continuous-Variable Quantum Networks}

\setcounter{footnote}{3}
\author{Yaqi Zhao}
\affiliation{College of Mathematics, Taiyuan University of Technology, Taiyuan, 030024, China}

\author{Kan He}
\email{hekan@tyut.edu.cn}
\affiliation{College of Mathematics, Taiyuan University of Technology, Taiyuan, 030024, China}
\affiliation{School of Mathematics and Statistics, Shaanxi Normal University, Xi'an, 710119, China}

\author{Yongtao~Zhang}
\affiliation{Network Science and Technology Center, Rensselaer Polytechnic Institute, Troy, New York 12180, USA}
\affiliation{Department of Physics, Applied Physics, and Astronomy, Rensselaer Polytechnic Institute, Troy, New York 12180, USA}

\author{Jinchuan Hou}
\email{jinchuanhou@aliyun.com}
\affiliation{College of Mathematics, Taiyuan University of Technology, Taiyuan, 030024, China}

\author{Jianxi Gao}
\affiliation{Network Science and Technology Center, Rensselaer Polytechnic Institute, Troy, New York 12180, USA}
\affiliation{Department of Computer Science, Rensselaer Polytechnic Institute, Troy, New York 12180, USA}

\author{Shlomo Havlin}
\affiliation{Department of Physics, Bar-Ilan University, Ramat Gan 52900, Israel}

\author{Xiangyi Meng}
\email{xmenggroup@gmail.com}
\affiliation{Network Science and Technology Center, Rensselaer Polytechnic Institute, Troy, New York 12180, USA}
\affiliation{Department of Physics, Applied Physics, and Astronomy, Rensselaer Polytechnic Institute, Troy, New York 12180, USA}

\date{\today }
\begin{abstract}
\section*{Abstract}

Quantum networks (QNs) have been predominantly driven by discrete‑variable (DV) architectures. Yet, optical platforms naturally generate Gaussian states---the common states of continuous-variable (CV) systems, making CV-based QNs an attractive route toward scalable, chip‑integrated quantum computation and communication. To bridge the gap between well‑studied DV entanglement percolation theories and their CV counterpart, we introduce a Gaussian-to-Gaussian entanglement distribution scheme that deterministically transports two‑mode squeezed vacuum states across large CV networks. 
Analysis of the scheme's collective behavior using statistical-physics methods reveals a new form of entanglement percolation---negativity percolation theory (NegPT)---characterized by a bounded entanglement measure called the ratio negativity.
We discover that NegPT exhibits a mixed-order phase transition, marked simultaneously by both an abrupt change in global entanglement and a long-range correlation between nodes. This distinctive behavior places CV-based QNs in a new universality class, fundamentally distinct from DV systems.
Additionally, the abruptness of this transition introduces a critical vulnerability of CV-based QNs: conventional feedback mechanism becomes inherently unstable near the threshold, highlighting practical implications for stabilizing large-scale CV-based QNs. Our results unify statistical models for CV-based entanglement distribution \yaqi{and} uncover previously unexplored critical phenomena unique to CV systems, providing valuable insights and guidelines essential for developing robust, feedback-stabilized QNs.
\end{abstract}

\maketitle

\section*{Introduction}

\begin{center}
    \begin{table*}[t]
    \begin{threeparttable}
    \caption{\yaqi{Comparison of percolation theories (PTs) for entanglement distribution on discrete-variable (DV) and continuous-variable (CV) quantum networks (QNs). The PTs are based on distinct entanglement measures: concurrence~\cite{Hill1997,G_concurrence2005} for DV qudit systems, and ratio negativity~\cite{Ratio_negativity2024} for CV infinite-dimensional squeezed states.}}
    \label{table_comparision}
    \begin{tabular}{l|c|c|c}
     \Xhline{1pt}
		\hline\hline
		&  Entanglement distribution scheme & PT / entanglement measure  & Phase transition type\\
		\hline
		  {DV} & Deterministic entanglement transmission (DET) & ConPT / Concurrence~$c$ & {Continuous} (second order)\\
		{CV} & Gaussian-to-Gaussian (G-G) DET &  NegPT / Ratio negativity~$\chi$ &  {Discontinuous} (mixed order) \\
         \noalign{\smallskip}\hline
	\end{tabular}
\end{threeparttable}
\end{table*}
\end{center}

{Quantum information technologies~\cite{MI2000} fundamentally rely on our ability to {distribute entanglement}~\cite{CGCRAW1993,JPHH1997},
the process of establishing quantum correlations between distant nodes across network architectures, from chips to the internet. This capability  empowers applications ranging from quantum computation~\cite{Cohen2018,Hu2024} to secure quantum communication\yaqi{~\cite{CS1992,HWJP1998,NCHN2011,Munro2012,QKD_photon2022,secure_network2022,QKD2023}}.
The collective behavior of entanglement distribution across large-scale systems is often described by the quantum network (QN) model~\cite{AJM2007}. 
Within this framework, a QN is represented as a graph where nodes represent spatially separated parties, and links denote bipartite entangled states shared between adjacent nodes. 
The distribution process itself is executed via protocols based on local operations and classical communication (LOCC), such as entanglement swapping~\cite{MAMA1993,SVP1999} and entanglement concentration~\cite{CHSB1996,FM2001}, along with more sophisticated schemes~\cite{KH2016,DES2017,RS2017,MHD2019}, applied across the network's links to establish long-range correlations.}

{Investigating how LOCC schemes distribute entanglement across a QN has uncovered fundamental correspondences to percolation theories. The ability to establish long-range entanglement in QN often mirrors the emergence of large-scale connectivity described by percolation theories---the statistical physics framework governing phase transitions on graphs~\cite{cross-probab-square_k80,cross-probab-cft_c92,netw-percolation_ceah00,BOCCALETTI2006,ARENAS2008,hypergr-percolation_bd23,hypergr-percolation_lglk23}.
This conceptual bridge between entanglement distribution and percolation was first built in 2007, linking a probabilistic entanglement distribution scheme~\cite{AJM2007} to classical bond percolation~\cite{JWE1980}. 
Later, the development of a more effective deterministic entanglement transmission (DET) scheme~\cite{XYJSA2023} (where LOCC steps succeed deterministically, eliminating classical randomness unlike earlier probabilistic schemes) led to its mapping onto an essentially  distinct variant, concurrence percolation theory (ConPT)~\cite{XJS2021,OXSGBJ2022}.
This interplay continues with more recent mappings identified between quantum memories and continuum percolation~\cite{Meng2024}, and between entanglement routing and path percolation~\cite{path-percolation_kr24,path-percolation_mhrk24}, underscoring how incorporating quantum resources like memories or repeaters may also profoundly alter the underlying percolation landscape.
}

{These established mappings are, however, exclusively confined to discrete-variable (DV) systems (e.g.,~qubits). They all overlook an alternative architecture particularly prominent in optical settings: the {continuous-variable (CV) system}, leveraging the continuous degrees of freedom of light fields (amplitude and phase)~\yaqi{\cite{SP2005,Wang2007,CSRNTJS2012,quantumCV_2023}}.
In contrast to DV encoding, which is inherently reliant on random single-photon sources, the CV encoding is more advantageous with its unconditional, consistent generation of entanglement via nonlinear optical interactions, bypassing the hurdles of unpredictability in DV sources. This nonrandom nature unlocks considerable potential for the scalability of quantum optical applications~\cite{EMJ2013,JT2017,FW2018,JMNT2020}, a potential now being realized with the advent of CV photonic chips~\cite{CV_code_chip2025}.}

{Even with their clear advantages and growing experimental relevance, the  collective characteristics of CV systems---particularly as modeled by QN---remain unexplored.
Crucial questions persist: How to distribute entanglement across a CV-encoded QN? Does a unifying correspondence to  percolation theory also apply? If so, does the continuous nature of CV lead to fundamentally distinct network physics compared to DV systems?}

To address these conceptual gaps, here we
introduce a Gaussian-to-Gaussian (G-G) DET scheme, wherein both the inputs and outputs are {two-mode squeezed vacuum states} (TMSVSs), {defined as} $|\psi^{r}\rangle=\sqrt{1-\tanh^2 r}\sum\nolimits_{n=0}^{\infty}\tanh^n r|nn\rangle$, which are a class of entangled CV Gaussian states characterized by the squeezing parameter $r\in [0,\infty)$~\cite{SP2005,Wang2007,CSRNTJS2012}.
Akin to DV-based DET~\cite{XYJSA2023}, the G-G DET employs deterministic entanglement swapping~\cite{P2002} and concentration~\cite{MRN} adapted for TMSVSs, offering an avenue for entanglement distribution in optical systems.

We find that the G-G DET scheme corresponds to a new variant of percolation: the {negativity percolation theory} (NegPT). The theory is distinct from its DV counterparts: we define it neither on probability~\cite{JWE1980} nor concurrence~\cite{Wootters1998}, but on a bounded entanglement measure $\chi_{\mathcal{N}}\in[0,1]$ termed {ratio negativity}~\cite{Ratio_negativity2024} that simplifies to $\chi\equiv \tanh r$ for TMSVS $|\psi^{r}\rangle$.
This presents a broader spectrum of correspondences between entanglement distribution schemes and percolation theories 
(Table~\ref{table_comparision}).

Further, we show that the NegPT exhibits a mixed-order phase transition~\cite{mix-order-phase-transit_m23},
characterized by both a discontinuous change at a critical threshold $\chi_\text{th}$ and a divergent correlation length $l^*\sim \left|\chi-\chi_\text{th}\right|^{-z\nu}$. 
{While mixed-order phase transitions have been widely observed {in spin systems~\cite{mix-order-phase-transit_m23}, interdependent networks~\cite{Interdependent_net2014}, 
superconducting devices~\cite{interdependent_bglvfh23},
overload systems~\cite{mix-order-phase-transit_pbplrbh24}, $k$-core percolation~\cite{k_core2024}}, and colloidal crystals~\cite{mix-order-phase-transit_atc17}, their observation in QNs is unprecedented. Our discovery extends the reach of mixed-order phase transitions beyond traditional statistical physics systems, into the rapidly evolving field of quantum information for the first time. Indeed,}
our finding of the NegPT, with a unique thermal critical exponent $z\nu\approx  1/2$ in Bethe lattices, {clearly highlights a universality class that is distinct from both classical and concurrence percolation (where $z\nu\approx 1$)~\cite{XJS2021,OXSGBJ2022}.} 
We also notice that the onset of mixed-order transitions may persist for generalized CV operations and even non-Gaussian states, suggesting that the NegPT may underlie the critical phenomena for generic CV-based QNs.

{Moreover, our theoretical analysis reveals a potential practical {vulnerability} within large-scale CV-based QNs: We show that the sharp transition at $\chi_\text{th}$ challenges stabilization via standard feedback, risking unstable ``on/off'' oscillations near the threshold---in stark contrast with continuous transitions. Given the prevalence of quantum feedback control~\cite{feedback_coherent1994,quantum_control1999,feedback_coherenct2000,feedback2002,feedback2002_2,feedback2005,quantum_control2008Domenico,quantum_control2010,feedback2012,feedback2012_2,feedback2017,feedback_remote_entanglement2015} in quantum optical implementations, this instability acts as a crucial {alert}: scaling CV systems requires careful feedback strategies to maintain robust operation near the ``brink of failure.'' We initiate the characterization of this behavior and its network implications.}

\section*{Results}

\subsection{Gaussian QN}

To begin with, we consider a Gaussian QN built by bipartite, pure Gaussian states $|\Psi\rangle_{AB}$ involving multiple modes (say, $K_A$ modes + $K_B$ modes) in general. A key insight for Gaussian states is that $|\Psi\rangle_{AB}$ can be mapped into a tensor product of $K\yaqi{\leq}\min{(K_A,K_B)}$ TMSVSs (along with trivial vacuum states) via local Gaussian unitaries~\cite{Modewise2003}: $|\Psi\rangle_{AB}=\bigotimes_{k=1}^{K}|\psi^{r_{k}}\rangle\otimes |00\dots\rangle$,
where each $r_k$ corresponds to the squeezing parameter of the $k$-th TMSVS. 
This mode decomposition~\cite{Modewise2003} simplifies our QN model to primarily focus on TMSVSs, such that each TMSVS represents a QN link between two nodes. A bipartite multi-mode pure Gaussian state can be treated as multi-links that connect the same pair of nodes multiple times, which are allowed in the formation of the QN.

\begin{figure*}
    \centering
    \includegraphics[width=360pt]{Figure1.png}
    
    \caption{\yaqi{Gaussian-to-Gaussian deterministic entanglement transmission (G-G DET) scheme. Applicable to Gaussian quantum networks (QN), the scheme consists of two LOCC protocols:  
    (a)~Entanglement swapping, facilitated by homodyne detection and displacement~\cite{P2002};
    (b)~Entanglement concentration, facilitated by non-standard optical components~\cite{MRN}. Both protocols are deterministic, taking two (or more) TMSVS $|\psi^{r_{i}}\rangle$ as input and a new TMSVS $|\psi^{r}\rangle$ as output.
    (c)~The two LOCC protocols map to series and parallel rules, respectively, to construct G-G DET.
    (d)~Consider a QN example built upon three node. The G-G DET scheme consists of two steps:
   First, the parallel rule converts the states $|\psi^{r_1}\rangle$ and $|\psi^{r_2}\rangle$ ($r_1\geq r_2$) into $|\psi^{r_{1,2}}\rangle$ with $\sinh r_{1,2} = \sinh r_1 \cosh r_2$ between $S$ and $R$; second, the series rule transforms $|\psi^{r_{1,2}}\rangle$ and $|\psi^{r_{3}}\rangle$ to the final state $|\psi^{r}\rangle$ with the ratio negativity $\text{X}_{\text{SC}}=\tanh r_{1,2} \tanh r_3$ between $S$ and $T$.}}\label{figure1}
\end{figure*}

\subsection{A DET scheme for series-parallel Gaussian QNs}

Our G-G DET scheme, following and extending the original DET scheme proposed for DV systems~\cite{XYJSA2023}, consists of entanglement swapping and concentration operations applied to Gaussian states.
The two operations, taking TMSVSs as the input and a new TMSVS as the output, are inherently deterministic in the CV domain. These operations can be intuitively understood as ``series and parallel rules,'' as commonly used in studying the effective resistance in circuit networks~\cite{XYJSA2023}. Hence, our G-G DET scheme can be applied to any network with a series-parallel topology~\cite{Duffin1965}. The corresponding series (swapping) and parallel (concentration) rules are introduced as follows:

First, we introduce the G-G entanglement swapping.
Consider a setup involving three spatially separated parties: the Source ($S$), Relay ($R$), and Target ($T$). A TMSVS $|\psi^{r_{1}}\rangle$ is shared between $S$ and $R$, while another $|\psi^{r_{2}}\rangle$ is shared between $R$ and $T$ \yaqi{[Fig.~\ref{figure1}(a)]}. 
The two parties $S$ and $T$ do not directly share any entangled state.
Nevertheless, executing homodyne measurement at $R$ and applying matching displacements at $S$ and $T$ can deterministically convert $|\psi^{r_{1}}\rangle$ and $|\psi^{r_{2}}\rangle$ into a new TMSVS, $|\psi^{r}\rangle$, directly shared between $S$ and $T$~\cite{P2002}. The new squeezing parameter $r$ is (ideally) given by $\tanh{r}=\tanh{r_{1}}\tanh{r_{2}}$~\cite{P2002}.
This can be extended to a one-dimensional chain consisting of a series of $N$ links corresponding to TMSVSs $|\psi^{r_1}\rangle,|\psi^{r_2}\rangle,\cdots,|\psi^{r_{N}}\rangle$ (see \yaqi{Methods}), respectively, leading to the series rule [\yaqi{Fig}.~\ref{figure1}(c)] for $|\psi^r\rangle$ being the final created TMSVS.

Note that the series rule is
independent of the order in which the links are transformed. This commutativity, as we will see, does not apply to the parallel rule below.

\begin{table*}[t!]
\centering
\begin{threeparttable}
    \caption{Connectivity rules for the NegPT on different topologies. Each link has identical ratio negativity $\chi\equiv \tanh r$.}
    \label{table_negpt}
    \begin{tabular}{c|cc|ccc}
     \Xhline{1pt}
    \hline\hline\noalign{\smallskip}
     & \multicolumn{2}{c|}{Exact rules} & \multicolumn{3}{c}{Approximate rules (by the star-mesh transform~\cite{XJS2021,OXSGBJ2022})} \\
     \noalign{\smallskip}\hline\noalign{\smallskip}
    Topology:  &Series & Parallel & Wheatstone bridge\tnote{\yaqi{a}} & Kelvin bridge\tnote{\yaqi{a}} & (higher-order topologies)\\
     & \includegraphics[width=70pt]{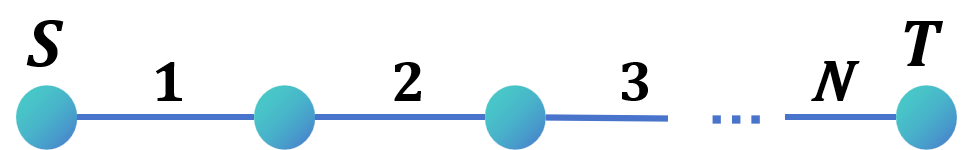} & \includegraphics[width=70pt]{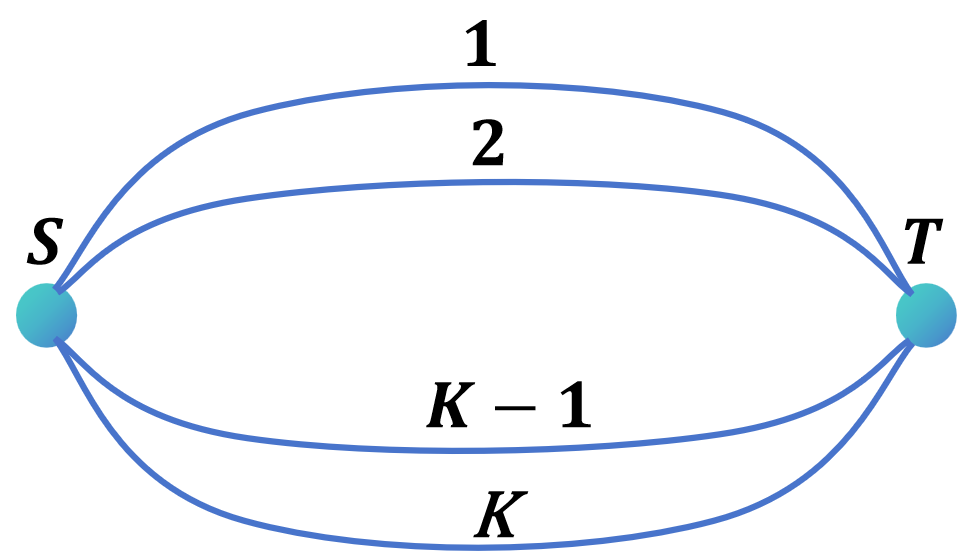} & \includegraphics[width=70pt]{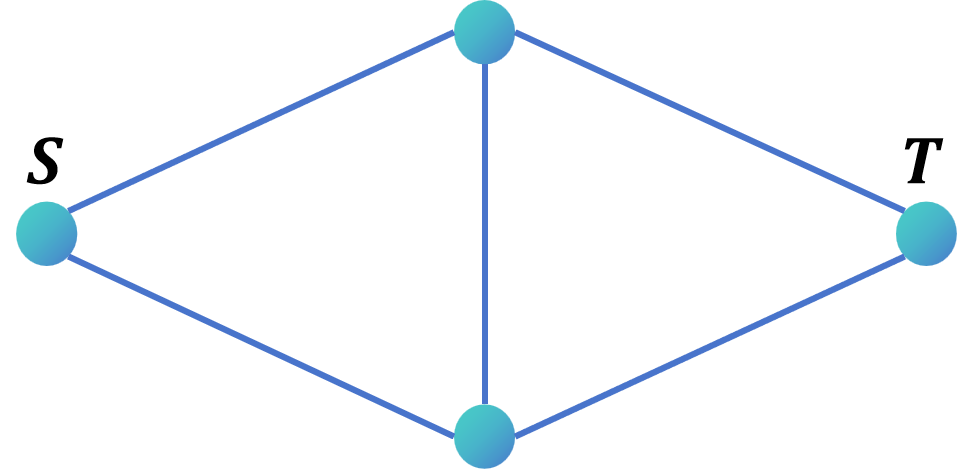} & \includegraphics[width=70pt]{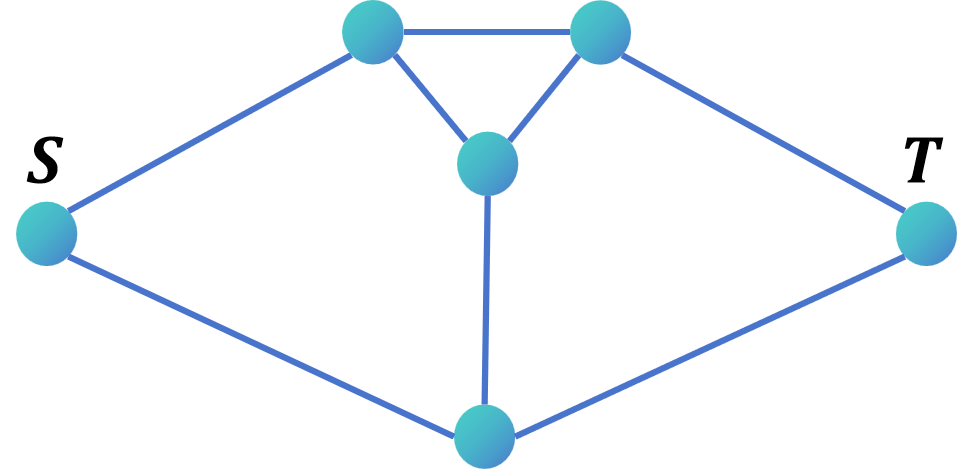} & $\dots$\\
    $\mathrm{X}_\text{SC}$ & $=\chi^N$ & $={\chi} \left[\chi^2+\left(1-\chi^{2}\right)^K\right]^{-1/2} $ & $\approx b(\chi,y(\chi))$ & $\approx b(\chi,y(\chi \delta(\chi)))$ & $\dots$\\
     \noalign{\smallskip}\hline
	\end{tabular}
    \begin{tablenotes}
    \footnotesize
    \item[\yaqi{a}]{
    $b(\chi,y)=\frac{\chi^2}{\sqrt{\chi^4+\left(1-y^2\right)\left[\left(\chi^2y^2-y^2+1\right)^2-\chi^4\right]}} $; $y(x)$ and $\delta(x)$ are the Y-$\Delta$ and $\Delta$-Y transforms of $x$, respectively~(\yaqi{see Methods}).
     }\\
    \end{tablenotes}
     \end{threeparttable}
\end{table*}

We then give the G-G entanglement concentration.
It is known that Gaussian states cannot be feasibly concentrated using standard optical components~\cite{JSM2002}.
The only exception we know of is provided in Ref.~\cite{MRN}, offering a conceptual construction of LOCC  for concentrating entanglement from pure Gaussian states. Such LOCC must involve non-standard optical components.
Assuming this, if we consider a tensor product of two TMSVSs {shared by two parties $S$ and $T$}, $|\psi^{r_{1}}\rangle\otimes |\psi^{r_{2}}\rangle$ ($r_1\geq r_2$), this state can be deterministically converted to a single TMSVS and a vacuum state, $|\psi^{r}\rangle \otimes |00\rangle$ by LOCC \yaqi{[Fig.~\ref{figure1}(b)]}~\cite{MRN}, where $\sinh{r}=\sinh{r_{1}}\cosh{r_{2}}$~\cite{MRN}. 
The result also scales to $K$ parallel links,
$|\Psi\rangle_{{ST}}=\mathop{\bigotimes}\nolimits_{k=1}^{K}|\psi^{r_{k}}\rangle$  (\yaqi{see Methods}). 
By iteratively applying the Gaussian concentration  $K-1$ times, the multi-mode state $|\Psi\rangle_{{ST}}$ is concentrated into a single TMSVS, leading to the parallel rule [\yaqi{Fig.}~\ref{figure1}(c)].

The parallel rule is not commutative due to the presence of a $\sinh r_1$ term before other $\cosh r_k$ terms. We find that the optimal entanglement of $|\psi^{r}\rangle$ is achieved by placing the largest squeezing parameter $r_1$ into $\sinh r_1$ (\yaqi{see Methods}).

\subsection{Negativity percolation theory  (NegPT)}

The G-G DET scheme provides the foundation for a percolation-like theory of ``quantum connectivity,'' which we term negativity percolation. We consider a homogeneous QN where each link represents a TMSVS with an identical entanglement weight $\chi$. This framework is analogous to classical bond percolation on a network where each link is active with an identical probability $p$~\cite{DS2000}. 

In classical percolation theory, the connectivity between two arbitrary sets of nodes, $S$ and $T$, can be quantified by the ``sponge-crossing'' probability, $P_\text{SC}$~\cite{XJS2021,OXSGBJ2022}. This quantity is the combination of probabilities over all paths connecting the two sets. Drawing a parallel, we introduce the ``sponge-crossing'' ratio negativity, $\mathrm{X}_\text{SC}$, which represents an analogous combination of all paths but is framed in terms of the ratio negativity. \yaqi{Specifically, consider a Gaussian QN with endpoints $S$ and $T$ that initially share no entanglement. By implementing the G-G DET scheme across the QN, $S$ and $T$ ultimately share a TMSVS $|\psi^r\rangle$, with entanglement established between them. The final entanglement is quantified by the ratio negativity of $|\psi^r\rangle$, which we refer to as the sponge-crossing ratio negativity $\mathrm{X}_\text{SC}$.}

For both classical and negativity percolation, this aggregation of all path contributions is calculated as follows:

(1) When the network topology between $S$ and $T$ is series-parallel, $\mathrm{X}_\text{SC}$ can be decomposed into the iteration of the exact series and parallel rules \yaqi{[e.g., Fig.~\ref{figure1}(d)]}. Additionally, in a series network \yaqi{of $N$  TMSVSs with identical ratio negativity $\chi$}, $\mathrm{X}_\text{SC}=\chi^N$ has the same form as the probability measure $p$ (or concurrence $c$) in classical (or concurrence) percolation~\cite{XJS2021,OXSGBJ2022}. 
This similarity prompts us to refer to the theory as the NegPT, emphasizing on the special role that ratio negativity $\chi$ plays that is analogous to $p$ (or $c$).

(2) When the network topology between $S$ and $T$ is not series-parallel, we do not know the exact transmission rules (beyond the scope of G-G DET). However, we can approximate these higher-order rules using transforms based {only} on the series/parallel rules (\yaqi{see Methods}). 

Table~\ref{table_negpt} {shows} examples on how to calculate $\mathrm{X}_\text{SC}$ (exactly or approximately) for various network topologies. When the number of QN nodes goes to infinity, similar to classical percolation, we assume that a nontrivial critical threshold exists,
i.e.,~a minimum entanglement $\chi_{\text{th}} \in [0, 1]$ per link that makes $\mathrm{X}_\text{SC}$ exceed zero.

\begin{figure*}[t!]
    \centering
     \includegraphics[width=500pt]{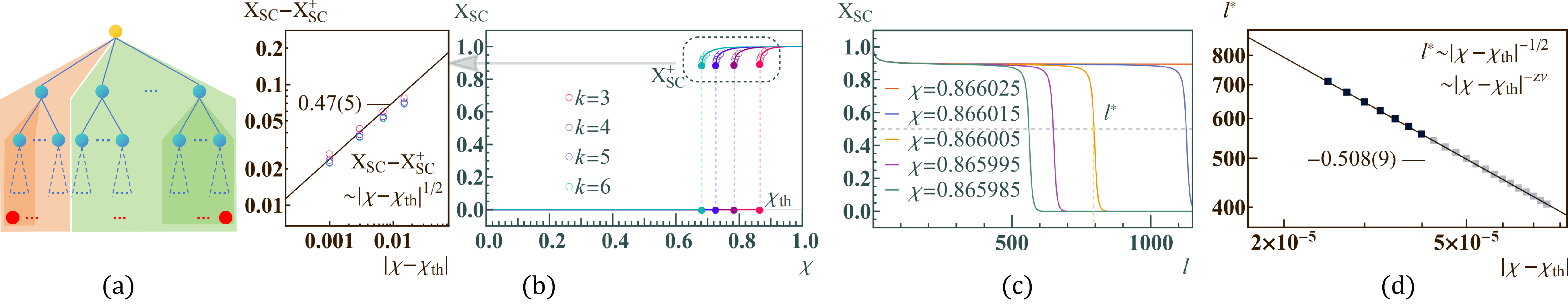}
    \caption{Bethe lattice.
    (a)~A Bethe lattice of degree $k$ (i.e.,~each node is incident to $k$ links) and network depth $l$ (the path length from the yellow node to the red nodes).
    (b)~The sponge-crossing ratio negativity $\mathrm{X}_\text{SC}$ {between $S$ and $T$ for various $k$ (right panel), satisfying the power law $\mathrm{X}_\text{SC}-\mathrm{X}_\text{SC}^{+}\sim |\chi-\chi_{\text{th}}|^{0.47(5)}$ as $\chi\to\chi_{\text{th}}^+$ (left panel).}
    \yaqi{The numerical value~$0.47\pm0.05$ is derived by a linear least-squares fit to the sixteen data points.}
    (c)~
    {When $\chi\to\chi_{\text{th}}^-$,} $\mathrm{X}_\text{SC}$ exhibits a plateau behavior until the network depth $l$ exceeds the correlation length $l^*$  
    {(defined as the depth $l$  at which $\mathrm{X}_\text{SC}=0.5$),}
    after which $\mathrm{X}_\text{SC}$ abruptly drops to zero.
    (d)~Near the critical threshold, we observe $l^* \sim \left|\chi - \chi_{\text{th}}\right|^{-0.508(9)}$, indicating $z\nu \approx 1/2$.
    }\label{Figure2}
  \end{figure*}

\subsection{Mixed-order phase transition}

We start by focusing on the Bethe lattice~\cite{DS2000}, {an infinite tree-like series-parallel network where every node has an identical degree $k>2$} [Fig.~\ref{Figure2}(a)]. 
The root (yellow) corresponds to the set $S$, and all the nodes (red) situated at the outermost boundary are collectively regarded as the set $T$. 
{We define the (finite) network depth $l$ as the number of links in the shortest path between $S$ and $T$ (see Supplemental Information (SI), Sec.~I~A).}

In the $l\to\infty$ limit, the Bethe lattice manifests self-similarity: the dark and light orange shaded regions are identical, as are the dark and light green shaded regions. This enables us to derive the exact self-consistent renormalization group~\cite{J1996,AE2002} equation (SI, Sec.~I~B).
We observe a distinctive phenomenon: as $\chi$ increases from zero, $\mathrm{X}_\text{SC}$ initially remains zero until $\chi=\chi_{\text{th}}$, at which point $\mathrm{X}_\text{SC}$  exhibits an abrupt, discontinuous jump to a positive value $\mathrm{X}_\text{SC}^{+}$ before continuing to increase [Fig.~\ref{Figure2}(b)]. 
For general $k$, we find 
  \begin{eqnarray}
  \label{GGthreshold1}
    \chi_{\text{th}}=\sqrt{1-(k-1)^{-\frac{k-1}{k-2}}\left(k-2\right)}.
  \end{eqnarray}
Moreover, at the critical threshold, $\mathrm{X}_\text{SC}$ jumps to 
\begin{eqnarray}
  \mathrm{X}_\text{SC}^{+}=\sqrt{\frac{1-T(k)}{T^k(k)-T(k)+1}},
\end{eqnarray}
where $T(k)=\left(k-1\right)^{-1/{\left(k-2\right)}}$ (SI, Sec.~I~C). This is in stark contrast with classical or concurrence percolation, where the phase transition is shown to be continuous~\cite{XJS2021,OXSGBJ2022}.
Another difference is that in concurrence percolation, a nontrivial saturation point exists, above which the sponge-crossing entanglement becomes unity (even when the entanglement per link is below one)~\cite{XJS2021,OXSGBJ2022}. Yet, there is no saturation point in the NegPT. Instead, 
we find $1-\mathrm{X}_\text{SC}\sim (1-\chi)^k$
when $\chi\to 1^{-}$, and $\mathrm{X}_\text{SC}$ becomes one if and only if $\chi=1$ (SI, Sec.~I~D).

The discontinuity around $\chi_{\text{th}}$ prompts us to investigate its critical nature, showing that the phase transition is, indeed, {mixed order} (SI, Sec.~I~E). 
Below the critical threshold ($\chi< \chi_{\text{th}}$), we find that $\mathrm{X}_\text{SC}$ initially remains constant as a function of $l$, forming a plateau near $\mathrm{X}_\text{SC}\approx \mathrm{X}_\text{SC}^{+}$ [Fig.~\ref{Figure2}(c)]. This plateau persists until $l$ reaches a finite {value} $l^*$, after which $\mathrm{X}_\text{SC}$ abruptly drops to zero. 
The same plateau is also observed in interdependent physical networks, signaling a dynamic characteristic of mixed-order phase transitions~\cite{Interdependent_net2014}.

By analogy with classical percolation~\cite{DS2000}, this value $l^*$ represents the characteristic length, signifying the typical distance within which two nodes remain connected.
The log-log fitting [Fig.~\ref{Figure2}(d)] further indicates that $l^* \sim \left|\chi - \chi_{\text{th}}\right|^{-1/2}$, suggesting that $l^*$ approaches infinity as $\chi\to\chi_{\text{th}}^{-}$, with a power law exponent $1/2$ that is equal to the chemical-space (``time'') thermal exponent, $z\nu$, as determined by finite-size scaling analysis~\cite{J1988}. Notably, $z\nu\approx 1/2$ differs from the classical and concurrence percolation theories, where $z\nu\approx 1$ for both~\cite{XJS2021,OXSGBJ2022}, indicating that the underlying dynamics of the NegPT is fundamentally unique {and belongs to a different universality class}.

To identify the other critical exponent, $\beta$, note that
akin to concurrence percolation, the series rule of the NegPT also has a ``gauge redundancy''
under the change of variable $\chi\to\chi^x$~\cite{XJS2021,OXSGBJ2022}. However, this redundancy no longer leads to arbitrariness in $\beta$. 
Indeed, for arbitrary $x$, we always have $\left(\mathrm{X}_\text{SC}\right)^x-\left(\mathrm{X}_\text{SC}^{+}\right)^x\simeq x \left(\mathrm{X}_\text{SC}^{+}\right)^{x-1}  \left(\mathrm{X}_\text{SC}-\mathrm{X}_\text{SC}^{+}\right)\sim\left(\chi-\chi_{\text{th}}\right)^{1/2}$ (SI, Sec.~I~E), indicating a fixed critical exponent $\beta=1/2$ [Fig.~\ref{Figure2}(b)] independent of the gauge redundancy $x$. 

Mixed-order phase transitions also occur in classical interdependent percolation~\cite{gao2011robustness,interdependent_network_2015}, prompting a comparison to NegPT. Consider classical percolation in $M$ fully interdependent multilayer systems~\cite{multilayer_network2014,interdependent_mp19} of the Bethe lattice topology, and let $P_\text{SC}^{(1)}$ represent the sponge-crossing probability across a single branch [shaded orange in Fig.~\ref{Figure2}(a)].
Per the excess-degree generating function interpretation~\cite{gene_function_2005}, full interdependency requires that all branches in the $M$ layers must be simultaneously percolating in order to enable the percolation of the next-level branch~\cite{gao2011robustness,interdependent_network_2015}. This  consideration leads to a self-consistent equation (SI, Sec.~I~F) that also features a mixed-order jump from $0$ to $P_\text{SC}^{(1),+}$ at the critical threshold $p=p_{\text{th}}$. As $k\to \infty$, one has $P_\text{SC}^{(1),+}\to 0$, which leads to the expansion (SI, Sec.~I~G)
\begin{equation}
  \fontencoding{T1}
    \fontsize{8}{10}\selectfont
  \label{eq_p_infty}
  P_\text{SC}^{(1),+}=p_{\text{th}}\left[1-M\left(1-P_\text{SC}^{(1),+}\right)^{k-1}+\dots\right]
\end{equation}
near the critical threshold.
Analogously, for the NegPT, we find that the jump of the sponge-crossing ratio negativity across a single branch, $\mathrm{X}_\text{SC}^{(1),+}$, satisfies (SI, Sec.~I~G)
\begin{equation}
\fontencoding{T1}
    \fontsize{8}{10}\selectfont
  \label{eq_chi_infty}
    \left(\mathrm{X}_\text{SC}^{(1),+}\right)^2=\chi_{\text{th}}^2\left[1-\frac{1}{\left(\mathrm{X}_\text{SC}^{(1),+}\right)^2}\left(1-\left(\mathrm{X}_\text{SC}^{(1),+}\right)^2\right)^{k-1}+\dots\right]
\end{equation}
{for $k>2$}.
Comparing Eqs.~\eqref{eq_p_infty}~and~\eqref{eq_chi_infty}, we see that as $k\to \infty$, 
the NegPT mirrors classical percolation in an $(\mathrm{X}_\text{SC}^{(1),+})^{-2}$-layer interdependent Bethe lattice under the change of variable $\chi^2 \to p$. 
In general, however, the two models are driven by fundamentally different mechanisms. While classical interdependent percolation requires an additional degree of freedom (i.e.,~the number of layers $M$) to induce a mixed-order transition, the NegPT does {not} require any such extra freedom.

\yaqi{To verify that the discontinuous nature of NegPT's phase transition extends beyond the Bethe lattice, we also performed numerical simulations of entanglement percolation on a square lattice using the star-mesh transform~\cite{XJS2021}. As shown in Fig.~\ref{Figure3}(a), 
the curve drops abruptly already for finite side length $L$, suggesting a (weakly) discontinuous transition~\cite{discontinuous1994,path-percolation_mhrk24} may indeed exist near the critical threshold $\chi_{\text{th}} \approx 0.715$. Moreover, 
we define the correlation length $\xi$ by $\chi \sim e^{-L/\xi}$ for $\chi < \chi_{\text{th}}$, measuring how fast $\chi$ approaches zero with respect to increasing length $L$ in the subcritical regime.
Checking the scaling $\xi \sim |\chi - \chi_{\text{th}}|^{-\nu}$ yields an exponent $\nu = 0.02(2)$ [Fig.~\ref{Figure3}(b)], a value that is distinct from $\nu \approx 1.3(3)$ for ConPT~\cite{XJS2021} and is close to zero, which is a clear signature of finite correlation length for discontinuous phase transitions~\cite{zotero-3851}.
}

\subsection{Generalized NegPT}

\yaqi{To examine how the mixed-order phase transition depends on the form of the series/parallel rules in Fig.~\ref{figure1}(c), we consider generalizing these rules. Specifically, we introduce two constant factors, $\eta_s$ and $\eta_p$, modifying the series rule to $\tanh r = \eta_s \tanh r_1 \tanh r_2 \dots$ and the parallel rule to $\sinh r = \eta_p \sinh r_1 \cosh r_2 \cosh r_3 \dots$.
These generalized rules may not be physically valid across the entire parameter range for arbitrary $\eta_s$ and $\eta_p$. 
However, our quantitative analysis is restricted to the region near the critical threshold, where these forms serve as a valid local approximation.}
\yaqi{On the Bethe lattice, we find that the type of the phase transition depends on $\eta_p$. Specifically, when $\eta_p < \sqrt{k-1}$, only a mixed-order phase transition can exist; when $\eta_p \ge \sqrt{k-1}$, only a continuous phase transition can exist (SI, Sec.~II~A).}

What could be the possible physical origins of $\eta_s$ and $\eta_p$? First, the entanglement swapping may take a more general form~\cite{GPOVM2002a,GJ2002,gaussian-povm_fm07}, leading to $\eta_s\le 1$ which reflects the generalization of homodyne measurement (SI, Sec.~II~B). Second, the direct-product construction of the concentration operation (\yaqi{see Methods}) appears to rely solely on the exponential tail of TMSVS in the Fock space, $\sim \chi^n |nn\rangle$.
Therefore, we may consider generic CV states that are non-Gaussian at low photon numbers while retaining the exponential tail at high photon numbers. This could yield $\eta_p\le 1$, reflecting the arbitrariness in the low-photon coefficients (SI, Sec.~II~C). \yaqi{Interestingly, comparing this with the $\eta_p < \sqrt{k-1}$ criterion in the Bethe lattice example suggests that only mixed-order transitions can be observed with generic CV operations as in our settings.}

\begin{figure}
    \centering
    \includegraphics[width=230pt]{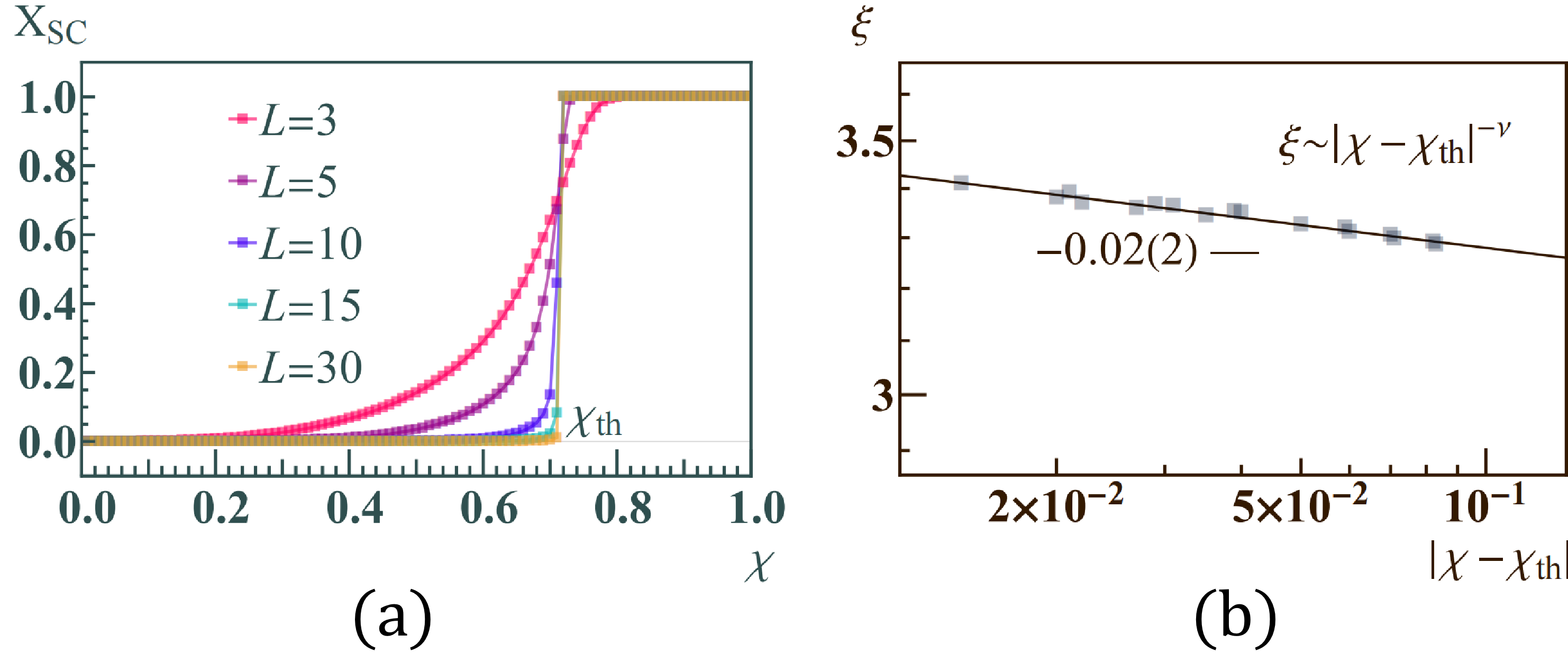}
    \caption{\yaqi{Entanglement percolation in two-dimensional square  lattices.
    (a)~$\text{X}_{\text{SC}}$ for square lattices with different side length $L$.
    (b)~Scaling of the correlation length {$\xi$} near the critical threshold $\chi_{\text{th}}\approx0.715$ follows {$\xi \sim |\chi - \chi_{\text{th}}|^{-\nu}$}, with a fitted critical exponent $\nu = 0.02\pm0.02$.}} 
    \label{Figure3}
\end{figure}

\subsection{\yaqi{Instability under feedback control}}

The identified abrupt transition distinguishes CV-based QNs from DV-based QNs, revealing potential instability inherent in the QN design. 
To illustrate this,
note that in any realistic scenario, quantum states and their entanglement inevitably evolve over time due to environmental influences~\cite{Control_decay_measure1993,control_CV2018,Concurrence_decay_exponential09,Negativity_decay_exponential12,feedback2009,Time_delay_quantum15,Time_delay_quantum16,PID_control1995}. 
As entanglement degradation is often an exponentially fast process~\cite{Concurrence_decay_exponential09,Negativity_decay_exponential12}, we assume the ratio negativity $\chi(t)$ of each link follows an exponential decay, $\chi(t)\sim \exp{\left(-t/\tau\right)}$.
This implies a decay rate proportional to the current negativity, $d \chi(t)/dt\sim-\tau^{-1}\chi(t)$, where the timescale $\tau$ characterizes the system's noise and decoherence.
Although this decay is a local process (as it occurs per link), its collective effect can rapidly compromise global functionality, driving $\mathrm{X}_\text{SC}\to 0$ particularly as $\chi(t)$ breaches $\chi_\text{th}$.

To mitigate such entanglement degradation, a standard approach may utilize measurement-based quantum feedback control~\cite{feedback2009}. This involves constantly monitoring the system's output
$\mathrm{X}_\text{SC}(t)$ and comparing it to a desired target value $\mathrm{X}^\text{target}_\text{SC}\gg 0$. Any difference {$\epsilon(t)=\mathrm{X}^\text{target}_\text{SC} - \mathrm{X}_\text{SC}(t)$} triggers a feedback mechanism to actively adjust system settings, counteracting the entanglement decay and maintaining $\chi(t)$ at operational levels. 
{We further assume that such feedback mechanism contains a time delay due to processing and transmission time across the global scale~\cite{Time_delay_quantum15,Time_delay_quantum16}. Taken all assumptions together, we may approximate the full system dynamics using a first-order-plus-time-delay (FOPTD) model~\cite{PID_control1995}:}
\begin{eqnarray}\label{eq-FOPTD}
    \frac{d \chi(t)}{d t}&=&-\tau^{-1}\chi(t)+u(t-T_0).
\end{eqnarray}
Here, the first term represents decay of the ratio negativity over time and $u(t)$ signifies the feedback action applied after a delay $T_0$.
The feedback signal $u(t)$ is determined by the observed deviation $\epsilon(t)$
(SI, Sec.~III~A). In a physical system, for instance, $u(t)$ might correspond to increasing the pump laser power for down-conversion sources, boosting the squeezing as well as entanglement to compensate for ongoing losses~\cite{down_conversion2005,down_conversion2010,down_conversion2017,control_CV2018}.

\begin{figure}[t!]
    \centering
     \includegraphics[width=250pt]{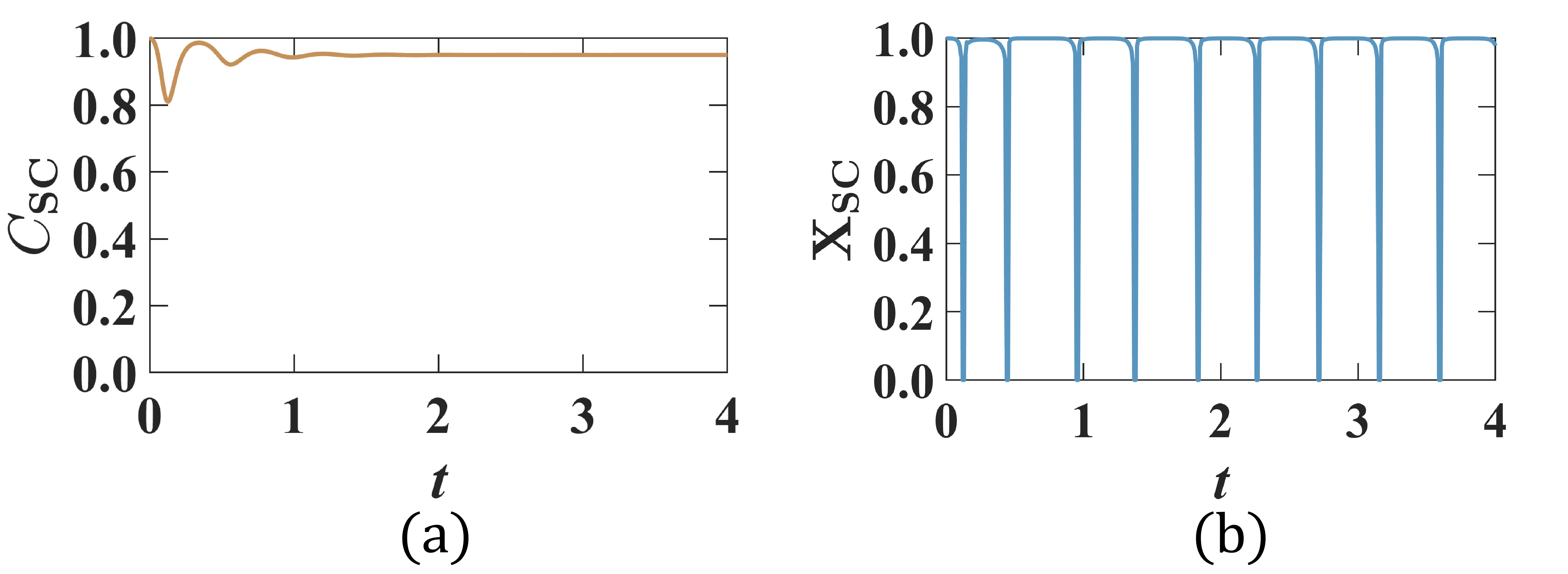}
    \caption{Feedback stabilization of QN  against entanglement decay. 
    (a)~Under the same feedback control [Eq.~\eqref{eq-FOPTD}], 
    the DV-based QN ($k=3$ Bethe lattice) exhibits rapid stabilization; (b)~whereas the CV-based QN exhibits long-term ``on/off''
    instability, a direct result of the abrupt drop in Fig.~\ref{Figure2}(b).
    }\label{Figure4}
\end{figure}

Our aim is to understand how the sharp, discontinuous nature of $\mathrm{X}_\text{SC}$ vs.~$\chi$ [Fig.~\ref{Figure2}(b)] undermines the feedback efficacy.
Simulating an FOPTD model with identical $\tau$, $T_0$, and controller $u(t)$ on both DV- and CV-based Bethe-lattice QNs reveals starkly different behaviors. 
Under the concurrence percolation~\cite{XJS2021,OXSGBJ2022}, 
the sponge-crossing concurrence $C_{\text{SC}}(t)$ for DV-based states recovers smoothly under feedback [Fig~\ref{Figure4}(a)]. \yaqi{In contrast, the sponge-crossing ratio negativity $\mathrm{X}_{\text{SC}}(t)$ for CV-based states, dictated by its mixed-order transition, exhibits unstable ``on/off'' oscillations using the same $u(t)$ [Fig.~\ref{Figure4}(b)].}

\yaqi{Crucially, this instability is {not} a result of the limitation of the FOPTD model's capability. Indeed, under Eq.~\eqref{eq-FOPTD}, $\chi(t)$ is successfully stabilized within a narrow range and only undergoes small perturbations near the critical threshold $\chi_{\text{th}}\approx 0.866$  (SI, Sec.~III~B).
Therefore, 
the instability of $\mathrm{X}_{\text{SC}}$ comes directly from 
the sharp nature of $\mathrm{X}_\text{SC}$ vs.~$\chi$ [Fig.~\ref{Figure2}(b)]. We also find that the instability persists across a wide range settings for $u(t)$ (SI, Sec.~III~C),
indicating that the CV system stabilization is substantially more demanding than for DV counterparts,
calling for a more careful feedback design for CV-based QNs.}

\section*{Discussion}

{Despite the rich studies on mixed-order phase  transitions~\cite{mix-order-phase-transit_m23,Interdependent_net2014,interdependent_bglvfh23,mix-order-phase-transit_pbplrbh24,k_core2024,mix-order-phase-transit_atc17}, the underlying mechanisms remain not fully understood.} The traditional perspective of spinodal points~\cite{spinodal_ku83} seems not applicable here, since our negativity percolation model---based solely on series and parallel rules---lacks a manifestation of {hysteresis} or metastability 
(which is also commonly absent in interdependent percolation~\cite{gao2011robustness,interdependent_network_2015}).
Consequently, the emergence of singularity may not necessarily be described by a partition function as for spinodal points~\cite{spinodal_ku83}.
{This heralds future theoretical investigations  into how the interplay of series and parallel rules can induce mixed-order transitions. Such studies could unveil novel phase transition landscapes potentially observable in quantum optical experiments implementing our G-G DET scheme.}

For practical QN design, several unresolved questions remain. For instance, unlike in DV systems~\cite{XYJSA2023}, our CV parallel rule (entanglement concentration from Ref.~\cite{MRN}) is not optimal. Understanding the feasibility and implementation limits of entanglement concentration has been a focus for CV system designs. This understanding is becoming even more urgent as we approach the deployment of larger-scale, realistic QNs.

\section*{Methods}

This section presents the foundational tools for studying negativity percolation on Gaussian quantum networks. Further details are provided in the Supplementary Information.

\setcounter{subsection}{0}

\subsection{G-G entanglement swapping}\label{sec-GGswapp}

{We employ Gaussian-to-Gaussian entanglement swapping for pure Gaussian states as proposed in Ref.~\cite{P2002}. 
Consider the simplest case of three nodes ($S$, $R$, and $T$) sharing four modes, labeled as 1 through 4, corresponding to one mode on $S$, two modes on $R$, and one mode on $T$. 
When a TMSVS with squeezing parameters $r_1$ is shared by $S$ and $R$ and another TMSVS with squeezing parameter $r_2$ is shared by $R$ and $T$, the CV Bell measurement is performed on modes $2$ and $3$ by detecting the output modes of a $50:50$ beam splitter with these two modes as the input modes. The quadratures $\hat{x}_u=(\hat{x}_2-\hat{x}_3)/\sqrt{2}$ and $\hat{p}_v=(\hat{p}_2+\hat{p}_3)/\sqrt{2}$ are measured, where $\hat{x}_m$ and $\hat{p}_m$ denote the ``position'' and ``momentum'' operators of mode $m$ ($m=2,3$).
Then, the quadratures for modes $1$ and $4$ undergo displacements, $\hat{x}_m'=\hat{x}_m-g_m\hat{x}_u/\sqrt{2}$ and $\hat{p}_m'=\hat{p}_m+g_m\hat{p}_v/\sqrt{2}$ ($m=1,4$). With the right choice of the gains,
\begin{eqnarray}
    g_1=\frac{2\sinh 2r_1}{\cosh 2r_1+\cosh 2r_2},\\
    g_4=\frac{2\sinh 2r_2}{\cosh 2r_1+\cosh 2r_2},
\end{eqnarray}
the output is a new TMSVS with  squeezing parameter $r$ that satisfies
\begin{eqnarray}
    \tanh r=\tanh r_1\tanh r_2.
\end{eqnarray}}
This protocol can be extended to a one-dimensional chain, consisting of nodes $S$, $T$, and $N-1$ intermediate nodes $R_{1},\dots,R_{N-1}$ (Fig.~\ref{Swapping-n}), leading to the following theorem:
\begin{theorem}
     Given $N$ TMSVSs $|\psi^{r_{1}}\rangle$, $|\psi^{r_{2}}\rangle$, ..., $|\psi^{r_{N}}\rangle$, where $|\psi^{r_{1}}\rangle$ is shared by two parties $S$ and R$_1$, $|\psi^{r_{2}}\rangle$ is shared by R$_1$ and R$_2$, ..., and $|\psi^{r_{N}}\rangle$ is shared by R$_{N-1}$ and $T$, the TMSVSs can be  deterministically converted to a single TMSVS $|\psi^{r}\rangle$ between $S$ and $T$ by LOCC, where
\begin{eqnarray}
   \label{1D}
     \tanh{r}=\prod_{j=1}^{N}\tanh{r_{j}}.
  \end{eqnarray}
  \end{theorem}
Notably, the final result is
independent of the order of execution. This property, as we will see, does not apply to the concentration protocol.

We may also express Eq.~\eqref{1D} in terms of the ratio negativity. Let $ \chi_{j}\in [0,1]$ be the ratio negativity of each input TMSVS $|\psi^{r_j}\rangle$, and let $\chi\in [0,1]$ be the ratio negativity of the final output TMSVS $|\psi^{r}\rangle$. We define the series function
  \begin{eqnarray}
  \label{eq_series}
    \text{seri}(\chi_1,\chi_2,...):=\chi=\prod^{N}_{j=1}\chi_{j},
  \end{eqnarray}
which corresponds to the series rule for the G-G DET.

\begin{figure}[t!]
  \centering
  \includegraphics[width=3in]{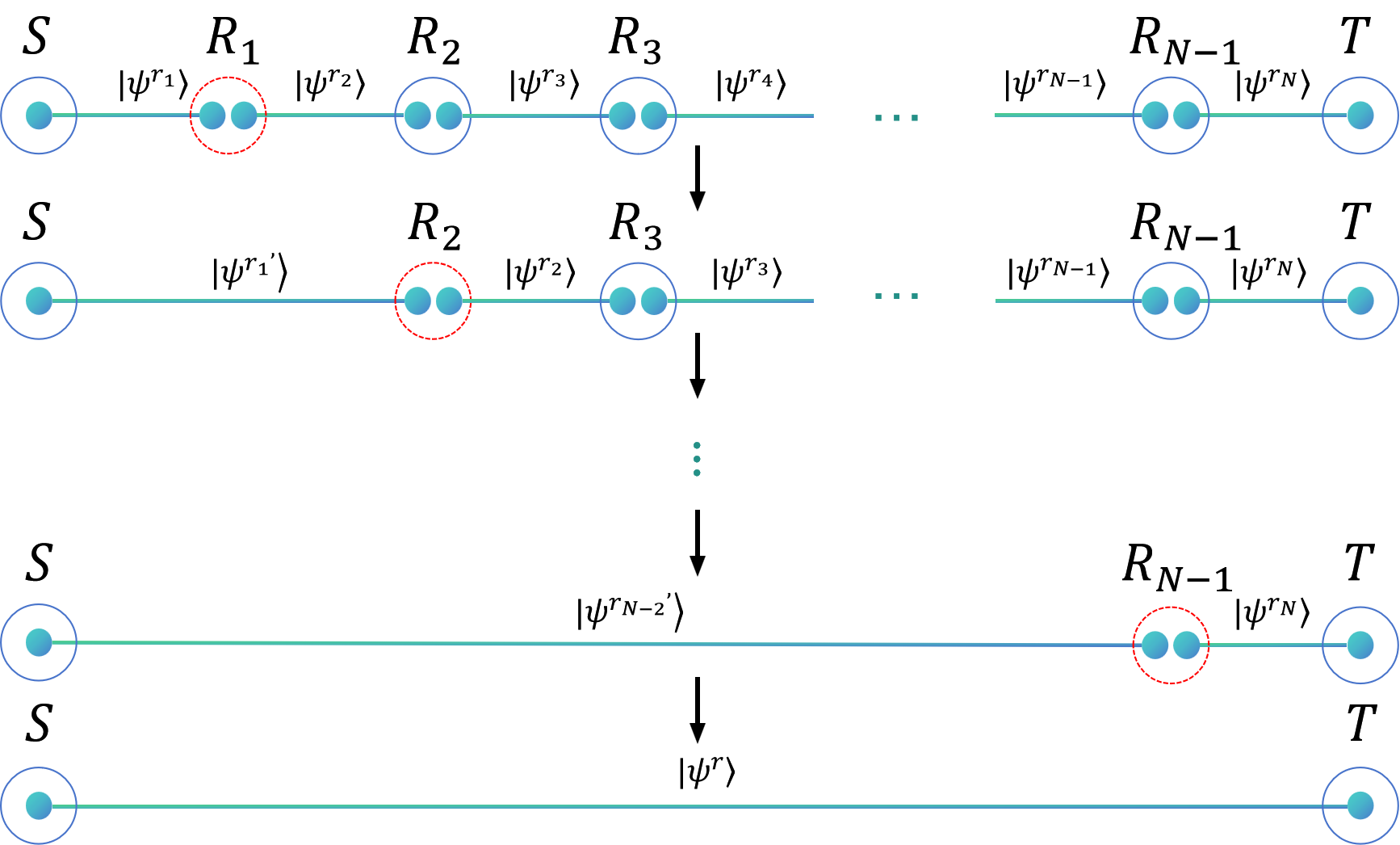}\\

   \caption{Continuous-variable entanglement swapping. 
   Given $N$ TMSVSs $|\psi^{r_{1}}\rangle,|\psi^{r_{2}}\rangle,\dots,|\psi^{r_{N}}\rangle$, the final state $|\psi^{r}\rangle$ between $S$ and $T$ derived by entanglement swapping is a TMSVS, where the squeezing parameter $r$ satisfies Eq.~\eqref{1D}.}   \label{Swapping-n}
\end{figure}

\subsection{G-G entanglement concentration}\label{sec-GGconcen}

In DV systems, the process of entanglement concentration enables the deterministic conversion of two (or more) bipartite pure states into a single bipartite pure state with enhanced entanglement through LOCC~\cite{CHSB1996}. We aim to do the same for TMSVSs, deterministically concentrating two (or more) TMSVSs into another one with higher entanglement (larger squeezing parameter). We begin with some preliminaries on deterministic LOCC:

We now introduce the majorization theory between infinite-dimensional vectors. Let $\vec{x}$ and $ \vec{y}$ be positive (each element is a non-negative real number), infinite-dimensional vectors. $\vec{x}$ is majorized by $\vec{y}$, denoted as $\vec{x}\prec\vec{y}$, iff
 \begin{eqnarray}
   \Sigma_{j=1}^{k}x_{j}^{\downarrow}\leq\Sigma_{j=1}^{k}y_{j}^{\downarrow}
  \end{eqnarray}
holds for all positive integers $k$ and the equality holds for $k=\infty$.
Here, $x_j^{\downarrow}$ ($y_{j}^{\downarrow}$) is the $j$th element of vector $\vec{x}$ ($\vec{y}$) arranged in descending order. Majorization theory implies that a \yaqi{bipartite pure} state $|\psi\rangle$ can be deterministically transformed into another state $|\psi'\rangle$  using LOCC iff $\vec{\lambda}\prec\vec{\lambda}'$, {where $\vec{\lambda}$ and $\vec{\lambda'}$ are the Schmidt-value vectors of $|\psi\rangle$
 and $|\psi'\rangle$, with $\vec{\lambda}=(\lambda_1,\lambda_2,\cdots)^{\mathrm{T}}$ and $\vec{\lambda'}=(\lambda_1',\lambda_2',\cdots)^{\mathrm{T}}$}~\cite{Nielsen1999} where $\sum_{j}\lambda_j=\sum_{j}\lambda_j'=1$.
For a TMSVS $|\psi^r\rangle$, its Schmidt-value vector $\vec{\lambda}(r)$ has entries $\lambda_n=\left(1-\tanh^2 r\right) \tanh^{2n} (r)$. 
The  majorization relation $\vec{x}\prec\vec{y}$ holds iff there exists a column-stochastic matrix $\mathbf{D}$ (whose column sum $=1$ for each column, and row sum \yaqi{$\leq 1$} for each row), such that $\vec{x}=\mathbf{D} \vec{y}$.
Therefore, to construct deterministic LOCC for two TMSVSs with Schmidt-value vectors $\vec{x}$ and $\vec{y}$, respectively, conceptually one only needs to find a column-stochastic matrix $\mathbf{D}$ that satisfies $\vec{x}=\mathbf{D} \vec{y}$.

Here, we adopt the exact construction of $\mathbf{D}$ given in Ref.~\cite{MRN}, where  $\mathbf{D}$ is constructed as a tensor product: $\mathbf{D}=\mathbf{D}^{\mathcal{L}}\otimes\mathbf{D}^{\mathcal{A}}$. The matrix $\mathbf{D}^{\mathcal{L}}$ is given by
\begin{equation}\label{eq-DL}
    \mathbf{D}^{\mathcal{L}}=
    \begin{bmatrix}
1      & 1-\eta          & (1-\eta)^2              & (1-\eta)^3                & \dots \\
0      & {1\choose1}\eta & {2\choose1}\eta(1-\eta) & {3\choose1}\eta(1-\eta)^2 & \dots \\
0      & 0               & {2\choose2}\eta^2       & {3\choose2}\eta^2(1-\eta) & \dots \\
0      & 0               & 0                       & {3\choose3}\eta^3         & \dots \\
\vdots & \vdots          & \vdots                  & \vdots                    &\ddots
   \end{bmatrix}.
\end{equation}
Similarly, the matrix $\mathbf{D}^{\mathcal{A}}$ is given by
\begin{eqnarray}\label{eq-DA}
    \mathbf{D}^{\mathcal{A}}=
    \begin{bmatrix}
a      & 0                  & 0                & 0       & \dots \\
ab     & a^2                & 0                & 0       & \dots \\
a b^2  & {2\choose1}a^2 b   & a^3              & 0       & \dots \\
a b^3  & {3\choose1}a^2 b^2 & {3\choose2}a^3 b & a^4     & \dots \\
\vdots & \vdots             & \vdots           & \vdots  &\ddots
   \end{bmatrix},
\end{eqnarray}
where $G>0$, $a=1/G$ and $b=1-a$.

Note that the forms of these matrices are, indeed, motivated by the operation of a pure-loss channel with transmissivity $\eta\in(0,1)$ and a quantum-limited amplifier with a gain $G>1$ for thermal lights, respectively. This physical analogy may provide some insights on how to construct the exact LOCC using non-standard optical components.

The introduction of these two matrices is particularly useful for TMSVS, since their tensor product $\mathbf{D}^{\mathcal{L}}\otimes\mathbf{D}^{\mathcal{A}}$ can transform the Schmidt-value vector of the tensor product of two TMSVSs, $\vec{\lambda}(r_1')\otimes\vec{\lambda}(r_2')$, into another vector of the same form $\vec{\lambda}(r_1)\otimes\vec{\lambda}(r_2)$, where $\tanh r_1'=\tanh r_1/\sqrt{\eta+(1-\eta)\tanh^2 r_1}$ and $\tanh r_2'=\sqrt{1-G(1-\tanh^2 r_2)}$.
The majorization relation $\vec{\lambda}(r_1)\otimes\vec{\lambda}(r_2)\prec\vec{\lambda}(r_1')\otimes\vec{\lambda}(r_2')$ holds if $\mathbf{D}^{\mathcal{L}}\otimes\mathbf{D}^{\mathcal{A}}$ is column-stochastic, which is true iff $\eta G\geq1$. 
This gives rise to the following Lemma:
\begin{lemma}
\label{Lemma2}(Theorem 1 in Ref.~\cite{MRN})
Let $|\Psi\rangle$, $|\Psi'\rangle$ be two $(2+2)$-mode pure Gaussian states, with decreasingly ordered squeezing vectors $\vec{r}^{\downarrow}=(r_{1},r_{2})$ and $\vec{r'}^{\downarrow}=(r_{1}',r_{2}')$ such that the two components of vector $(r_{1}'-r_{1},r_{2}'-r_{2})$ have opposite signs. Then $|\Psi\rangle$ can be transformed into $|\Psi'\rangle$ using deterministic non-Gaussian LOCC if
    \begin{eqnarray}
     \label{distill2}
     \frac{\sinh(r_{1}+ r_{2})\pm\sinh(r_{1}-r_{2})}
       {\sinh(r_{1}'+r_{2}')\pm\sinh(r_{1}'-r_{2}')}\geq1,
    \end{eqnarray}
  where $\pm$ follows the sign of $r_{1}'-r_{1}$.
 \end{lemma}
 
Lemma~\ref{Lemma2} provides a sufficient condition to convert two TMSVSs $|\psi^{r_1}\rangle$ and $|\psi^{r_2}\rangle$ into another one with squeezing parameter $r_1'>r_1$ for the case of $r_2'=0$. To be exact, the maximum value $r$ of $r_{1}'$ that satisfies Eq.~\eqref{distill2} is given by $r=\sinh^{\mbox{-}1}(\sinh{r_{1}}\cosh{r_{2}})>r_{1}$ and $r_2'=0$ \yaqi{with $G=\eta^{-1}=(1-\tanh r_2)^{-1}$}, implying that entanglement is concentrated.
Thus, consider a $(2+2)$-mode pure Gaussian state $|\psi^{r_{1}}\rangle\otimes |\psi^{r_{2}}\rangle$ with decreasingly ordered squeezing vectors $\vec{r}^{\downarrow}=(r_{1},r_{2})$. This state can be deterministically converted into a single TMSVS $|\psi^{r}\rangle$ by non-Gaussian LOCC, where
     \begin{eqnarray}
      \label{(2+2)-CONCEN}
       \sinh{r}=\sinh{r_{1}}\cosh{r_{2}}.
     \end{eqnarray}
\yaqi{That is, according to Ref.~\cite{MRN}, the bipartite state $|\psi^{r_{1}}\rangle\otimes |\psi^{r_{2}}\rangle$ can be converted into $|\psi^{r}\rangle$ ($r_1<r$) via deterministic LOCC when the reduced state of $|\psi^r\rangle$ can be transformed into the reduced state of $|\psi^{r_1}\rangle$ through a pure loss channel with intensity transmittance $\eta=1-\tanh^2 r_2$ ($r_2\leq r_1$).}

Equation~\eqref{(2+2)-CONCEN} is the building block for scaling the concentration protocol.  Let $|\Psi\rangle$ be a pure Gaussian state and, w.l.o.g.,~write $|\Psi\rangle=\mathop{\bigotimes}\nolimits_{k=1}^{K}|\psi^{r_{k}}\rangle$ $(r_1\geq r_2\geq\cdots\geq r_K>0)$ with squeezing vector $\vec{r}^{\downarrow}=\left(r_1,r_2,\dots,r_K\right)$ by Ref.~\cite{Modewise2003}. 
By iteratively applying Eq.~\eqref{(2+2)-CONCEN}  $K-1$ times on $|\Psi\rangle$ (Fig.~\ref{FIG-concen1}), the pure Gaussian state $|\Psi\rangle$ is deterministically concentrated into a single TMSVS.
Only LOCC is used throughout the entire process, giving rise to the following theorem:
\begin{theorem}\label{N+M-CONCEN} A pure Gaussian state $|\Psi\rangle_{AB}=\mathop{\bigotimes}\nolimits_{k=1}^{K}|\psi^{r_{k}}\rangle_{A_{k}B_{k}}$ with $K$-dimensional squeezing parameter vector $\vec{r}^{\downarrow}=(r_{1},\dots,r_{K})$ (where \yaqi{$r_{1}\geq \dots\geq r_{K}>0$})
can be deterministically converted to a TMSVS $|\psi^{r}\rangle$ by non-Gaussian LOCC, where {$r$ satisfies the series rule}
    \begin{eqnarray}
     \label{(N+M)-CONCEN}
     \sinh{r}=\sinh{r_{1}}\prod_{k=2}^{K}\cosh{r_{k}}.
    \end{eqnarray}
\end{theorem}

\begin{proof}
We regard $|\psi\rangle$ as the tensor $\otimes_{k=1}^{K}|\psi^{r_{k}}\rangle$ of TMSVSs.
Firstly, state $|\psi^{r_{1}'}\rangle$ with
  \begin{eqnarray*}
    \sinh{r_{1}'}=\sinh{r_{1}}\cosh{r_{2}}>\sinh{r_{1}}.
  \end{eqnarray*}
can be generated from $|\psi^{r_{1}}\rangle\otimes|\psi^{r_{2}}\rangle$
by entanglement concentration.
Next, executing the scheme again on state $|\psi^{r_{1}'}\rangle\otimes|\psi^{r_{3}}\rangle$, we get $|\psi^{r_{2}'}\rangle$ for
  \begin{eqnarray*}
    \sinh{r_{2}'}=\sinh{r_{1}}\prod_{k=2}^{3}\cosh{r_{k}}.
  \end{eqnarray*}
Similarly, applying the scheme iteratively, we derive a state $|\psi^{r}\rangle$ with
  \begin{eqnarray*}
    \sinh{r}=\sinh{r_{1}}\prod_{k=2}^{K}\cosh{r_{k}}.
  \end{eqnarray*}
\end{proof}

\begin{figure}[t!]
    \centering
    \includegraphics[width=3.3in]{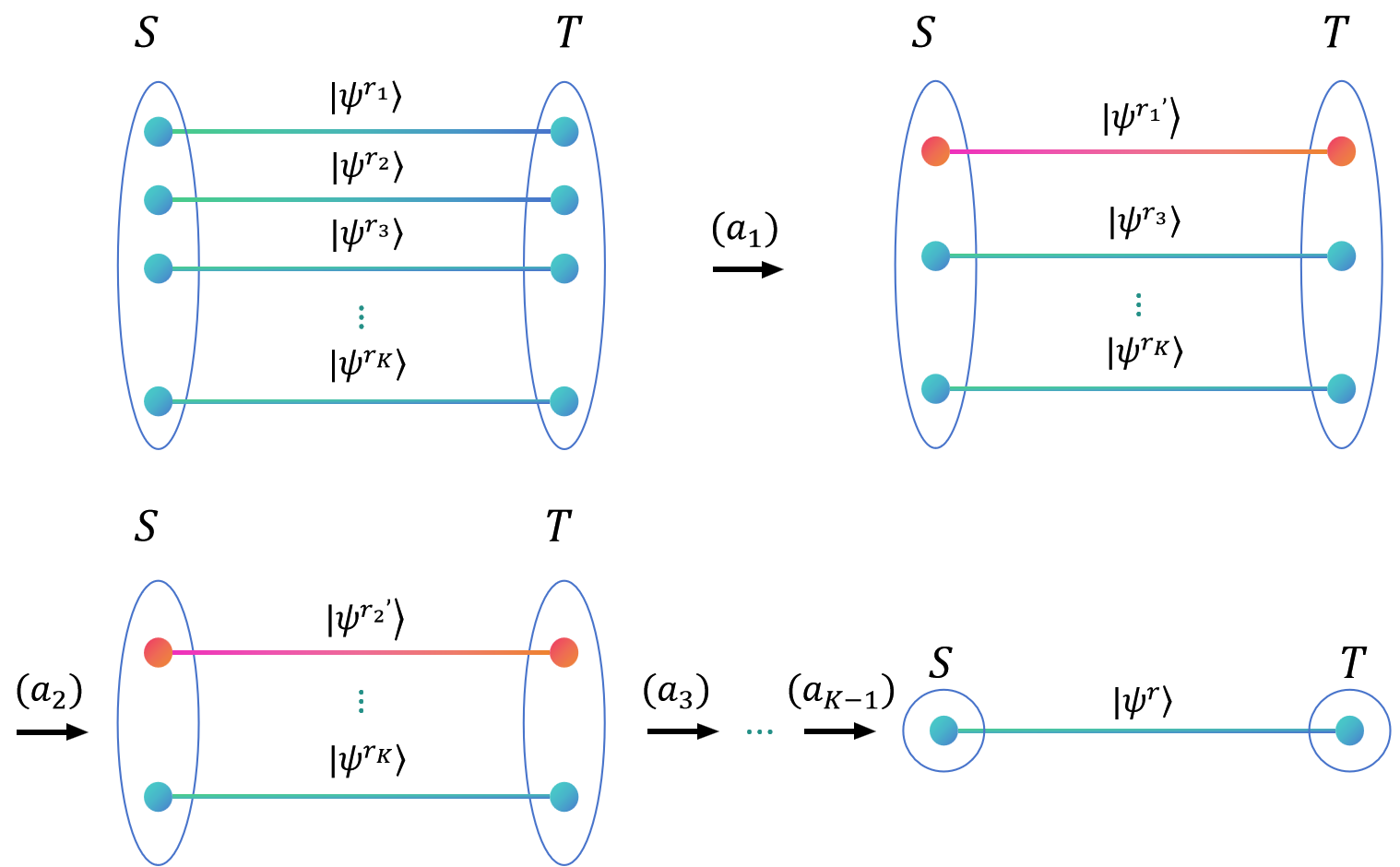}\\
     \caption{Continuous-variable entanglement concentration.
     Consider $K$ parallel TMSVSs written as the tensor product: $\bigotimes\nolimits_{k=1}^{K}|\psi^{r_{k}}\rangle$.
     Step $a_1$: The TMSVS $|\psi^{r_{1}'}\rangle$ is obtained by executing the TMSVS entanglement concentration protocol on  $|\psi^{r_{1}}\rangle\otimes|\psi^{r_{2}}\rangle$;
     Step $a_2$: Performing the scheme again on $|\psi^{r_{1}'}\rangle\otimes|\psi^{r_{3}}\rangle$ yields the TMSVS $|\psi^{r_{2}'}\rangle$; and so on.
     This eventually results in a single TMSVS $|\psi^{r}\rangle$, where the squeezing parameter $r$ satisfies Eq.~\eqref{(N+M)-CONCEN}.}\label{FIG-concen1}
  \end{figure}

We now illustrate the optimal concentration order. In general, depending on which two modes are first concentrated, the prefactor in Eq.~\eqref{(N+M)-CONCEN} could be $\sinh{r_{k}}$ for any $k$ other than $\sinh{r_{1}}$. However, the optimal entanglement of $|\psi^{r}\rangle$ is obtained by first concentrating the mode of the highest entanglement (largest squeezing parameter, $r_1$):

{\begin{theorem}
The squeezing parameter $r$ of the final TMSVS $|\psi^{r}\rangle$ is maximized if and only if the TMSVS with the largest squeezing parameter ($\max\limits_{1\leq k\leq K}r_k$) is among the first two TMSVSs to be concentrated. The corresponding optimal $r$ is given by $\sinh r=(\tanh{\max\limits_{1\leq k\leq K}r_k)} \left(\prod_{k=1}^{K}\cosh r_{k}\right)$.
\end{theorem}} 
\begin{proof} Assume w.l.o.g.~$r_1\ge r_2\ge\dots\ge r_K$.  We first consider the simplest nontrivial case $K=3$ and then generalize to $K>3$:

{Case~1 ($K=3$).} If $|\psi^{r_2}\rangle$ and $|\psi^{r_3}\rangle$ are concentrated first (without $|\psi^{r_1}\rangle$), the final squeezing parameter $r$ satisfies
  \begin{eqnarray*}
    \sinh r=\max\{\cosh r_{1}\sinh r',\ \sinh r_{1}\cosh r'\},
  \end{eqnarray*}
where $r'=\text{arcsinh}(\sinh r_{2}\cosh r_{3})$.
Since $r_{1}\geq r_{k}$ for $k=2,3$, we have $\sin r_{1}\cosh r_{2}\cosh r_{3}\geq\cosh r_{1}\sinh r'$.
Furthermore, $\sin r_{1}\cosh r_{2}\cosh r_{3}\geq\sinh r_{1}\cosh r'$ holds, due to
  \begin{eqnarray*}
    \cosh r'
    &=&\sqrt{1+\sinh^{2}r'}\\
    &=&\sqrt{1+\sinh^{2}r_{2}\cosh^{2}r_{3}}\\
    &=&\sqrt{\cosh^{2}r_{2}\cosh^{2}r_{3}-(\cosh^{2}r_{3}-1)}\\
    &\leq&\cosh r_{2}\cosh r_{3}.
  \end{eqnarray*}
Thus, $r$ would be greater if $|\psi^{r_1}\rangle$ was concentrated first.

{Case~2 ($K>3$).}
Assume that this {theorem} holds true for $K=n\ (n\geq3)$.
Then, for $K=n+1$, $|\psi^{r_1}\rangle$ must either be the first to concentrate (since it is the largest among the first $n$ TMSVSs to be concentrated), or the last to concentrate.
In the latter case, $|\psi^{r_2}\rangle$ (with the second largest $r_2$), which is now the largest among the first $n$ TMSVSs, must be the first to be concentrated. This gives rise to
  \begin{eqnarray*}
    \sinh r=\max\{\cosh r_{1}\sinh r'',\ \sinh r_{1}\cosh r''\},
  \end{eqnarray*}
where $r''=\text{arcsinh}\left(\sinh r_{2}\prod_{k=3}^{n+1}\cosh r_{k}\right)$.
Because $r_{1}\geq r_{2}$ and
  \begin{eqnarray*}
    \cosh r''
    &=&\sqrt{1+\sinh^{2}r''}\\
    &=&\sqrt{1+\sinh^{2}r_2\prod_{k=3}^{n+1}\cosh^{2}r_k}\\
    &=&\sqrt{\prod_{k=2}^{n+1}\cosh^{2} r_{k}-(\prod_{k=3}^{n+1}\cosh^{2} r_{k}-1)}\\
    &\leq&\prod_{k=2}^{n+1}\cosh r_{k},
  \end{eqnarray*}
we have $\sinh r_{1}\prod_{k=2}^{n+1}\cosh r_{k}\geq\cosh r_{1}\sinh r''$ and $\sinh r_{1}\prod_{k=2}^{n+1}\cosh r_{k}\geq\sinh r_{1}\cosh r''$, respectively.
It is therefore optimal to concentrate $|\psi^{r_1}\rangle$ first.
\end{proof}

To express Eq.~\eqref{(N+M)-CONCEN} in terms of entanglement measure, let $\chi_{k}$ and $\chi$ be the ratio negativity of input TMSVS $|\psi^{r_k}\rangle$ and  the output TMSVS $|\psi^{r}\rangle$, respectively. We have:
  \begin{eqnarray}
  \label{G-G concen}
    \frac{\chi^{2}}{1-\chi^{2}}=
    \frac{\max\limits_{1\leq k\leq K}\chi_{k}^{2}}
         {\prod^{K}_{k=1}\left(1-\chi_{k}^{2}\right)},
  \end{eqnarray}
{yielding a parallel function
\begin{eqnarray}\label{(N+M)-CONCEN-v}
    \text{para}(\chi_1,\chi_2,...):=\chi=\frac{\max\limits_{1\leq k\leq K}\chi_k}{\sqrt{\max\limits_{1\leq k\leq K}\chi_k^2+\prod_{k=1}^K(1-\chi_k^2)}}.\nonumber\\
\end{eqnarray}}

\subsection{Negativity Percolation in Non-Series-Parallel QNs}\label{sec-bridge}

For non-series-parallel QNs, the negativity percolation theory (NegPT) based on the ratio negativity is defined using approximate star-mesh transforms (more details are provided in Ref.~\cite{XJS2021}).
A star-mesh transform~\cite{star-mesh_v70} can establish a local equivalence of the connectivity between an $(s+1)$-node star graph and an $s$-node complete graph, based only on series and parallel rules. Note that the star-mesh transform is generally irreversible. Specifically, while a star graph with arbitrary weights can most likely  be transformed into a complete graph with matching weights, the inverse transformation from a complete graph with arbitrary weights back to a star graph is usually not possible. An exception to this is when $s=3$, in which case the star-mesh transform is reversible and is known as the ``Y-$\Delta$/$\Delta$-Y'' transforms.

In the following, we examine several non-series-parallel examples (``bridge'' topologies and lattices) and explain how to calculate the sponge-crossing ratio negativity $\mathrm{X}_\text{SC}$ using the star-mesh transform.

  \begin{figure}[t!]
    \centering
     \includegraphics[width=230pt]{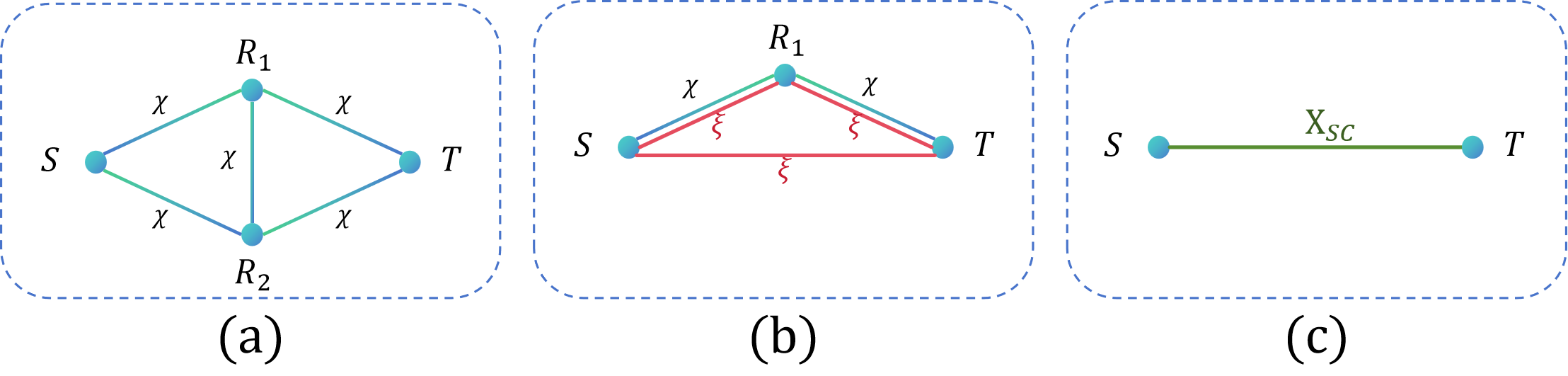}
    \caption{Wheatstone bridge.~(a)~The original QN.~(c)~The QN after the ``Y-$\Delta$'' transform.~(c)~The final QN.}\label{Figure7}
  \end{figure}

From Eqs.~\eqref{eq_series}~and~\eqref{(N+M)-CONCEN-v}, we have
\begin{eqnarray*}
    \text{seri}(\chi_1,\chi_2)&=&\chi_1\chi_2
\end{eqnarray*}
and
\begin{eqnarray*}
    \text{para}(\chi_1,\chi_2)&=&\frac{\max\{\chi_1,\chi_2\}}{\sqrt{\max\{\chi_1^2,\chi_2^2\}+(1-\chi_1^2)(1-\chi_2^2)}}.
\end{eqnarray*}
The simplest non-series-parallel network is a 4-node network: a Wheatstone bridge, as shown in Fig.~\ref{Figure7}(a).
Let each link in the QN be a TMSVS with identical ratio negativity $\chi$. We apply the star-mesh transform to $R_2$, namely the ``Y-$\Delta$'' transform; then, the QN becomes a series-parallel graph shown in Fig.~\ref{Figure7}(b), with three new TMSVSs of ratio negativity $\xi$ shared between $S$ and $R_1$, $R_1$ and $T$, and $S$ and $T$, respectively. Here, $\xi$ satisfies (by the star-mesh transform):
 \begin{eqnarray}
  \mathrm{seri}(\chi,\chi)=\mathrm{para}(\xi,\xi^2),
 \end{eqnarray}
 namely,
 \begin{eqnarray}\label{solve}
  \chi^2=\frac{\xi}{\sqrt{\xi^6-\xi^4+1}}.
 \end{eqnarray}
The solution of $\xi$ is given by $\xi=y(\chi)$, where
  \begin{eqnarray}
    y(x)=\sqrt{\frac{1+2\sqrt{1+3x^{-4}}\cos{\frac{\theta+\pi}{3}}}{3}},
  \end{eqnarray}
and $\theta=\arccos{\left[2^{-1}(-9x^{-4}+25)(1+3x^{-4})^{-\frac{3}{2}}\right]}$.
Specifically, let $W=\chi^4$ and $Y=\xi^2$. Then, Eq.~\eqref{solve} is equivalent to 
  \begin{eqnarray}
    Y^3-Y^2-W^{-1}Y+1=0,
  \end{eqnarray}
where $0\leq Y\leq W\leq1$.
The  Shengjin’s formula~\cite{Shengjin1989} implies that the roots of this cubic polynomial are
  \begin{equation}
  \label{SJ}
    \left\{
      \begin{array}{l}
         Y_1={\left[1-2\sqrt{1+3W^{-1}}\cos{\frac{\theta}{3}}\right]}/{3},\\
         Y_{2,3}={\left[1+2\sqrt{1+3W^{-1}}\left(\cos{\frac{\theta\mp\pi}{3}}\right)\right]}/{3},\\
      \end{array}
    \right.
  \end{equation}
where $\theta=\arccos\left[2^{-1}(-9W^{-1}+25)(1+3W^{-1})^{-\frac{3}{2}}\right]$.
Only the third root $Y_3$ meets the numerical requirements $0\leq Y\leq W\leq1$. 
Therefore, we have $\xi=\sqrt{Y_3}=y(\chi)$.

Solving Fig.~\ref{Figure7}(b), the sponge-crossing ratio negativity $\mathrm{X}_\text{SC}$ is the ratio negativity of the final TMSVS between node $S$ and $T$ [Fig.~\ref{Figure7}(c)], which is given by $\mathrm{X}_\text{SC}=b(\chi,y(\chi))$ with 
  \begin{eqnarray}
    b(x,y)=\frac{x^2}{\sqrt{x^4+(1-y^2)\left[\left(x^2y^2-y^2+1\right)^2-x^4\right]}}.\nonumber\\
  \end{eqnarray}

 \begin{figure}[t!]
    \centering
     \includegraphics[width=230pt]{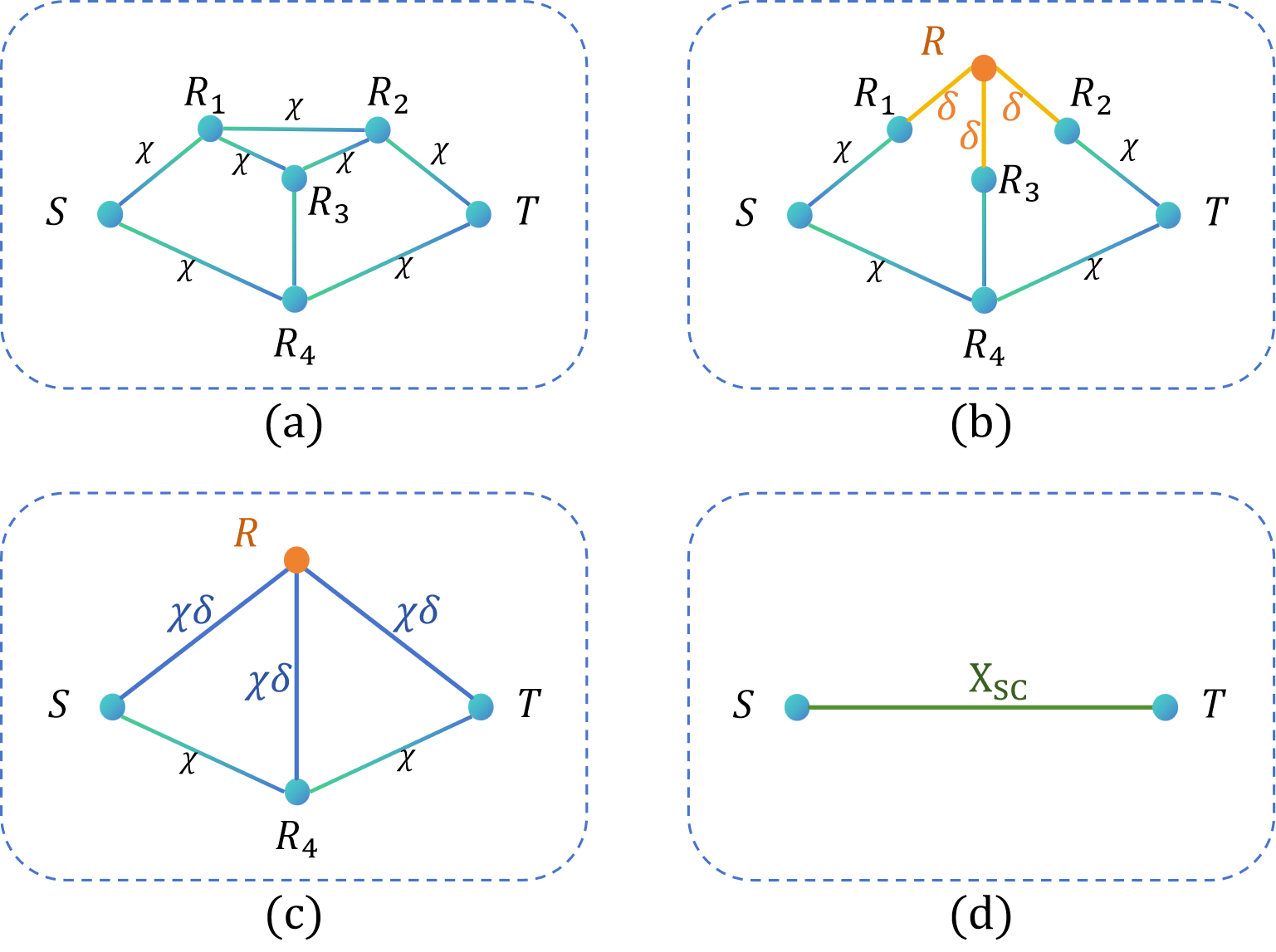}
    \caption{Kelvin bridge.~(a)~The original QN.~(b)~The QN after the ``$\Delta$-Y'' transform.~(c)~The QN after applying series rules to $R_1$, $R_2$ and $R_3$.~(d)~The final QN. }\label{Figure8}
  \end{figure}

A Kelvin bridge, the second simplest non-series-parallel network with 5 nodes  [Fig.~\ref{Figure8}(a)], can be viewed as the QN after performing the star-mesh transform on node $R$ in the QN shown in Fig.~\ref{Figure8}(b), where $\delta:=y^{-1}(\chi)$ is the inverse function of $y(\chi)$. This inverse progress from Fig.~\ref{Figure8}(a) to Fig.~\ref{Figure8}(b) is thus known as the ``$\Delta$-Y'' transform. Applying the series rule on nodes $R_1$, $R_2$ and $R_3$, we obtain Fig.~\ref{Figure8}(c), which is the Wheatstone bridge QN as discussed above. Consequently,  $\mathrm{X}_\text{SC}$ for the Kelvin bridge [Fig.~\ref{Figure8}(d)]  is given by
  \begin{eqnarray}
    \mathrm{X}_\text{SC}=b\left(\chi,y\left(\chi\delta(\chi)\right)\right).
  \end{eqnarray}

\section*{Data availability}
Data generated and analyzed in this study are available from the corresponding authors upon request.

\section*{Code availability}
Code used to generate data in this study are available from the corresponding authors upon request.

\section*{References}

\bibliography{NegPT}

\section*{Legends of Figures}

Figure 1:
Gaussian-to-Gaussian deterministic entanglement transmission (G-G DET) scheme. Applicable to Gaussian quantum networks (QN), the scheme consists of two LOCC protocols:  
    (a)~Entanglement swapping, facilitated by homodyne detection and displacement;
    (b)~Entanglement concentration, facilitated by non-standard optical components. Both protocols are deterministic, taking two (or more) TMSVS $|\psi^{r_{i}}\rangle$ as input and a new TMSVS $|\psi^{r}\rangle$ as output.
    (c)~The two LOCC protocols map to series and parallel rules, respectively, to construct G-G DET.
    (d)~Consider a QN example built upon three node. The G-G DET scheme consists of two steps:
   First, the parallel rule converts the states $|\psi^{r_1}\rangle$ and $|\psi^{r_2}\rangle$ ($r_1\geq r_2$) into $|\psi^{r_{1,2}}\rangle$ with $\sinh r_{1,2} = \sinh r_1 \cosh r_2$ between $S$ and $R$; second, the series rule transforms $|\psi^{r_{1,2}}\rangle$ and $|\psi^{r_{3}}\rangle$ to the final state $|\psi^{r}\rangle$ with the ratio negativity $\text{X}_{\text{SC}}=\tanh r_{1,2} \tanh r_3$ between $S$ and $T$.

Figure 2:
Bethe lattice.
    (a)~A Bethe lattice of degree $k$ (i.e.,~each node is incident to $k$ links) and network depth $l$ (the path length from the yellow node to the red nodes).
   (b)~The sponge-crossing ratio negativity $\mathrm{X}_\text{SC}$ {between $S$ and $T$ for various $k$ (right panel), satisfying the power law $\mathrm{X}_\text{SC}-\mathrm{X}_\text{SC}^{+}\sim |\chi-\chi_{\text{th}}|^{0.47(5)}$ as $\chi\to\chi_{\text{th}}^+$ (left panel).}
    The numerical value~$0.47\pm0.05$ is derived by a linear least-squares fit to the sixteen data points.
   (c)~
    When $\chi\to\chi_{\text{th}}^-$, $\mathrm{X}_\text{SC}$ exhibits a plateau behavior until the network depth $l$ exceeds the correlation length $l^*$  
    {(defined as the depth $l$  at which $\mathrm{X}_\text{SC}=0.5$),}
    after which $\mathrm{X}_\text{SC}$ abruptly drops to zero.
    (d)~Near the critical threshold, we observe $l^* \sim \left|\chi - \chi_{\text{th}}\right|^{-0.508(9)}$, indicating $z\nu \approx 1/2$.

Figure 3:
Entanglement percolation in two-dimensional square  lattices.
    (a)~$\text{X}_{\text{SC}}$ for square lattices with different side length $L$.
    (b)~Scaling of the correlation length {$\xi$} near the critical threshold $\chi_{\text{th}}\approx0.715$ follows {$\xi \sim |\chi - \chi_{\text{th}}|^{-\nu}$}, with a fitted critical exponent $\nu = 0.02\pm0.02$.

Figure 4:
Feedback stabilization of QN  against entanglement decay. 
    (a)~Under the same feedback control [Eq.~\eqref{eq-FOPTD}], 
    the DV-based QN ($k=3$ Bethe lattice) exhibits rapid stabilization; (b)~whereas the CV-based QN exhibits long-term ``on/off''
    instability, a direct result of the abrupt drop in Fig.~\ref{Figure2}(b).

Figure 5:
Continuous-variable entanglement swapping. 
   Given $N$ TMSVSs $|\psi^{r_{1}}\rangle,|\psi^{r_{2}}\rangle,\dots,|\psi^{r_{N}}\rangle$, the final state $|\psi^{r}\rangle$ between $S$ and $T$ derived by entanglement swapping is a TMSVS, where the squeezing parameter $r$ satisfies Eq.~\eqref{1D}.

Figure 6:
Continuous-variable entanglement concentration.
     Consider $K$ parallel TMSVSs written as the tensor product: $\bigotimes\nolimits_{k=1}^{K}|\psi^{r_{k}}\rangle$.
     Step $a_1$: The TMSVS $|\psi^{r_{1}'}\rangle$ is obtained by executing the TMSVS entanglement concentration protocol on  $|\psi^{r_{1}}\rangle\otimes|\psi^{r_{2}}\rangle$;
     Step $a_2$: Performing the scheme again on $|\psi^{r_{1}'}\rangle\otimes|\psi^{r_{3}}\rangle$ yields the TMSVS $|\psi^{r_{2}'}\rangle$; and so on.
     This eventually results in a single TMSVS $|\psi^{r}\rangle$, where the squeezing parameter $r$ satisfies Eq.~\eqref{(N+M)-CONCEN}.

Figure 7:
Wheatstone bridge. 
    (a)~The original QN. (b)~The QN after the ``Y-$\Delta$'' transform. 
    (c)~The final QN.

Figure 8:
Kelvin bridge. 
    (a)~The original QN.~(b)~The QN after the ``$\Delta$-Y'' transform.~(c)~The QN after applying series rules to $R_1$, $R_2$ and $R_3$.~(d)~The final QN.

\vspace{5mm}
\section*{Acknowledgements}

Y.Z. (Zhao), K.H., and J.H. were supported by the National Natural Science Foundation of China under Grants No.~12271394 and No.~12071336. 
S.H. was supported by the Israel Ministry of Innovation, Science \& Technology (grant number 01017980), the Israel Science Foundation (grant number 201/25), the Binational Israel-China Science Foundation (grant number 3132/19), and the European Union under the Horizon Europe grant OMINO (grant number 101086321).

\section*{Author contributions}

All authors participated in discussing the results and reviewing the manuscript. The original concept was conceived by Y.Z. (Zhao), K.H., and J.H. Numerical computations were performed by Y.Z. (Zhao) and X.M. The theoretical analysis was conducted by Y.Z. (Zhao), X.M., and Y.Z. (Zhang). 
The initial draft of the manuscript was written by Y.Z. (Zhao) and X.M., and revised by Y.Z. (Zhang), J.G., S.H., K.H., and J.H. All authors have read and approved the manuscript.

\section*{Competing interests}
The authors declare no competing interests.

\end{document}


\title{Supplemental Information for ``Negativity Percolation in Continuous-Variable Quantum Networks''}

\setcounter{footnote}{3}
\author{Yaqi Zhao}
\affiliation{College of Mathematics, Taiyuan University of Technology, Taiyuan, 030024, China}

\author{Kan He}
\email{hekan@tyut.edu.cn}
\affiliation{College of Mathematics, Taiyuan University of Technology, Taiyuan, 030024, China}
\affiliation{School of Mathematics and Statistics, Shaanxi Normal University, Xi'an, 710119, China}

\author{Yongtao~Zhang}
\affiliation{Network Science and Technology Center, Rensselaer Polytechnic Institute, Troy, New York 12180, USA}
\affiliation{Department of Physics, Applied Physics, and Astronomy, Rensselaer Polytechnic Institute, Troy, New York 12180, USA}

\author{Jinchuan Hou}
\email{jinchuanhou@aliyun.com}
\affiliation{College of Mathematics, Taiyuan University of Technology, Taiyuan, 030024, China}

\author{Jianxi Gao}
\affiliation{Network Science and Technology Center, Rensselaer Polytechnic Institute, Troy, New York 12180, USA}
\affiliation{Department of Computer Science, Rensselaer Polytechnic Institute, Troy, New York 12180, USA}

\author{Shlomo Havlin}
\affiliation{Department of Physics, Bar-Ilan University, Ramat Gan 52900, Israel}

\author{Xiangyi Meng}
\email{xmenggroup@gmail.com}
\affiliation{Network Science and Technology Center, Rensselaer Polytechnic Institute, Troy, New York 12180, USA}
\affiliation{Department of Physics, Applied Physics, and Astronomy, Rensselaer Polytechnic Institute, Troy, New York 12180, USA}

\maketitle


\renewcommand{\thesection}{\Roman{section}}
\renewcommand{\theequation}{\arabic{equation}}
\renewcommand{\thesubsection}{\Alph{subsection}}

\renewcommand{\thefigure}{S\arabic{figure}}  
\makeatletter
\renewcommand{\p@subfigure}{\thefigure}     
\makeatother

\setcounter{figure}{0}
\setcounter{equation}{0}

\section{Negativity Percolation in Bethe Lattice}\label{sec-Bethe}

\subsection{Bethe lattice}\label{sec-Bethe_lattice}

The Bethe lattice of degree $k$ is an infinite tree network where every node is connected to exactly $k$ others. This structure can be visualized as an endless outward expansion from a central root. The first expansion connects the root to $k$ new nodes. Each subsequent expansion adds $k-1$ new nodes to every node in the outermost layer. This process continues indefinitely, as shown for $k=3$ in Fig.~\ref{fig-Bethe_k3}.

A finite-size Bethe lattice of degree $k$ and depth $l$, also known as a Cayley tree, is a truncated version of this structure. It is formed by limiting the expansion to $l$ generations from the root. While the internal nodes retain their degree $k$, the outermost ``leaf'' nodes have a degree of $1$. The subgraph outlined by the dashed line in Fig.~\ref{fig-Bethe_k3}, for instance, shows a finite Bethe lattice with $k=3$ and depth $l=2$.

\begin{figure}[h!]
    \centering
    \includegraphics[width=100pt]{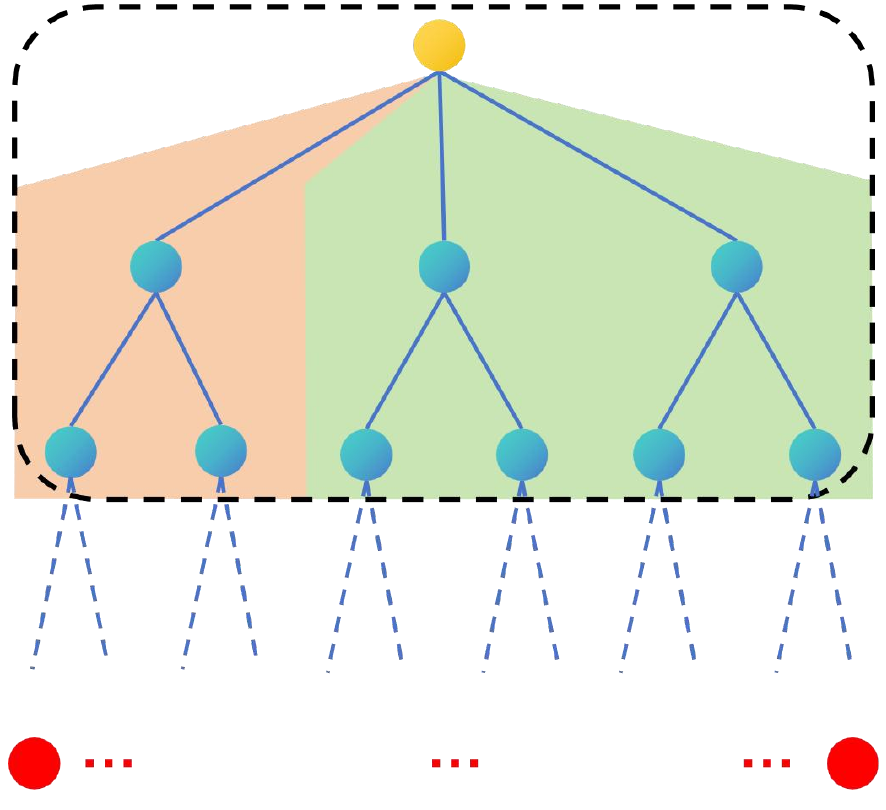}
    \caption{A Bethe lattice with degree $k=3$.}
    \label{fig-Bethe_k3}
\end{figure}

\subsection{Self-consistent renormalization group equation}\label{sec-th}
In the thermodynamic limit (the number of nodes goes to infinity, $N\to \infty$), we assume that a nontrivial threshold on $\chi$ exists:
\begin{align}
\label{eq_cth}
\chi_{\text{th}} = \text{inf}\left\{\chi \in [0, 1] :\ \lim_{N \rightarrow \infty} \mathrm{X}_\text{SC} > 0\right\}.
\end{align}
In other words, there exists an infimum  entanglement $\chi_{\text{th}}$ per link that makes $\mathrm{X}_\text{SC}$ exceed zero.
A saturation point $\chi_{\text{sat}}\in[0,1]$ of $\chi$ 
is defined as
\begin{eqnarray}
    \chi_{\text{sat}}=\text{inf}\left\{\chi\in[0,1]:\ \lim_{N \rightarrow \infty}\mathrm{X}_{\text{SC}}=1\right\},
\end{eqnarray}
i.e.,~a minimum entanglement $\chi_{\text{sat}}$ per link that maximizes $\mathrm{X}_{\text{SC}}$.

For a finite-size Bethe lattice of degree $k$, we denote by $\chi_{l}^{(1)}$ the sponge-crossing ratio negativity for one branch of the lattice (shaded orange in Fig.~\ref{fig-Bethe_k3}), and $\chi_{l}^{(k-1)}$ the parallel combination of $(k-1)$ branches ($k\geq3$) (shaded green in Fig.~\ref{fig-Bethe_k3}). 
The topology of the Bethe lattice allows us to write down the equations
\begin{eqnarray}
    \chi_{l}^{(1)}
    &=&\chi\chi_{l-1}^{(k-1)},\nonumber\\
    \frac{\left(\chi_{l}^{(k-1)}\right)^2}{1-\left(\chi_{l}^{(k-1)}\right)^2}
    &=&\frac{\left(\chi_{l}^{(1)}\right)^2}{\left[1-\left(\chi_{l}^{(1)}\right)^2\right]^{k-1}}.\label{eq-Appendix_finite_l}
\end{eqnarray}
As the depth $l$ increases, from numerical simulation we observe the emergence of a critical threshold [e.g.,~Fig.~\ref{k3}] based on our series-parallel rules.
This distinguishes two regimes---a subcritical regime where the Bethe lattice does not allow entanglement transmission over infinite range, and a supercritical regime where entanglement transmission over infinite range becomes enabled. This observation can be analytically derived in an exact form, as presented below.

For $l\to\infty$, the Bethe lattice exhibits exact self-similarity. 
In this limit, Eq.~\eqref{eq-Appendix_finite_l} yields the self-consistent renormalization group~\cite{J1996,AE2002} equation,
\begin{eqnarray}
  \label{eq_rg_chi}
\mathrm{X}_\text{SC}^{(1)}&=&\chi{\mathrm{X}_\text{SC}^{(k-1)}},\nonumber\\
     \frac{\left({\mathrm{X}_\text{SC}^{(k-1)}}\right)^2}{1-\left({\mathrm{X}_\text{SC}^{(k-1)}}\right)^2}&=&\frac{\left({\mathrm{X}_\text{SC}^{(1)}}\right)^2}{\left[1-\left({\mathrm{X}_\text{SC}^{(1)}}\right)^2\right]^{k-1}},
\end{eqnarray}
where $\mathrm{X}_\text{SC}^{(1)}:=\chi_{l\to \infty}^{(1)}$ denotes the sponge‑crossing ratio negativity of each orange‑shaded sub-network in Fig.~\ref{fig-Bethe_k3}, and $\mathrm{X}_\text{SC}^{(k-1)}:=\chi_{l \to\infty}^{(k-1)}$ denotes that of each green‑shaded sub-network.
The overall sponge-crossing ratio negativity $\mathrm{X}_\text{SC}$ between $S$ and $T$ is given by
\begin{eqnarray}
  \label{eq_rg_chi_2}
         \frac{\left(\mathrm{X}_\text{SC}\right)^2}{1-\left(\mathrm{X}_\text{SC}\right)^2}&=&\frac{\left({\mathrm{X}_\text{SC}^{(1)}}\right)^2}{\left[1-\left({\mathrm{X}_\text{SC}^{(1)}}\right)^2\right]^{k}}.
\end{eqnarray}

\begin{figure}[t!]
    \centering
    \subfigure[]{
     \includegraphics[width=1.5in]{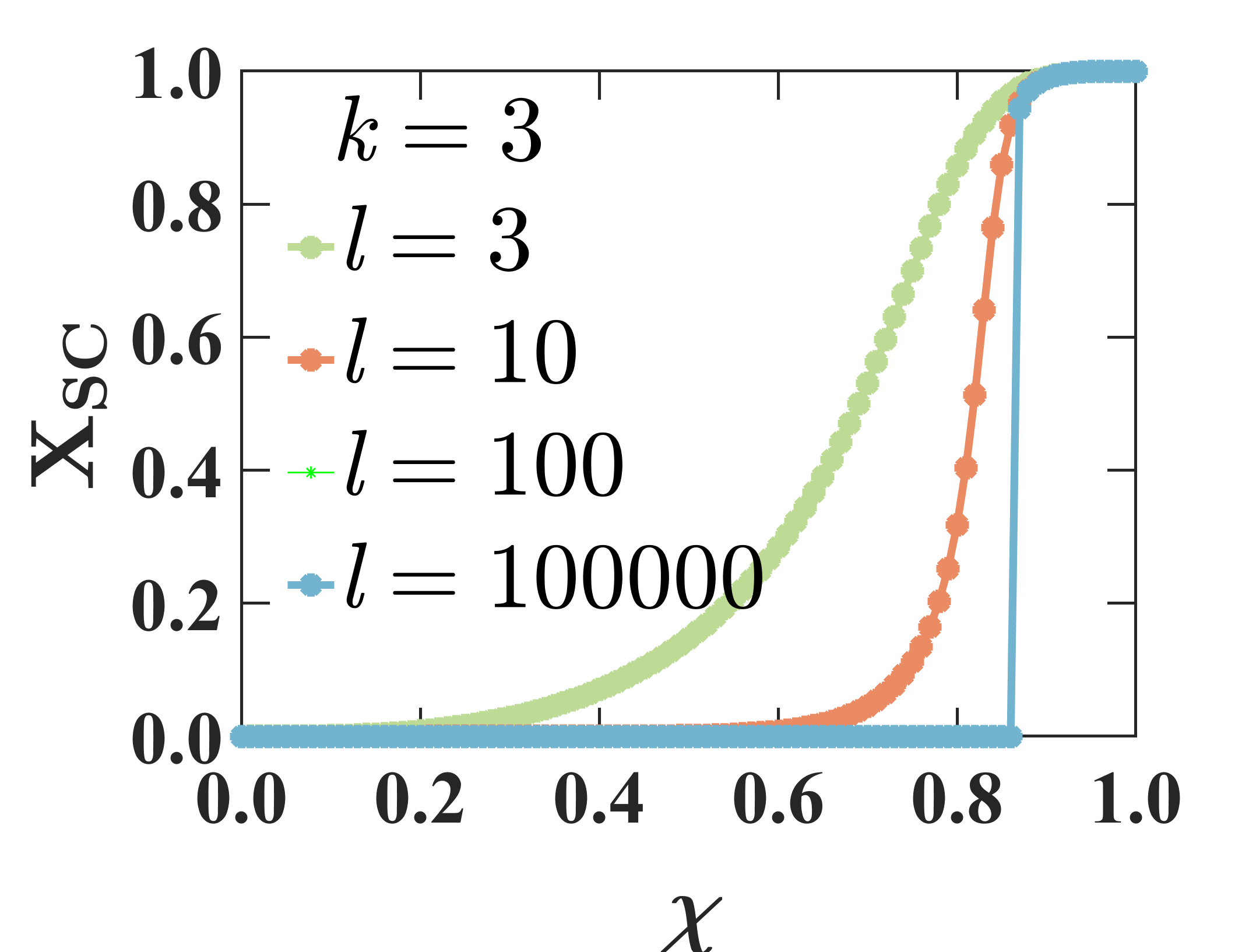}
     \label{k3}
    }
    \subfigure[]{
     \includegraphics[width=1.5in]{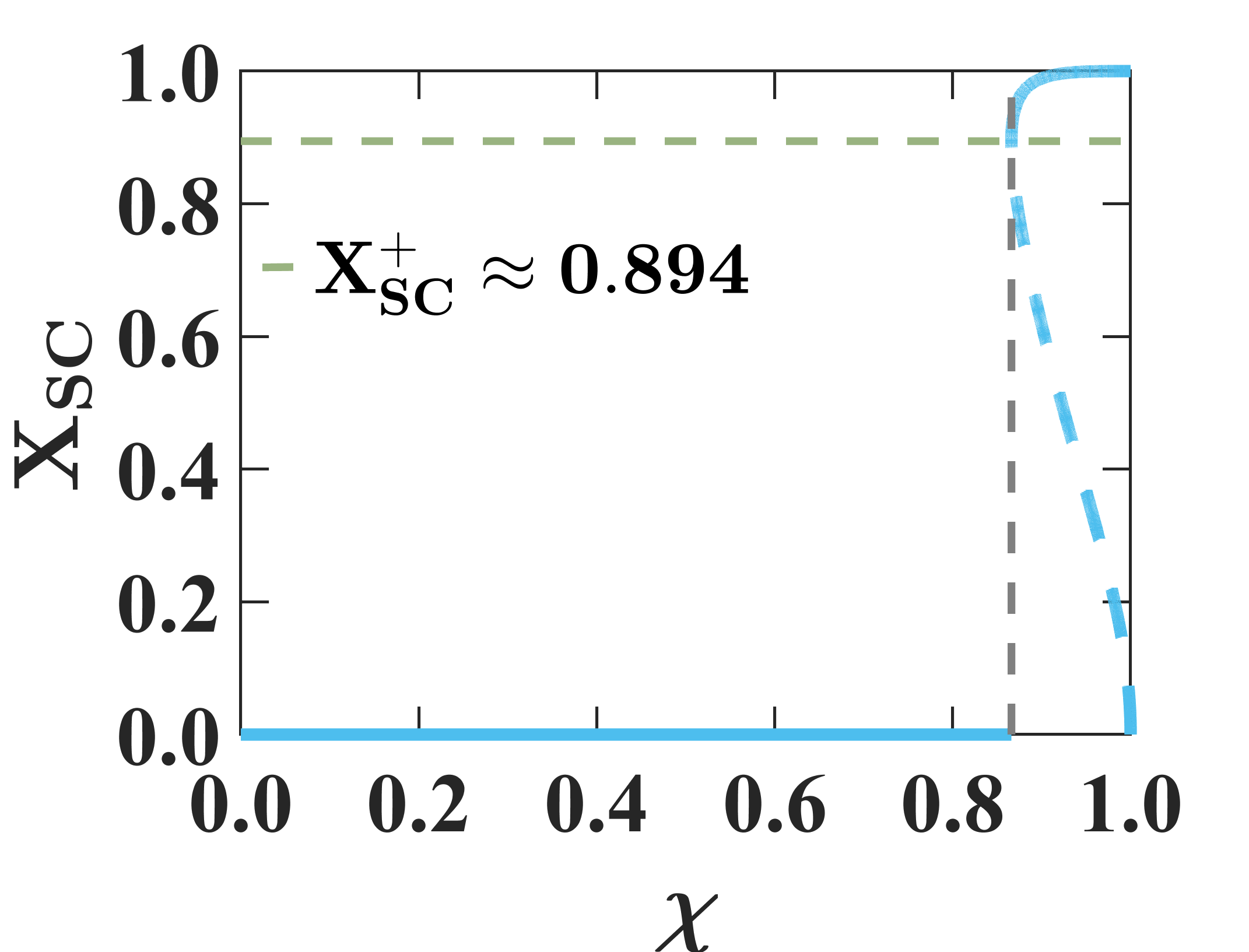}
     \label{Bethe_infty}
    }
    \caption{Bethe lattice. \subref{k3}~The sponge-crossing ratio negativity $\mathrm{X}_\text{SC}$ for finite-size Bethe lattices ($k=3$). The network depth is $l$. The curves for $l=100$ and $l=1000$ are nearly identical. \subref{Bethe_infty}~The exact solution of $\mathrm{X}_\text{SC}$ for $l\to \infty$. The critical threshold is $\chi_{\text{th}}=\sqrt{3}/2\approx0.866025$ (gray dashed). The blue dashed line represents the nonphysical solution in the connected phase. {The green dashed line denotes the value $\mathrm{X}_{\text{SC}}^{+}\approx 0.894$ of $\mathrm{X}_{\text{SC}}$ for $\chi=\chi_{\text{th}}$.}
    }\label{FIG-Bethe}
\end{figure}

\subsection{Critical threshold {and mixed-order phase transition}}\label{sec-pl}

The self-similarity [Eq.~\eqref{eq_rg_chi}] implies
  \begin{equation}
  \label{renormal2}
    \chi^2=\left(\mathrm{X}_\text{SC}^{(1)}\right)^2+\left[1-\left(\mathrm{X}_\text{SC}^{(1)}\right)^2\right]^{k-1},
  \end{equation}
when $\mathrm{X}_\text{SC}^{(1)}>0$.
The fact that $\mathrm{X}_\text{SC}$ and $\chi$ increase or decrease simultaneously requires that
\begin{eqnarray}\label{eq-partial}
    \frac{\partial{\left(\chi^2\right)}}{\partial{\left[\left(\mathrm{X}_\text{SC}^{(1)}\right)^2\right]}}\geq0
\end{eqnarray}
which in turn means the phase transition happens at 
\begin{eqnarray}\label{eq-2deriv}
    \frac{\partial{\left(\chi^2\right)}}{\partial{\left[\left(\mathrm{X}_\text{SC}^{(1)}\right)^2\right]}}\Bigg|_{\mathrm{X}_\text{SC}^{(1)}=\mathrm{X}_\text{SC}^{(1),+}}=0,
\end{eqnarray}
where $\mathrm{X}_\text{SC}^{(1),+} =\sqrt{1-T(k)}$, given $T(k)=(k-1)^{-{1}/{\left(k-2\right)}}$.
The critical threshold is thus given by the value of $\chi$ when $\mathrm{X}_\text{SC}^{(1)}=\mathrm{X}_\text{SC}^{(1),+}$, 
\begin{eqnarray}
  \label{GGthreshold2}
    \chi_{\text{th}}=\sqrt{1-(k-1)^{-\frac{k-1}{k-2}}\left(k-2\right)}.
  \end{eqnarray}

By Eq.~\eqref{eq_rg_chi_2}, we find that at the critical threshold, $\mathrm{X}_\text{SC}$ abruptly jumps to a positive value,
\begin{eqnarray}\label{eq-XSC_plus}
  \mathrm{X}_\text{SC}^{+}=\sqrt{\frac{1-T(k)}{T^k(k)-T(k)+1}}.
\end{eqnarray}

We plot the analytical solution of $\mathrm{X}_\text{SC}$ in Fig.~\ref{Bethe_infty} for $k=3$.
We observe the exact critical phenomenon: as $\chi$ decreases from one, $\mathrm{X}_\text{SC}$ also decreases, until reaching the critical threshold and exhibiting a sudden, discontinuous drop to zero, precisely at  $\chi_{\text{th}}=\sqrt{3}/2$ (gray dashed).
When $\chi<\chi_{\text{th}}$, the solution of Eq.~\eqref{renormal2} trivially marks $\mathrm{X}_\text{SC}=0$, suggesting that the phase transition at the critical threshold $\chi_{\text{th}}$ is discontinuous.

\subsection{Saturation point}\label{sec-saturation}

In concurrence percolation~\cite{XJS2021,OXSGBJ2022}, a nontrivial saturation point exists, above which the sponge-crossing entanglement becomes perfect (even when the entanglement per link is still imperfect).
For the NegPT, however, we expand $\mathrm{X}_\text{SC}$ around $1^{-}$ and find
\begin{eqnarray*}
   &&1-\mathrm{X}_\text{SC}\\
    &\sim& 1-\mathrm{X}_\text{SC}^{2}
    =\frac{\left[1-\left(\mathrm{X}_\text{SC}^{(1)}\right)^{2}\right]^k}{\left(\mathrm{X}_\text{SC}^{(1)}\right)^{2}+\left[1-\left(\mathrm{X}_\text{SC}^{(1)}\right)^{2}\right]^k}\\
    &\sim& \left(1-\mathrm{X}_\text{SC}^{(1)}\right)^k
\end{eqnarray*}
and
\begin{eqnarray*}
    &&1-\chi\\
    &\sim& 1-\chi^2\\
    &=&1-\left(\mathrm{X}_\text{SC}^{(1)}\right)^{2}-\left[1-\left(\mathrm{X}_\text{SC}^{(1)}\right)^{2}\right]^{k-1}\\
    &\sim& 1-\mathrm{X}_\text{SC}^{(1)},
\end{eqnarray*}
implying
\begin{eqnarray}
    1-\mathrm{X}_\text{SC}\sim (1-\chi)^k
\end{eqnarray}
 when $\chi\to 1^{-}$.
Since $\mathrm{X}_\text{SC}$ becomes one iff $\chi=1$ [Eq.~\eqref{renormal2}], we find that there is no non-trivial saturation point (or trivially speaking $\chi_{\text{sat}}=1$) in NegPT (in contrast to concurrence percolation~\cite{XJS2021}).

\subsection{Critical exponents $\beta$ and $z\nu$}\label{sec-power}

To calculate the critical exponent $\beta$, we expand the analytical solution (for $k=3$) of $\mathrm{X}_\text{SC}$ around $\mathrm{X}_\text{SC}^{+}$:
\begin{eqnarray*}
    && \left|\mathrm{X}_\text{SC}-\mathrm{X}_\text{SC}^{+}\right|\\
    &\sim&\left|\left(\mathrm{X}_\text{SC}\right)^2-\left(\mathrm{X}_\text{SC}^{+}\right)^2\right|\nonumber\\
    &\sim&\left|\left(\mathrm{X}_\text{SC}^{(1)}\right)^{2}-\left(\mathrm{X}_\text{SC}^{(1),+}\right)^{2}\right|.
\end{eqnarray*}
Meanwhile, we have
\begin{equation}
    |\chi-\chi_{\text{th}}|
    \sim\left|\chi^2-\chi_{\text{th}}^2\right|
    \sim \left|\left(\mathrm{X}_\text{SC}^{(1)}\right)^{2}-\left(\mathrm{X}_\text{SC}^{(1),+}\right)^{2}\right|^2,
\end{equation}
where the second step is derived by the following:
\begin{eqnarray}
   && \lim\limits_{\mathrm{X}_\text{SC}\rightarrow\mathrm{X}_\text{SC}^{+}}\frac{\chi^2-\chi_{\text{th}}^2}{\left[\left(\mathrm{X}_\text{SC}^{(1)}\right)^{2}-\left(\mathrm{X}_\text{SC}^{(1),+}\right)^{2}\right]^2}\nonumber\\
    &=&\lim\limits_{\left(\mathrm{X}_\text{SC}^{(1)}\right)^{2}\rightarrow \left(\mathrm{X}_\text{SC}^{(1),+}\right)^{2}} \frac{\chi^2-\chi_{\text{th}}^2}{\left[\left(\mathrm{X}_\text{SC}^{(1)}\right)^{2}-\left(\mathrm{X}_\text{SC}^{(1),+}\right)^{2}\right]^2}\nonumber\\
    &=&\lim\limits_{\left(\mathrm{X}_\text{SC}^{(1)}\right)^{2}\rightarrow \left(\mathrm{X}_\text{SC}^{(1),+}\right)^{2}} \frac{1-(k-1)\left[1-\left(\mathrm{X}_\text{SC}^{(1)}\right)^{2}\right]^{k-2}}{2\left[\left(\mathrm{X}_\text{SC}^{(1)}\right)^{2}-\left(\mathrm{X}_\text{SC}^{(1),+}\right)^{2}\right]}\nonumber\\
    &=&\lim\limits_{\left(\mathrm{X}_\text{SC}^{(1)}\right)^{2}\rightarrow \left(\mathrm{X}_\text{SC}^{(1),+}\right)^{2}} \frac{(k-1)(k-2)\left[1-\left(\mathrm{X}_\text{SC}^{(1)}\right)^{2}\right]^{k-3}}{2}\nonumber\\
    &=&\frac{1}{2}{(k-1)(k-2)\left[1-\left(\mathrm{X}_\text{SC}^{(1),+}\right)^{2}\right]^{k-3}}.
\end{eqnarray}
Here, the second and third steps are derived from Eq.~\eqref{eq-2deriv}, which is equivalent to
\begin{equation}
    1-\left(k-1\right)\left[1-\left(\mathrm{X}_\text{SC}^{(1),+}\right)^{2}\right]^{k-2}=0.
\end{equation}
We thus obtain the power law:
\begin{eqnarray}
   \label{relationship11}
    \left|\mathrm{X}_\text{SC}-\mathrm{X}_\text{SC}^{+}\right|
    &\sim&\left|\chi-\chi_{\text{th}}\right|^{\frac{1}{2}},
\end{eqnarray}
and thus $\beta=1/2$.

In the main paper, we have numerically identified the critical exponent $z\nu$ for the Bethe lattice. Another way to do this is to define the finite-size critical threshold for a Bethe lattice of finite network depth $l$, given by
\begin{eqnarray}
    \chi_{\text{th}}(l)=\left\{\chi\in [0,1]:\ {\partial^2{\mathrm{X}_\text{SC}(l)}}/{\partial{\chi}^2}=0\right\}.
\end{eqnarray}
As $l$ increases, $\chi_{\text{th}}(l)$ approaches $\chi_{\text{th}}$ from the lower side [Fig.~\ref{fig_th_l}], which follows
\begin{eqnarray}
    \chi_{\text{th}}-\chi_{\text{th}}(l)\sim l^{-1/\left(z\nu\right)}.
\end{eqnarray}
Simulation results [Fig.~\ref{fig_nu}] indicate again that $z\nu\approx 1/2$, matching our analysis in the main paper.

\begin{figure}[t!]
    \centering
     \hspace{-6mm}
     \subfigure[]{
     \includegraphics[width=115pt]{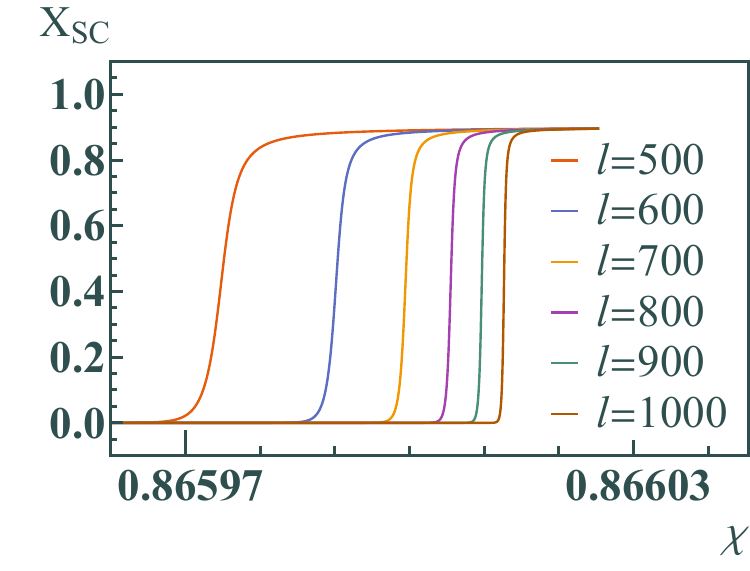}
     \label{fig_th_l}
     }
     \hspace{-5mm}
     \subfigure[]{
     \includegraphics[width=115pt]{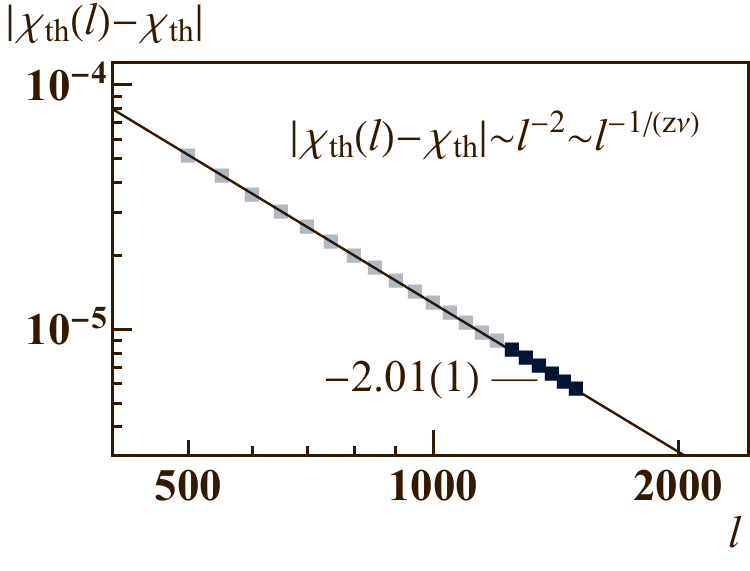}
     \label{fig_nu}
     } 
    \caption{Bethe lattice critical exponent $z\nu$.
    \subref{fig_th_l}~The abruptness of $\mathrm{X}_\text{SC}$ near the finite-size critical threshold $\chi_{\text{th}}(l)$ for finite $l$. Here, $\chi_{\text{th}}(l)$ is defined as the value $\chi$ when $\partial^2\mathrm{X}_{\text{SC}}(l)/\partial{\chi}^2$ vanishes.
    \subref{fig_nu}~The finite-size scaling of the shifting between $\chi_{\text{th}}(l)$ and the true threshold $\chi_{\text{th}}$, yielding $|\chi_{\text{th}}(l)-\chi_{\text{th}}|\sim l^{-2.01(1)}$ versus network depth $l$.}\label{regular lattice 2}
  \end{figure}

\begin{figure}[]
    \centering
     \includegraphics[width=150pt]{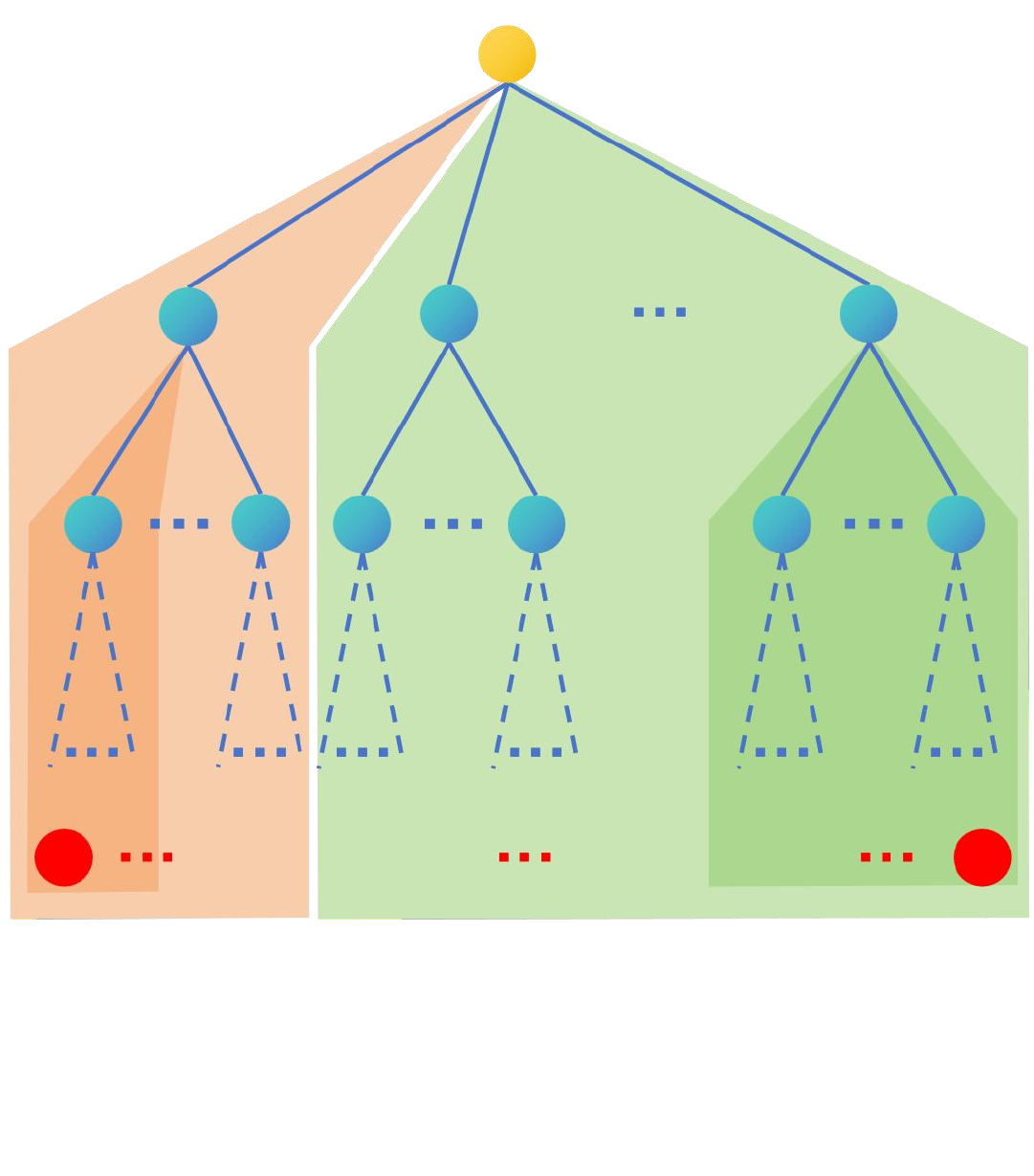}
     \vspace{-10mm}
    \caption{{Bethe lattice. The degree is $k$ (i.e.,~each node is incident to $k$ links) and the network depth is $l=\infty$ (the path length from the yellow node to the red nodes).}
    }\label{fig_bethe}
  \end{figure}

\subsection{Classical interdependent percolation}\label{sec-interdependent}

Comparing with Eq.~\eqref{eq_rg_chi}, the exact self-consistent renormalization group equation for classical percolation on the Bethe lattice is given by
\begin{eqnarray}
  \label{eq_rg_p}
         P_{\mathrm{SC}}^{(1)}&=&p P_{\mathrm{SC}}^{(k-1)},\nonumber\\
         1-P_{\mathrm{SC}}^{(k-1)}&=&\left(1- P_{\mathrm{SC}}^{(1)}\right)^{k-1}.
  \end{eqnarray}
Similarly, $P_{\mathrm{SC}}^{(1)}$ denotes the sponge-crossing probability for one branch of the lattice [shaded orange in Fig.~\ref{fig_bethe}], and $P_{\mathrm{SC}}^{(k-1)}$ denotes the parallel combination of $\left(k-1\right)$ branches ($k\ge 3$) [shaded green in Fig.~\ref{fig_bethe}]. Equation~\eqref{eq_rg_p} leads to
\begin{equation}
    P_{\mathrm{SC}}^{(1)}= p \left[1- \left(1- P_{\mathrm{SC}}^{(1)}\right)^{k-1}\right].
\end{equation}
Now, consider classical percolation on $M$-layer fully interdependent Bethe lattices. The parallel combination of $\left(k-1\right)$ branches, namely, $1- \left(1- P_{\mathrm{SC}}^{(1)}\right)^{k-1}$, must hold simultaneously across all $M$ identical copies of the Bethe lattice. Therefore, this expression should be modified to $\left[1- \left(1- P_{\mathrm{SC}}^{(1)}\right)^{k-1}\right]^M$. This leads to
\begin{equation}
\label{eq_rg_p_2}
    P_{\mathrm{SC}}^{(1)}= p \left[1- \left(1- P_{\mathrm{SC}}^{(1)}\right)^{k-1}\right]^M,
\end{equation}
which bears a similar form to Eq.~(8) in the 2011 paper~\cite{gao2011robustness,interdependent_network_2015}. This self-consistent equation for classical interdependent percolation is different from that of the NegPT as given in Eq.~\eqref{renormal2}, indicating that the two are driven by different mechanisms.

\subsection{Expansions of $P_\text{SC}^{(1)}$ and $\mathrm{X}_\text{SC}^{(1)}$}\label{sec-subInf}

In the following, we show that the expansions of both $P_\text{SC}^{(1)}$ and $\mathrm{X}_\text{SC}^{(1)}$ exist, as given in the main paper.

For classical percolation,  Eq.~\eqref{eq_rg_p_2} simply yields $P_\text{SC}^{(1)}=p\sum_{m=0}^{M}\binom{M}{m}(-1)^m\left(1-P_\text{SC}^{(1)}\right)^{m(k-1)}$.
Therefore, near $P_\text{SC}^{(1),+}$, we have
\begin{eqnarray*}
     P_\text{SC}^{(1),+}
    &=&p_{\text{th}}\sum_{m=0}^{M}\binom{M}{m}(-1)^m\left(1-P_\text{SC}^{(1),+}\right)^{m(k-1)}\\
    &=&p_{\text{th}}\left[1-M\left(1-P_\text{SC}^{(1),+}\right)^{k-1}+\cdots\right].
\end{eqnarray*}

For the NegPT, Eq.~\eqref{renormal2} gives rise to
\begin{eqnarray}\label{eq-chiInf}
    \left(\mathrm{X}_\text{SC}^{(1)}\right)^2=\chi^2 (1-z)^{-1},
\end{eqnarray}
where $z=-\left(\mathrm{X}_\text{SC}^{(1)}\right)^{-2}\left[1-\left(\mathrm{X}_\text{SC}^{(1)}\right)^2\right]^{k-1}$.
Note that $(1-z)^{-1}$ can be expanded as $\sum_{m=0}^{+\infty}z^m$ only when $|z|<1$.
To examine whether $\mathrm{X}_\text{SC}^{(1)}$ can be expanded near $\mathrm{X}_\text{SC}^{(1),+}$, 
let
\begin{eqnarray*}
    z^{+}&=&-\left(\mathrm{X}_\text{SC}^{(1),+}\right)^{-2}\left[1-\left(\mathrm{X}_\text{SC}^{(1),+}\right)^2\right]^{k-1}\\
    &=&-\frac{1}{(k-1)\left[(k-1)^{\frac{1}{k-2}}-1\right]}.
\end{eqnarray*}
Then we find that $|z^{+}|<1$ holds for all $k>2$ since the fact that
\begin{eqnarray*}
    &&|z^{+}|<1\\
    &\leftrightharpoons&(k-1)\left[(k-1)^{\frac{1}{k-2}}-1\right]>1\\
    &\leftrightharpoons&(k-1)\ln{(k-1)>(k-2)\ln{k}}
\end{eqnarray*}
where the last inequality holds for all $k>2$.
Therefore, {for $k>2$}, the value of $z^{+}$ satisfies $|z^{+}|<1$, resulting in the expansion of Eq.~\eqref{eq-chiInf}:
{\fontencoding{T1}
    \fontsize{8}{10}
    \selectfont
\begin{eqnarray*}
    \left(\mathrm{X}_\text{SC}^{(1),+}\right)^2
    &=&\chi_{\text{th}}^{2} \sum_{m=0}^{+\infty}\left(z^{+}\right)^m\\
    &=&\chi_{\text{th}}^2\left\{1-\frac{1}{\left(\mathrm{X}_\text{SC}^{(1),+}\right)^2}\left[1-\left(\mathrm{X}_\text{SC}^{(1),+}\right)^2\right]^{k-1}+\dots\right\}.
\end{eqnarray*}}

\section{Generalized G-G DET Scheme}\label{sec-scheme2}

In this section, we propose generalizations of the series and parallel rules within the G-G DET framework and examine these broader scenarios:

\subsection{Generalized series and parallel rules}\label{sec-general-rule}

First, we generalize the series-parallel rules of the G-G DET in the main paper by incorporating two positive parameters, $\eta_s$ and $\eta_p$, as shown in Table~\ref{table_parameter}.
While these rules do not apply to the entire range of the squeezing parameter $r$ (or $\chi\equiv \tanh r$), as $\eta_s$ or $\eta_p$ are arbitrary, we expect that the rules will capture the dominant effect near the critical threshold.

\begin{table}[h!]
      \caption{Generalized series and parallel rules.}
  \begin{tabular}{ll}
     \Xhline{1pt}
    \hline\hline\noalign{\smallskip}
     Rules \quad\quad & Squeezing parameters {$r_1\geq r_2 \geq \dots$} \\
     \noalign{\smallskip}\hline\noalign{\smallskip}
     series\quad\quad & $\tanh{r}=\eta_s \tanh{r_{1}}\tanh{r_{2}} \tanh{r_{3}}\dots$\\
     parallel\quad\quad & $\sinh{r}=\eta_p \sinh{r_{1}}\cosh{r_{2}}\cosh{r_{3}}\dots$ \\
     \noalign{\smallskip}\hline
     \label{table_parameter}
     \end{tabular}
	\end{table}

\begin{figure}[t!]
    \centering
    \includegraphics[width=2.7in]{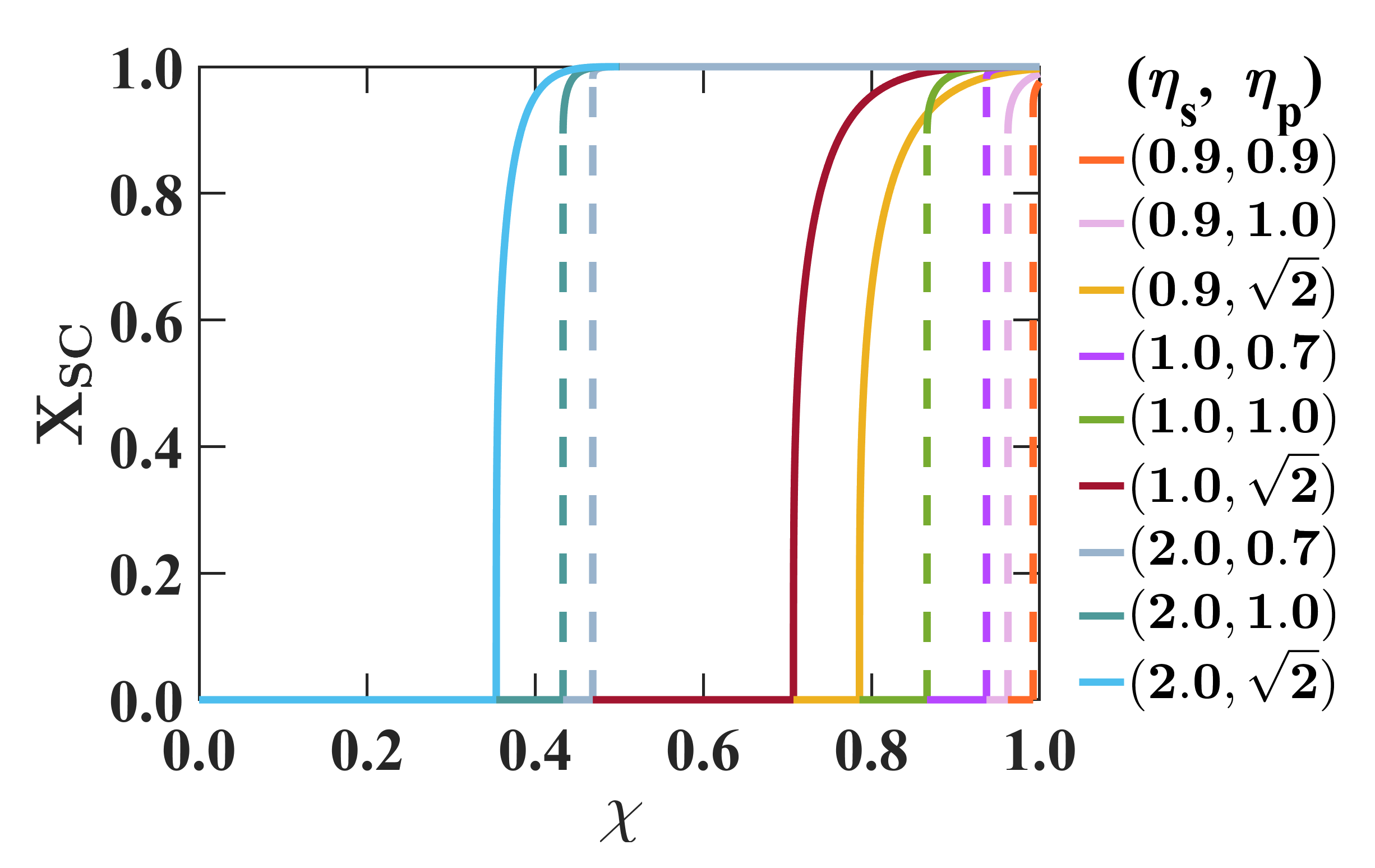}
    \caption{Generalized G-G DET scheme (Table~\ref{table_parameter}) on the Bethe lattice ($k=3$). Dashed curves represent mixed-order phase transitions (with discontinuous jumps), while solid curves represent second-order phase transitions.}
    \label{fig-parameters}
\end{figure}

Consider again the Bethe lattice for general degree $k$. 
Compared to Eqs.~\eqref{eq_rg_chi}~and~\eqref{eq_rg_chi_2}, the self-consistent renormalization group equations are modified to:
\begin{eqnarray}
  \label{renormal3}
    \mathrm{X}_\text{SC}^{(1)}
         &=&\eta_s\chi{\mathrm{X}_\text{SC}^{(k-1)}}\nonumber\\
         \frac{\left({\mathrm{X}_\text{SC}^{(k-1)}}\right)^2}{1-\left({\mathrm{X}_\text{SC}^{(k-1)}}\right)^2}
         &=&\eta_p^2\frac{\left({\mathrm{X}_\text{SC}^{(1)}}\right)^2}{\left[1-\left({\mathrm{X}_\text{SC}^{(1)}}\right)^2\right]^{k-1}},
  \end{eqnarray}
  and
  \begin{equation}   
         \frac{\mathrm{X}_\text{SC}^2}{1-\mathrm{X}_\text{SC}^2}
         ={\eta_p^2}\frac{\left({\mathrm{X}_\text{SC}^{(1)}}\right)^2}{\left[1-\left({\mathrm{X}_\text{SC}^{(1)}}\right)^2\right]^{k}}.
  \end{equation}
The nature of the phase transition (mixed-order or second-order) is determined by the existence of the real root $\mathrm{X}_\text{SC}^{(1),+}=\sqrt{1-\left[\eta_p^{2}/(k-1)\right]^{1/(k-2)}}$ of the equation $\partial{\left(\chi^2\right)}/\partial{\left[\left(\mathrm{X}_\text{SC}^{(1)}\right)^2\right]}=0$.
Since the necessary and sufficient condition for $\mathrm{X}_\text{SC}^{(1),+}>0$ is given by $\eta_p<\sqrt{k-1}$, we continue to discuss the phase transition in two cases:

{\bf Case~1 ({$\eta_p<\sqrt{k-1}$}).}
For this case,
we find that $\mathrm{X}_\text{SC}$ exhibits a {mixed-order phase transition} when 
\begin{eqnarray}
 \label{condition1}
   \sqrt{1-\eta_p^{\frac{2}{k-2}}(k-1)^{-\frac{k-1}{k-2}}(k-2)}<\eta_s;
\end{eqnarray}
otherwise, $\mathrm{X}_\text{SC}$ is zero for all $\chi$.
When Eq.~\eqref{condition1} holds, we obtain the critical threshold
\begin{eqnarray}
  \chi_{\text{th}}=\eta_s^{-1}\sqrt{1-\eta_p^{\frac{2}{k-2}}(k-1)^{-\frac{k-1}{k-2}}(k-2)},
\end{eqnarray}
whose variation is inversely proportional to the changes in two parameters.
The value is well defined since 
\begin{eqnarray}
    \eta_p^{\frac{2}{k-2}}(k-1)^{-\frac{k-1}{k-2}}(k-2)<\frac{k-2}{k-1}<1.
\end{eqnarray}
At the threshold $\chi_\text{th}$, the sponge-crossing ratio negativity $\mathrm{X}_\text{SC}$ abruptly jumps to 
\begin{eqnarray}
  \mathrm{X}_\text{SC}^{+}=\sqrt{\frac{1-F(k,\eta_p)}{{\eta_p^{-2}}F^{k}(k,\eta_p)-F(k,\eta_p)+{1}}},
\end{eqnarray}
where $F(k,\eta_p)={\eta_p}^{{2}/{\left(k-2\right)}}(k-1)^{-{1}/{\left(k-2\right)}}$.
When $\chi\rightarrow\chi_{\text{th}}^{+}$, we again find that  
\begin{eqnarray}
   |\mathrm{X}_\text{SC}-\mathrm{X}_\text{SC}^{+}|\propto|\chi-\chi_{\text{th}}|^{\frac{1}{2}}.
\end{eqnarray}

{{\bf Case~2 ($\eta_p\geq\sqrt{k-1}$).}
For this case, 
we find that $\mathrm{X}_\text{SC}$ exhibits a {second-order phase transition} when 
\begin{eqnarray}
 \label{condition2}
   \eta_p^{{-1}}<\eta_s;
\end{eqnarray}
otherwise, $\mathrm{X}_\text{SC}$ is zero for all $\chi$.
When Eq.~\eqref{condition2} holds, we obtain the critical threshold
\begin{eqnarray}
     \chi_{\text{th}}=\frac{1}{\eta_s{\eta_p}}.
\end{eqnarray}
Near the threshold, we find
\begin{eqnarray}
    \mathrm{X}_\text{SC}\propto|\chi-\chi_{\text{th}}|^{\frac{1}{2}}
\end{eqnarray}
for $\eta_p>\sqrt{k-1}$, and 
\begin{eqnarray}
    \mathrm{X}_\text{SC}\propto|\chi-\chi_{\text{th}}|^{\frac{1}{4}}
\end{eqnarray}
for $\eta_p=\sqrt{k-1}$.
{Different phase transition curves are shown in Fig.~\ref{fig-parameters}. }
}

In the following sections, we explore the possible physical origins of the additional parameters $\eta_s$ and $\eta_p$. As we will see, $\eta_s\le 1$ could stem from the consideration of a more general class of Gaussian operations, known as Gaussian
positive operator-valued measurements (GPOVMs), whereas $\eta_p \le 1$ could stem from considering the entanglement concentration operation tailored for a generic category of non-Gaussian states that exhibits TMSVS-like exponentially decaying tails $\sim \chi^n |nn\rangle$ in the Fock space.

\subsection{General entanglement swapping using GPOVMs}\label{sec-swap2}

There exists a more general family of Gaussian-to-Gaussian entanglement swapping operations, necessitating the usage of GPOVMs~\cite{JSM2002}.
{These measurements utilize the standard toolbox of linear optics, meaning they can be implemented using beam splitters, phase shifters, squeezers, auxiliary vacuum states, and balanced homodyne detection.} 
Generally, a positive operator-valued measurement (POVM)  is represented as a set of positive operators $\{M_{\alpha}\}_{\alpha}$ satisfying the completeness $\Sigma_{\alpha}M_{\alpha}=I$ or $\int M_{\alpha}d\alpha=I$.
Specifically, for a $K$-mode CV system, a GPOVM $\{\Pi(\alpha)\}$ takes the following form~\cite{GPOVM2002a,GJ2002,gaussian-povm_fm07},
  \begin{eqnarray}
   \label{GPOVM}
   \Pi(\alpha)=\frac{1}{\pi^{K}}D(\alpha)\varpi D^{\dagger}(\alpha),
  \end{eqnarray}
where $D(\alpha)=\exp\left[\sum_{j=1}^{K}(\hat\alpha_{j}\hat{a}_{j}^{\dagger})
 -\hat\alpha_{j}^{\dagger}\hat{a}_{j}\right]$ is the $K$-mode Weyl displacement operator with $\alpha\in{\mathbb{C}}^{K}$, and $\varpi$ is a zero mean $K$-mode Gaussian state which is called the seed state. One can show that $\int\Pi(\alpha)d\alpha=I$, satisfying the definition of a POVM.

Let $\Pi^{A}(\alpha)=\pi^{-K}D(\alpha)\varpi D^{\dagger}(\alpha)$ be any GPOVM on the subsystem $\mathcal{H}_{A}$ with seed state $\varpi$. Denote by $\mathbf{\Sigma}$ the covariance matrix (CM) of $\varpi$ and $\rho_{\alpha}$
the post-measurement state of a state $\rho\in GS(\mathcal H_{AB})$ with measurement result $\alpha$.
If the CM of $\rho$ is 
$\mathbf{\Sigma}_{\rho}=
\begin{pmatrix}
    \mathbf{A} & \mathbf{C}\\
    \mathbf{C}^{\mathrm{T}} & \mathbf{B}
\end{pmatrix}.
$
Then the CM $\mathbf{\Gamma}_{\alpha}^{B}$ of the partial state $\rho^{B}_{\alpha}=\Tr_{A}(\rho_{\alpha})$ is independent of $\alpha$~\cite{GPOVM2002a,GJ2002,gaussian-povm_fm07}:
  \begin{eqnarray}
   \label{M-CM}
   \mathbf{\Gamma}_{\alpha}^{\mathbf{B}}=\mathbf{B}-\mathbf{C}^{\mathrm{T}}(\mathbf{A}+\mathbf{\Sigma})^{-1}\mathbf{C}.
  \end{eqnarray}
Therefore, this type of measurement can be considered deterministic.

{Consider a Gaussian QN with three nodes $S,R,T$, in which $S$ and $R$ share one TMSVS $|\psi^{r_1}\rangle$, and $R$ and $T$ share a second TMSVS $|\psi^{r_2}\rangle$.}

\begin{theorem}\label{GPOVM1}
     Let node $R$ perform a GPOVM $\{M_{\alpha}\}_\alpha$ with a TMSVS {$|\psi^{r_0}\rangle$} as the seed state. Then, the possible output Gaussian states $\{|\psi_{\alpha}\rangle\}$ between $S$ and $T$ have the same CM
     \begin{eqnarray}
      \mathbf{\Gamma}_{ST}=
        \begin{pmatrix}
         \cosh{2r}\mathbf{I}_{2} & \sinh{2r}\mathbf{Z} \\
         \sinh{2r}\mathbf{Z} & \cosh{2r}\mathbf{I}_{2}
        \end{pmatrix}
     \end{eqnarray}
    with
    \begin{eqnarray}
     \label{r-swapping}
      \tanh{r}=\chi^{(0)}\tanh{r_{1}}\tanh{r_{2}},
    \end{eqnarray}
    and $\chi^{(0)}=\tanh r_0$,
    which is independent of the measurement result $\alpha$.
    Here, $\mathbf{I}_{2}$ is the $2\times 2$ identity matrix and $\mathbf{Z}=
    \begin{pmatrix}
        1 & 0\\
        0 &-1
    \end{pmatrix}$ is one of the Pauli matrices.
\end{theorem}

\begin{proof}
{Consider $S$ and $R$ share a TMSVS $|\psi^{r_1}\rangle$, and $R$ and $T$ share another TMSVS $|\psi^{r_2}\rangle$. The two modes within $R$ perform a GPOVM with a seed state $|\psi^{r_0}\rangle$.}
The CM $\mathbf{\Gamma}_{j}$ of state $|\psi^{r_{j}}\rangle$ ($j=0,1,2$) is 
  \begin{eqnarray}
   \label{gamma-k}
   \mathbf{\Gamma}_{j}=
   \begin{pmatrix}
     \cosh{2r_{j}}\mathbf{I}_{2} & \sinh{2r_{j}}\mathbf{Z}\\
     \sinh{2r_{j}}\mathbf{Z} & \cosh{2r_{j}}\mathbf{I}_{2}
   \end{pmatrix},
  \end{eqnarray}
  So the total CM of $S\text{-}R\text{-}T$ is
  \begin{eqnarray}
   \label{SR1-R2T}
   \mathbf{\Gamma}_{S\text{-}R\text{-}T}
    =\mathbf{\Gamma}_{1}\oplus\mathbf{\Gamma}_{2},
   \end{eqnarray}
{which can also be described as a bipartite form between one party $R$ and the other, $ST$}:
   \begin{eqnarray*}
     \mathbf{\Gamma}_{R-ST}=
     \begin{pmatrix}
      \mathbf{R}_{M}& \mathbf{C}_{R-ST}\\
      \mathbf{C}_{R-ST}^{T} & \mathbf{R}_{M}
     \end{pmatrix}.
    \end{eqnarray*}
The CMs of $R$ and $ST$ are both 
   \begin{eqnarray*}
    \mathbf{R}_{M}=&
    \begin{pmatrix}
      \cosh{2r_{1}}\mathbf{I}_{2} & \mathbf{0} \\
      \mathbf{0} & \cosh{2r_{2}}\mathbf{I}_{2}
   \end{pmatrix}.
   \end{eqnarray*}
The matrix   
   \begin{eqnarray*}
    \mathbf{C}_{R-ST}=&
    \begin{pmatrix}
      \sinh{2r_{1}}\mathbf{Z} & \mathbf{0} \\
      \mathbf{0} & \sinh{2r_{2}}\mathbf{Z}
   \end{pmatrix}
  \end{eqnarray*}
is the CM between $R$ and $ST$.
Following  Eq.~\eqref{M-CM}, we can calculate the CM, $\mathbf{\Gamma}_{ST}$, of the output state $\rho_{\alpha}$:
  \begin{eqnarray*}
   \mathbf{\Gamma}_{ST}&=&\mathbf{R}_{M}-\mathbf{C}_{R-ST}^{T}(\mathbf{R}_{M}+\mathbf{\Gamma}_{0})^{-1}\mathbf{C}_{R-ST}\\
   &=&\begin{pmatrix}
       a\mathbf{I}_{2} & c\mathbf{Z}\\
       c\mathbf{Z} & a\mathbf{I}_{2}
    \end{pmatrix},
  \end{eqnarray*}
  where
  \begin{eqnarray*}
    &a=\frac{\cosh{2r_{0}}(\cosh{2r_{1}}\cosh{2r_{2}}+1)+\cosh{2r_{1}}+\cosh{2r_{2}}}
        {\cosh{2r_{1}}\cosh{2r_{2}}+1+\cosh{2r_{0}} (\cosh{2r_{1}}+\cosh{2r_{2}})},\\
    &c=\frac{\sinh{2r_{0}}\sinh{2r_{1}}\sinh{2r_{2}}}
       {\cosh{2r_{1}}\cosh{2r_{2}}+1+\cosh{2r_{0}}(\cosh{2r_{1}}+\cosh{2r_{2}})}.
  \end{eqnarray*}
Note that $c=\sqrt{a^2-1}$, and thus the post-measurement state remains a $(1+1)$-mode pure Gaussian state.
Suppose its squeezing parameter is $r$, then we have 
\begin{eqnarray*}
    \tanh r
    &=&\sqrt{\frac{a-1}{a+1}}\\
    &=&\sqrt{\frac{\cosh 2r_0-1}{\cosh 2r_0+1}\cdot\frac{\cosh 2r_1-1}{\cosh 2r_1+1}\cdot\frac{\cosh 2r_2-1}{\cosh 2r_2+1}}\\
    &=&\chi^{(0)}\tanh r_1\tanh r_2,
\end{eqnarray*}
where $\chi^{(0)}=\tanh r_0$.
\end{proof}

\emph{Generalized series rule.---}The output state $|\psi_{\alpha}\rangle$ can be transformed into TMSVS $|\psi^{r}\rangle$ deterministically. The optimal scenario $\tanh r=\tanh r_1\tanh r_2$ is the situation where $r_{0}\rightarrow \infty$ in Theorem~\ref{GPOVM1}. 
We consider a one-dimensional QN {consisting of nodes $S$, $T$, and $N-1$ intermediate nodes $R_{1},\dots,R_{N-1}$ [Fig. 5 in the main text]}. When the intermediate nodes $\{R_{j}\}_{j}$ perform GPOVMs, with a subset of them being general with seed states  of ratio negativity $\{\chi_{0_{m}}\}_{m}$, and the remainder being optimal, the resulting state between $S$ and $T$ is a TMSVS $|\psi^{r}\rangle$ with
  \begin{eqnarray}\label{eq-eta_s}
   \tanh{r}=\eta_s\prod_{j}\tanh{r_{j}},
  \end{eqnarray}
where $\eta_s=\prod_{m}\chi_{0_{m}}$.
This provides us the generalized series rule.

\subsection{Entanglement concentration for generic continuous-variable states with $\sim \chi^n |nn\rangle$}\label{sec-parameter}

Consider entanglement concentration for generic non-Gaussian states, $|\varphi_{D-C}\rangle$, that are in the form of the superposition of a DV state and a Gaussian CV state:
\begin{equation}\label{eq-GeneralCV}
    |\varphi_{D-C}\rangle=\sqrt{1-c_{\chi}}|\varphi_{n_0}\rangle + \sqrt{c_{\chi}}|\psi^r_{n_0}\rangle\  (0\leq c_{\chi}\leq1),\quad
\end{equation}
where the DV part is $|\varphi_{n_0}\rangle=\sum_{m=0}^{n_0-1}\mu_m|i_mj_m\rangle$, and the CV part is $|\psi^r_{n_0}\rangle=\sqrt{1-\chi^2}\sum\nolimits_{n=n_0}^{\infty}\chi^{n-n_0}|i_n j_n\rangle$ ($\chi\equiv \tanh r$). We denote its Schmidt-value vector by 
\begin{eqnarray}\label{eq-SCV}
    \vec{\lambda}_{c_{\chi},\vec{\mu},\vec{\chi}}=(1-c_{\chi}) \vec{\mu}\oplus c_{\chi} \vec{\chi},
\end{eqnarray}
with Schmidt-value vectors $\vec{\mu}=(\mu_0,\dots,\mu_{n_0-1})^{\mathrm{T}}$ and $\vec{\chi}=(1-\chi^2)(1,\chi^2,\dots)^{\mathrm{T}}$ for $|\varphi_{n_0}\rangle$ and 
$|\psi^r_{n_0}\rangle$, respectively.

Analogous to the G-G entanglement concentration in the main text, here we consider building entanglement concentration for this class of non-Gaussian states. We aim to construct a column-stochastic matrix $\mathbf{D}_{\vec{\mu}}=\mathbf{D}^{\mathcal{L}}_{\vec{\mu}_{(1)}}\otimes\mathbf{D}^{\mathcal{A}}_{\vec{\mu}_{(2)}}$ of a tensor-product form, that satisfies $\vec{\lambda}_{c_{\chi_1},\vec{\mu}_{(1)},\vec{\chi}_1}
    \otimes \vec{\lambda}_{c_{\chi_2},\vec{\mu}_{(2)},\vec{\chi}_2}
=\mathbf{D}_{\vec{\mu}}\left(\vec{\lambda}_{c_{\chi},\vec{\mu},\vec{\chi}}\otimes{\vec{e}_1}\right)$
given $\vec{\mu}_{(1)}=(\mu_0^{(1)},\dots,\mu_{n_0-1}^{(1)})^{\mathrm{T}}$, $\vec{\mu}_{(2)}=(\mu_0^{(2)},\dots,\mu_{n_2-1}^{(2)})^{\mathrm{T}}$, and $\vec{e}_1=(1,0,0,\dots)^{\mathrm{T}}$.
If $\mathbf{D}_{\vec{\mu}}$ can be constructed, then we have the fact: the tensor product of the two {bipartite pure} states with Schmidt-value vectors $\vec{\lambda}_{c_{\chi_1},\vec{\mu}_{(1)},\vec{\chi}_1}$ and $\vec{\lambda}_{c_{\chi_2},\vec{\mu}_{(2)},\vec{\chi}_2}$ can be converted into the tensor product of a state with Schmidt-value vector $\vec{\lambda}_{c_{\chi},\vec{\mu},\vec{\chi}}$ and a separable state by deterministic LOCC. 

First, {we construct the matrices $\mathbf{D}^{\mathcal{L}}_{\vec{\mu}_{(1)}}$ and $\mathbf{D}^{\mathcal{A}}_{\vec{\mu}_{(2)}}$}, with the sum of each column being 1, in a block form: 
\begin{eqnarray}
    \mathbf{D}^{\mathcal{L}}_{\vec{\mu}_{(1)}}=
    \begin{pmatrix}
        \mathbf{A} & \mathbf{0}\\
        \mathbf{0} & \mathbf{D}^{\mathcal{L}}
    \end{pmatrix}
\end{eqnarray}
and
\begin{eqnarray}
    \mathbf{D}^{\mathcal{A}}_{\vec{\mu}_{(2)}}=
  \begin{pmatrix}
    (1-c_{\chi_2})\mathbf{B} \\
    c_{\chi_2}\mathbf{D}^{\mathcal{A}}+(1-c_{\chi_2})\mathbf{C}
  \end{pmatrix},
\end{eqnarray}
respectively.
Here, {$\mathbf{A}$ is an arbitrary $n_0\times n_0$ real matrix with the sum of each row not greater than $1/\eta$, and it is used to convert vector $\vec{\mu}$ to $\vec{\mu}_{(1)}$.
The matrix $\mathbf{D}^{\mathcal{L}}$, given by 
\begin{equation}
    \mathbf{D}^{\mathcal{L}}=
    \begin{bmatrix}
1      & 1-\eta          & (1-\eta)^2              & (1-\eta)^3                & \dots \\
0      & {1\choose1}\eta & {2\choose1}\eta(1-\eta) & {3\choose1}\eta(1-\eta)^2 & \dots \\
0      & 0               & {2\choose2}\eta^2       & {3\choose2}\eta^2(1-\eta) & \dots \\
0      & 0               & 0                       & {3\choose3}\eta^3         & \dots \\
\vdots & \vdots          & \vdots                  & \vdots                    &\ddots
   \end{bmatrix},
\end{equation}
gives rise to $\mathbf{D}^{\mathcal{L}}\vec{\chi}=\vec{\chi}_{1}$ for $\chi=\chi_1/\sqrt{\eta\left(1-\chi_1^2\right)+\chi_1^2}$.
The matrix $\mathbf{B}$ has the following form which can transform vector $\vec{e}_1$ into $\vec{\mu}_{(2)}$:
\begin{eqnarray}
    \mathbf{B}=
   \begin{pmatrix}
    \mu_0^{(2)}         &   0     & \dots   \\
    \mu_1^{(2)}         &   0     & \dots   \\
    \vdots              & \vdots  & \dots   \\
    \mu_{n_2-1}^{(2)}   &   0     & \dots
   \end{pmatrix}.
\end{eqnarray}
The matrix $\mathbf{D}^{\mathcal{A}}$, given by 
\begin{eqnarray}
    \mathbf{D}^{\mathcal{A}}=
    \begin{pmatrix}
a      & 0                  & 0                & 0       & \dots \\
ab     & a^2                & 0                & 0       & \dots \\
a b^2  & {2\choose1}a^2 b   & a^3              & 0       & \dots \\
a b^3  & {3\choose1}a^2 b^2 & {3\choose2}a^3 b & a^4     & \dots \\
\vdots & \vdots             & \vdots           & \vdots  &\ddots
   \end{pmatrix},
\end{eqnarray}
can convert $\vec{e}_1$ to $\vec{\chi}_2$ for $\chi_2=\sqrt{1-G^{-1}}$.
The matrix $\mathbf{C}$ is constructed as
\begin{eqnarray}
    \mathbf{C}=
   \begin{pmatrix}
    0        &        &         &        &        \\
    0        & 1      &         &        &        \\
    0        & 0      & 1       &        &        \\
    0        & 0      & 0       & 1      &        \\
    \vdots   & \vdots & \vdots  & \vdots & \ddots 
   \end{pmatrix},
\end{eqnarray}
which is used to make sure that the sum of each column of $\mathbf{D}^{\mathcal{A}}_{\vec{\mu}_{(2)}}$ is $1$.
Then, the vector  $\vec{\lambda}_{c_{\chi},\vec{\mu},\vec{\chi}}\otimes{\vec{e}_1}$ can be transformed into a vector $\vec{\lambda}_{c_{\chi_1},\vec{\mu}_{(1)},\vec{\chi}_1}
    \otimes \vec{\lambda}_{c_{\chi_2},\vec{\mu}_{(2)},\vec{\chi}_2}$ by the matrix $\mathbf{D}_{\vec{\mu}}=\mathbf{D}^{\mathcal{L}}_{\vec{\mu}_{(1)}}\otimes\mathbf{D}^{\mathcal{A}}_{\vec{\mu}_{(2)}}$ for $c_{\chi}=c_{\chi_1}$ and $\chi=\chi_1/\sqrt{\eta\left(1-\chi_1^2\right)+\chi_1^2}$.}
The matrix $\mathbf{D}_{\vec{\mu}}$ is a column-stochastic matrix iff the sum of its each row is no greater than $1$, i.e.,~$\eta\geq 1-c_{\chi_2}\chi_2^2$ where $1-c_{\chi_2}\chi_2^2$ is the maximum value among all the row sums of $\mathbf{D}^{\mathcal{A}}_{\vec{\mu}_{(2)}}$.
We thus have
\begin{eqnarray}
   \vec{\lambda}_{c_{\chi_1},\vec{\mu}_{(1)},\vec{\chi}_1}
    \otimes \vec{\lambda}_{c_{\chi_2},\vec{\mu}_{(2)},\vec{\chi}_2}
    \prec \vec{\lambda}_{c_{\chi},\vec{\mu},\overrightarrow{\chi}} \otimes\vec{e}_1
\end{eqnarray}
for $\eta\geq1-c_{\chi_2}\chi_2^2$ when $\vec{\lambda}_{c_{\chi_1},\vec{\mu}_{(1)},\vec{\chi}_1}
    \otimes \vec{\lambda}_{c_{\chi_2},\vec{\mu}_{(2)},\vec{\chi}_2}
=\mathbf{D}_{\vec{\mu}}\left(\vec{\lambda}_{c_{\chi},\vec{\mu},\vec{\chi}}\otimes\vec{e}_1\right)$.
As a direct consequence, we arrive at the following result:
\begin{theorem}
    Let $|\varphi_{1}\rangle,|\varphi_{2}\rangle,|\varphi\rangle$ be three states in the form of $|\varphi_{D-C}\rangle$ [Eq.~\eqref{eq-GeneralCV}] with Schmidt-value vectors $\vec{\lambda}_{c_{\chi_1},\vec{\mu}_{(1)},\vec{\chi}_1}$, $\vec{\lambda}_{c_{\chi_2},\vec{\mu}_{(2)},\vec{\chi}_2}$, and $\vec{\lambda}_{c_{\chi_1},\vec{\mu},\vec{\chi}}$, respectively.
    If they satisfy
    \begin{eqnarray}
        \vec{\mu} &=& \mathbf{A}\vec{\mu}_{(1)},\nonumber\\
        \chi &=& \frac{\chi_1}{\sqrt{\eta(1-\chi_1^2)+\chi_1^2}},\nonumber\\        
        \mathbf{D}^{\mathcal{A}}_{\vec{\mu}_{(2)}}\vec{}e_1 &=& \vec{\lambda}_{c_{\chi_2},\vec{\mu}_{(2)},\vec{\chi}_2},\nonumber
    \end{eqnarray}
    and
    \begin{eqnarray*}\label{ineq-eta}
        \eta\geq
        1-c_{\chi_2}\chi_2^2,
    \end{eqnarray*} 
    then the tensor state $|\varphi_{1}\rangle\otimes|\varphi_{2}\rangle$ can be deterministically converted, via LOCC, into the tensor product of the entanglement state $|\varphi\rangle$ and a vacuum state {[Fig.~\ref{fig_general_concentration2}]}.
\end{theorem}

\begin{figure}
    \centering
    \includegraphics[width=230pt]{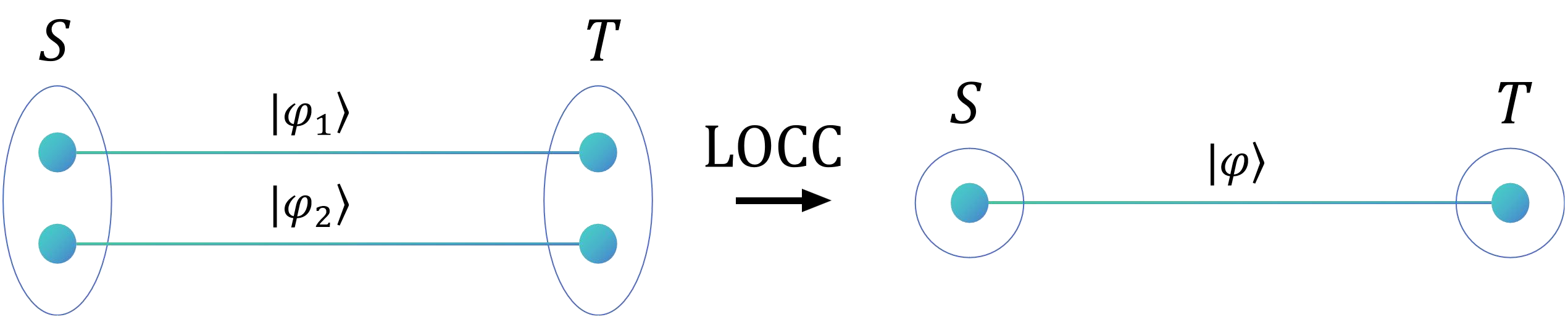}
    \caption{{General entanglement concentration. Initially, two generically entangled states $|\varphi_1\rangle$ and $|\varphi_2\rangle$ [Eq.~\eqref{eq-GeneralCV}] are shared between $S$ and $T$. LOCC-based entanglement concentration can convert these states into a more entangled state $|\varphi\rangle$ accompanying a vacuum state (omitted here).}}
    \label{fig_general_concentration2}
\end{figure}

\emph{Generalized parallel rule.---}It is obvious that $\chi$ reaches the maximum value 
when $\eta$ takes its minimum value, $\eta=1-c_{\chi_2}\chi_2^2$.
In this maximum scenario, we have 
\begin{eqnarray*}
    \frac{\chi^2}{1-\chi^2}
    &=&\frac{\chi_1^2}{(1-\chi_1^2)(1-c_{\chi_2} \chi_2^2)}\nonumber\\
    &=&\eta_p\frac{\chi_1^2}{(1-\chi_1^2)(1-\chi_2^2)},
\end{eqnarray*}
where
\begin{equation}
    \eta_p=\frac{1-\chi_2^2}{1-c_{\chi_2}\chi_2^2}<1.
\end{equation}
We thus derive the generalized parallel rule,
\begin{eqnarray}\label{eq-eta_p}
    \sinh r=\eta_p\sinh r_1\cosh r_2,
\end{eqnarray}
given $\chi\equiv \tanh r$,
$\chi_1 \equiv \tanh r_1$, and
$\chi_2\equiv \tanh r_2$.

Note that by the previous analysis in Sec.~\ref{sec-general-rule}, the condition $\eta_{p}<\sqrt{k-1}$ does not rule out the existence of mixed-order phase transition. Thus, we expect that the abrupt jump $0\to \mathrm{X}_\text{SC}^{+}$ in the order parameter may still be present, even when taking into account generalized CV operations (i.e.,~GPOVM) and non-Gaussian states (i.e.,~$|\varphi_{D-C}\rangle$).

\section{Feedback Control for QN}

Feedback control has been extensively studied and used for quantum systems~\cite{feedback_coherent1994,quantum_control1999,feedback_coherenct2000,feedback2002,feedback2002_2,feedback2005,quantum_control2008Domenico,quantum_control2010,feedback2012,feedback2012_2,feedback_remote_entanglement2015,feedback2017}. It plays a crucial role in achieving precise manipulation of quantum states, which is essential for applications in quantum computing, communication, and metrology. The core objective is to analyze, design, and experimentally validate techniques that guide the dynamics of quantum systems toward desired outcomes through external, time-dependent manipulation. The measurement-based feedback control~\cite{feedback2009} is one of the primary strategies, relying on acquiring information about the system's state through quantum-state tomography~\cite{Tomography_qudit2002,Tomography_entanglement_CV2002}. The information is then used to design the control laws aimed at enhancing control performance~\cite{feedback_measurement2006,feedback_measurement2008Yamamoto,quantum_control2013Qi} in quantum feedback.

{In the following, we propose a measurement-based feedback control strategy to
establish and (more importantly) stabilize the entanglement shared between parties $S$ and $T$ in a QN. This strategy utilizes real-time monitoring of the shared state's entanglement and enables adaptive adjustments based on measurement outcomes, ensuring reliable entanglement transmission. By performing state tomography between $S$ and $T$, 
one can determine whether entanglement has been successfully established. 
The difference between the measured entanglement and a target value is continuously fed back to all initial resource states (QN links). Based on this information, the links perform local operations designed to adjust their entanglement, thus driving the final entanglement between $S$ and $T$ closer to a target value.}

Specifically, for a bipartite Gaussian state, the ratio negativity is directly determined by the logarithmic negativity~\cite{Ratio_negativity2024} which can be experimentally quantified through state tomography~\cite{Tomography_Gaussian2009,Tomography_Gaussian2010,Tomography_higher_dim_adaptive2020,Tomography2024}. 
Suitable local operations for controlling and enhancing entanglement for CV states are also known~\cite{control_CV2018}. These established settings can be adopted within our CV-based QN model to implement quantum feedback control.

{\begin{figure}[t!]
    \centering
    \subfigure[]{
     \includegraphics[width=115pt]{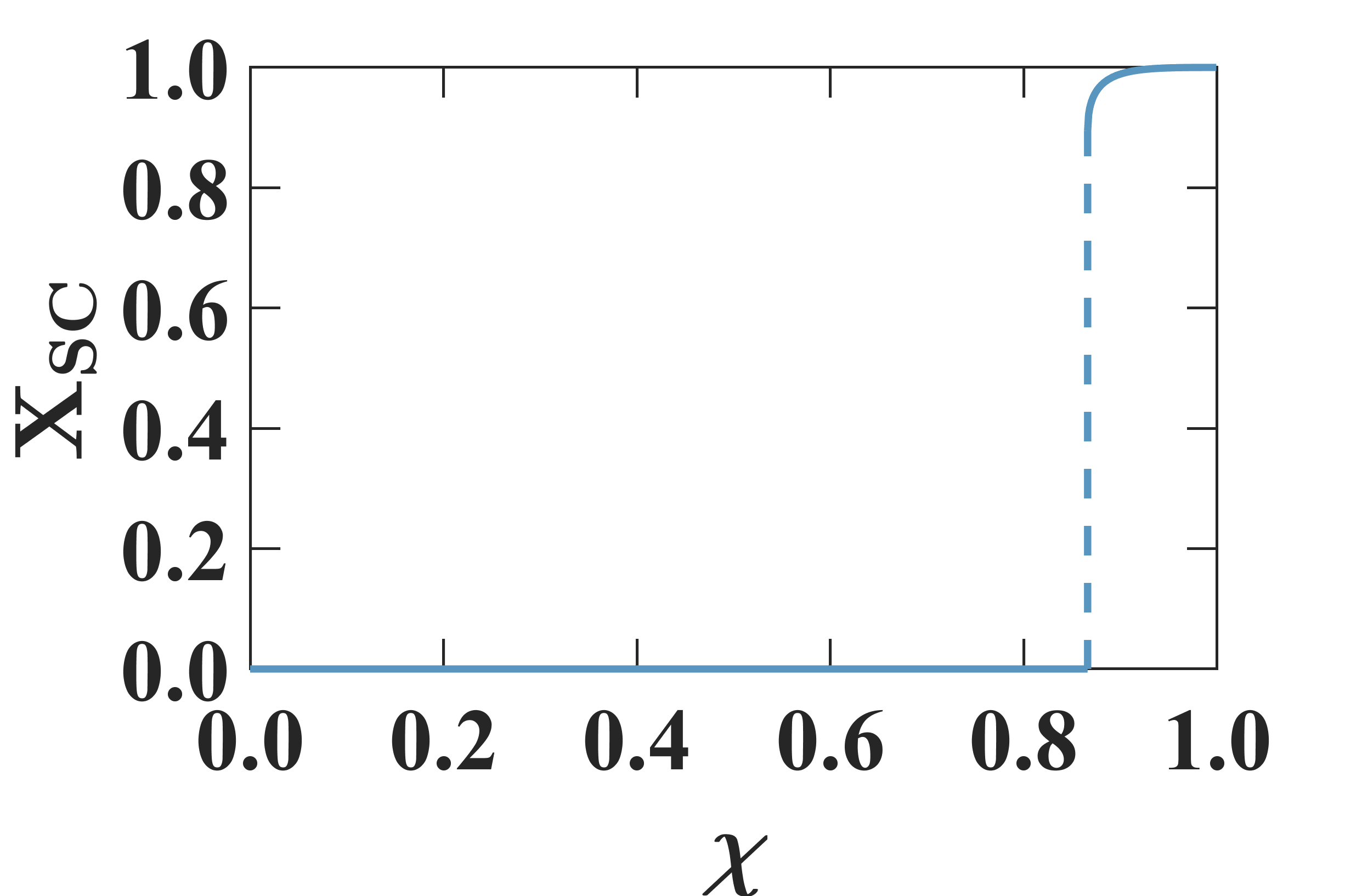}\vspace{-5mm}\label{fig_percolation_CV}
    }
    \subfigure[]{
     \includegraphics[width=115pt]{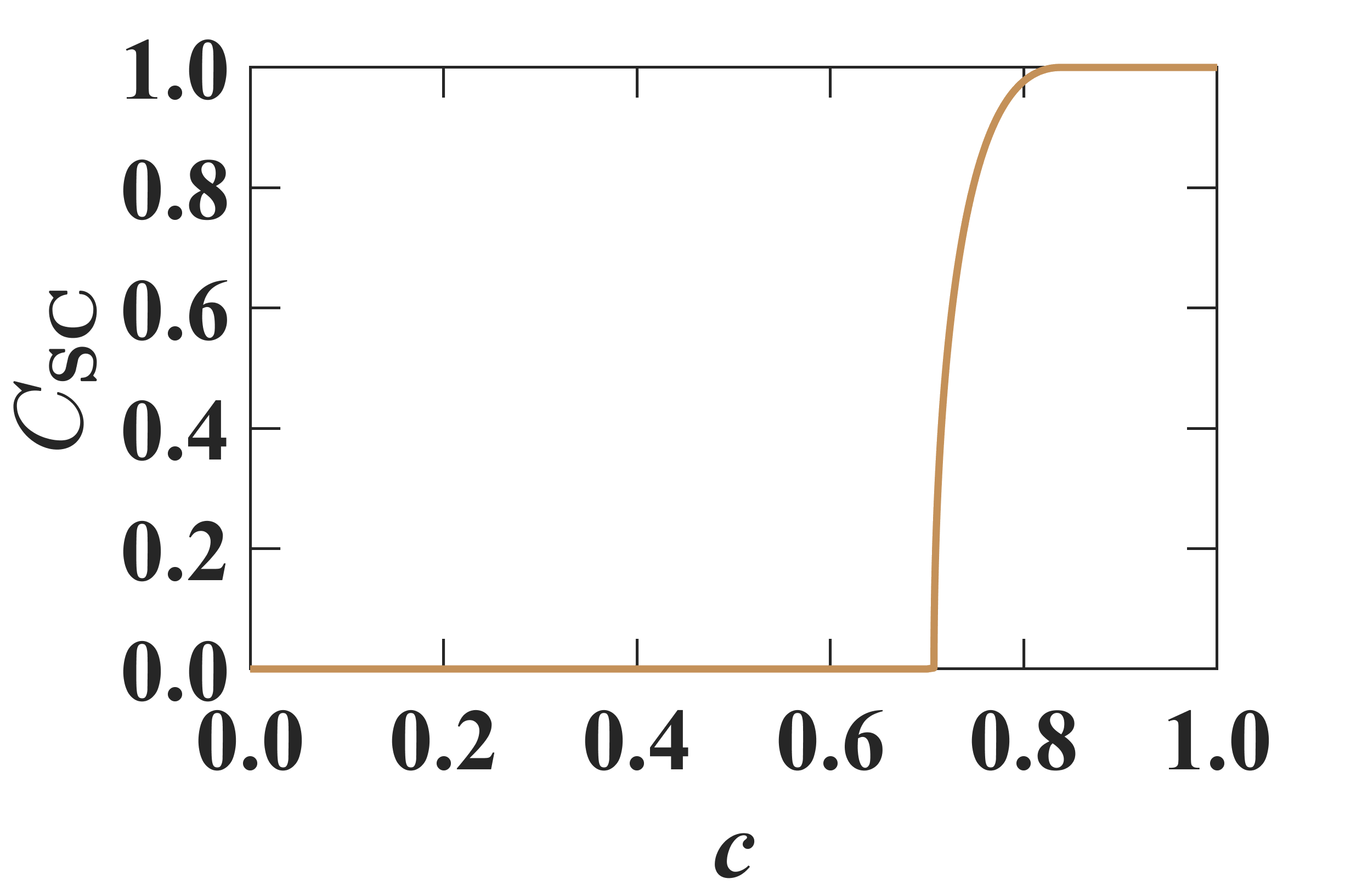}\vspace{-5mm}\label{fig_percolation_DV}
    }
    \caption{{Entanglement transmission on CV- and DV-based QNs (on Bethe lattice, $k=3$).~\subref{fig_percolation_CV}~Sponge-crossing ratio negativity $\mathrm{X}_{\text{SC}}$ of the CV-based QN related to the ratio negativity $\chi$ of each initial resource TMSVS.
    ~\subref{fig_percolation_DV}~Sponge-crossing concurrence $C_{\text{SC}}$ of the DV-based QN related to the concurrence of each initial 2-qubit resource state.} }\label{fig_percolation_CV_vs_DV}
\end{figure}}

\subsection{Feedback model}
\label{sec-PID}

{{To proceed, we model both the entanglement change and the compensation by feedback control using the first-order-plus-time-delay (FOPTD) model~\cite{PID_control1995} (which is also given in the main paper), 
\begin{equation}
\label{eq_foptd_si}
    \frac{d \chi(t)}{d t}=-\tau^{-1}\chi(t)+u(t-T_0),
\end{equation}
a common model used in control theory that has also been applied to quantum controls~\cite{control_CV2018}. Here, $\tau^{-1}$ measures the decay speed, and $u(t)$ denotes the control action applied after an activation time $T_0$. We propose to implement the feedback control term $u(t)$ using a proportional-integral-derivative (PID) controller~\cite{PID_control1995,PID_control2005}, one of the most widely used classical control heuristics that has also been applied to the quantum domain, particularly in the generation of entangled states~\cite{PID_quantum_entanglement2020,PID_quantum2020}.} This prompts us to write $u(t)$ as
\begin{eqnarray}\label{eq-u}
    u(t) = \begin{cases} 
      0, & t\le 0, \\
      K_P \epsilon(t) + K_I \int_0^{t}dt'\epsilon(t') + K_D d \epsilon(t)/dt, & t>0 ,
   \end{cases}\nonumber\\
\end{eqnarray}
where $\epsilon(t)=\mathrm{X}^\text{target}_\text{SC} - \mathrm{X}_\text{SC}(t)$ is the error between the desired and actual entanglement in the QN. }
{The three distinct terms each address the error signal $\epsilon(t)=\mathrm{X}^\text{target}_\text{SC} - \mathrm{X}_\text{SC}(t)$ differently:
\begin{enumerate}
    \item The Proportional ($P$) term, scaled by the gain $K_P$, provides an output directly proportional to the instantaneous error $\epsilon(t)$, ensuring an immediate response to deviations.
    \item The Integral ($I$) term, scaled by the gain $K_I$, integrates the error over time.
    This accumulation addresses past errors and is crucial for eliminating persistent steady-state deviations between the target $\mathrm{X}^\text{target}_\text{SC}$ and the system state $\mathrm{X}_\text{SC}(t)$.
    \item The Derivative ($D$) term, scaled by the gain $K_D$, acts on the rate of change of the error. 
    This component anticipates future error trends based on the current slope, often providing damping to reduce oscillations and improve stability.
\end{enumerate}}

\begin{figure}[t!]
    \centering
     \includegraphics[height=80pt]{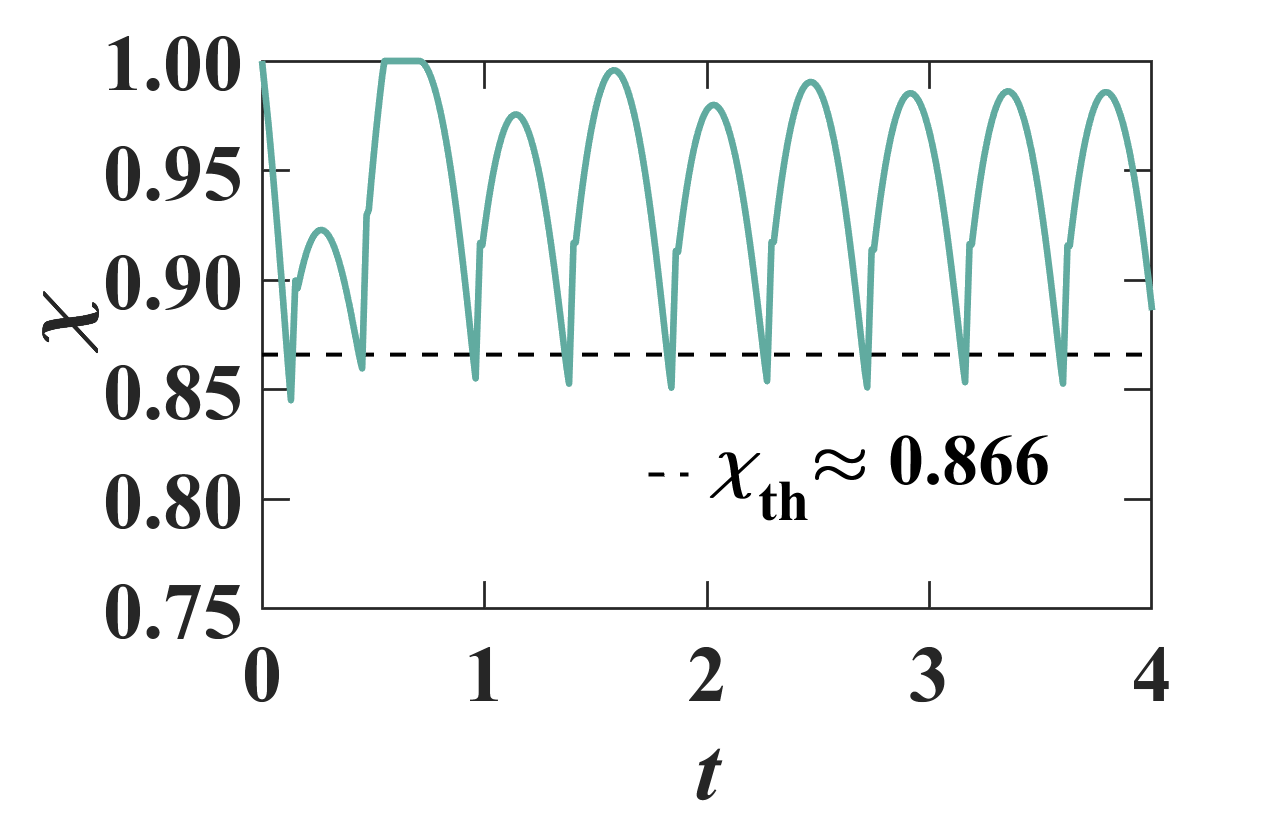}
    \caption{Output $\chi$ for FOPTD.}\label{fig_PID_CV_t_chi}
\end{figure}

\begin{figure*}[t!]
    \centering
    \hspace{-6mm}
     \subfigure[]{
     \includegraphics[height=70pt]{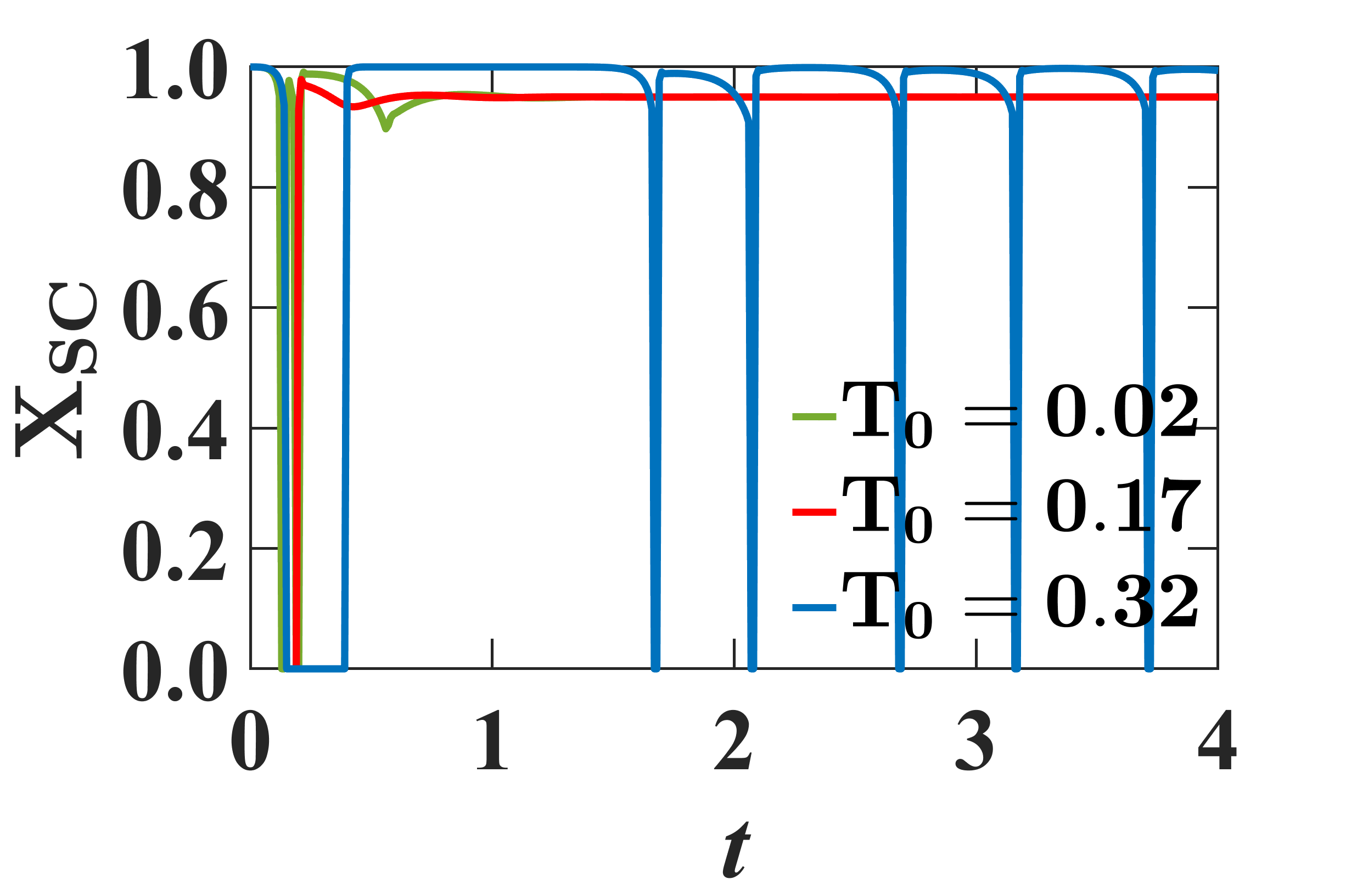}\hspace{-5mm}
     \label{fig_PID_CV_T0}
     }
     \subfigure[]{
     \includegraphics[height=70pt]{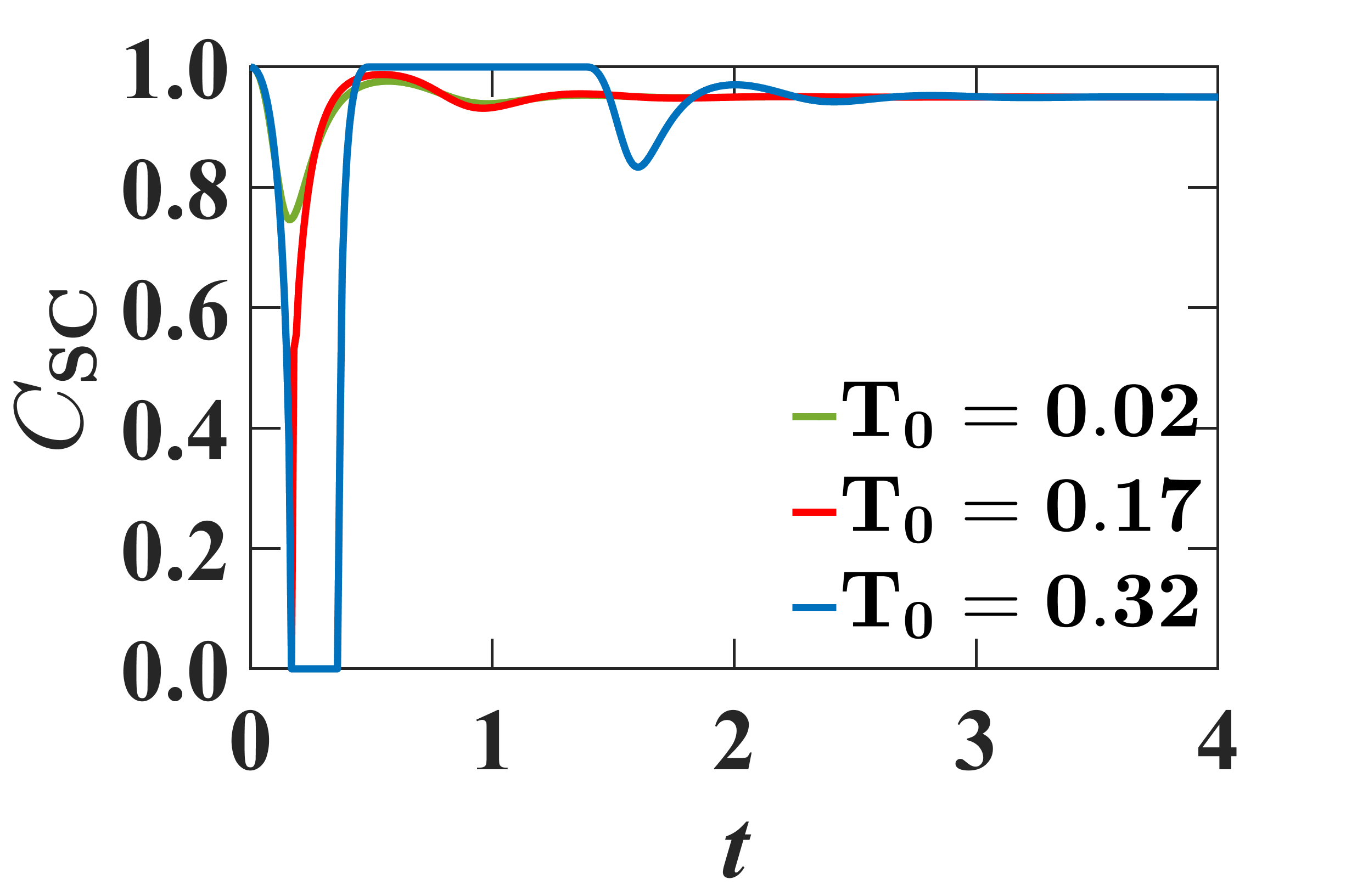}\hspace{-5mm}
     \label{fig_PID_DV_T0}
     }
     \subfigure[]{
     \includegraphics[height=70pt]{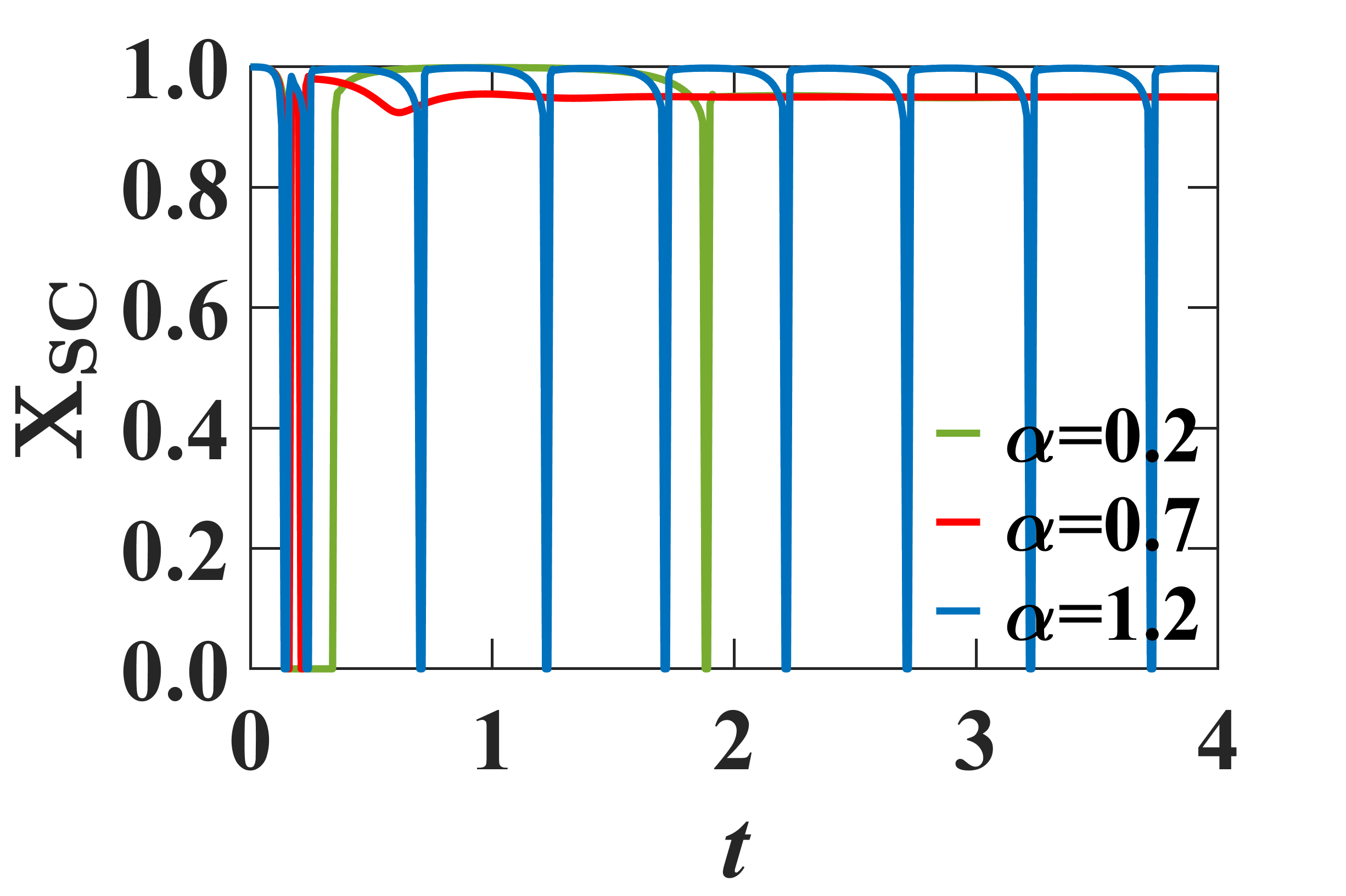}\hspace{-5mm}\label{fig_PID_CV_alpha}
     }
     \subfigure[]{
     \includegraphics[height=70pt]{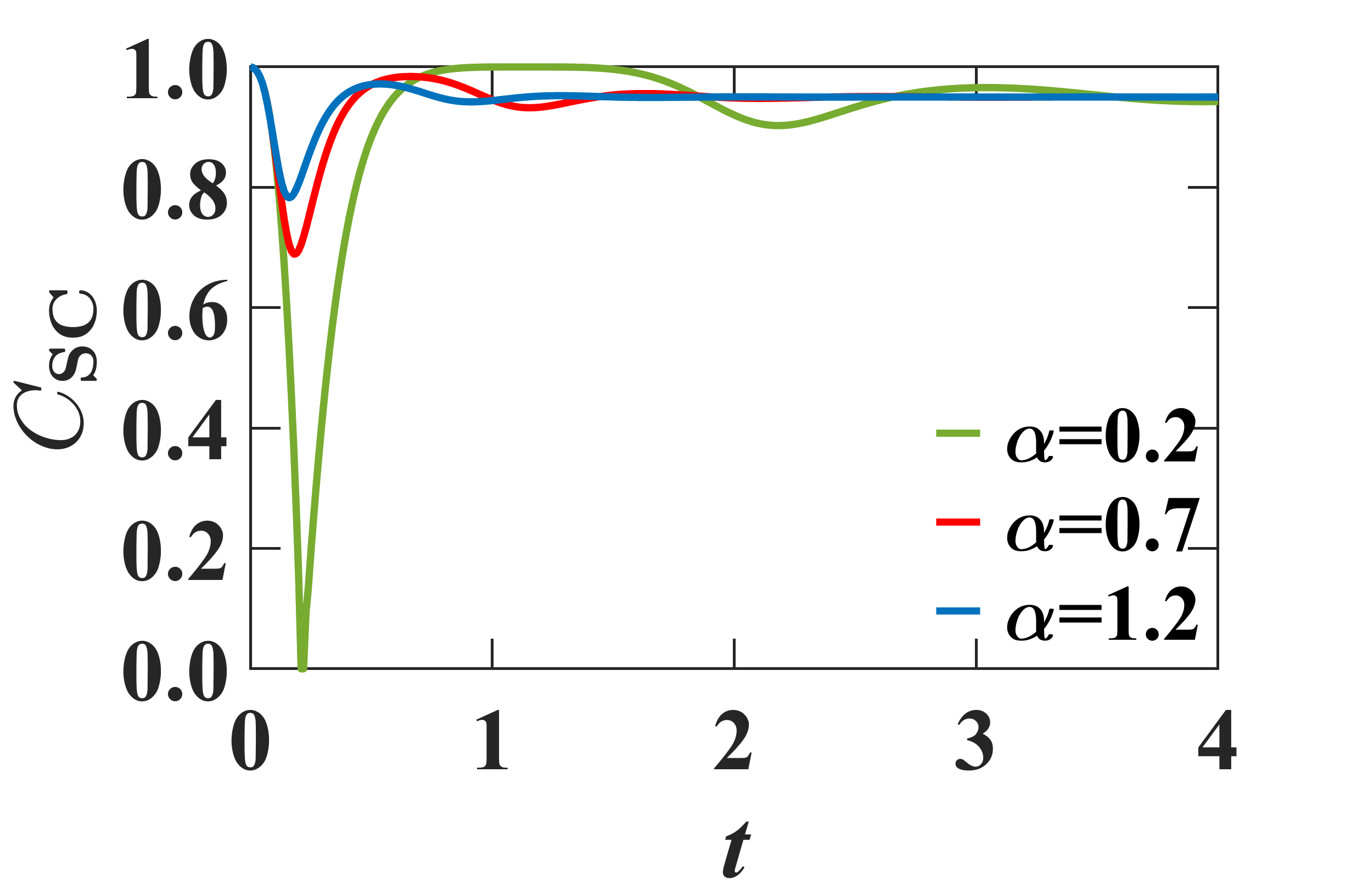}\hspace{-5mm}\label{fig_PID_DV_alpha}
     }
     \subfigure[]{
     \includegraphics[height=70pt]{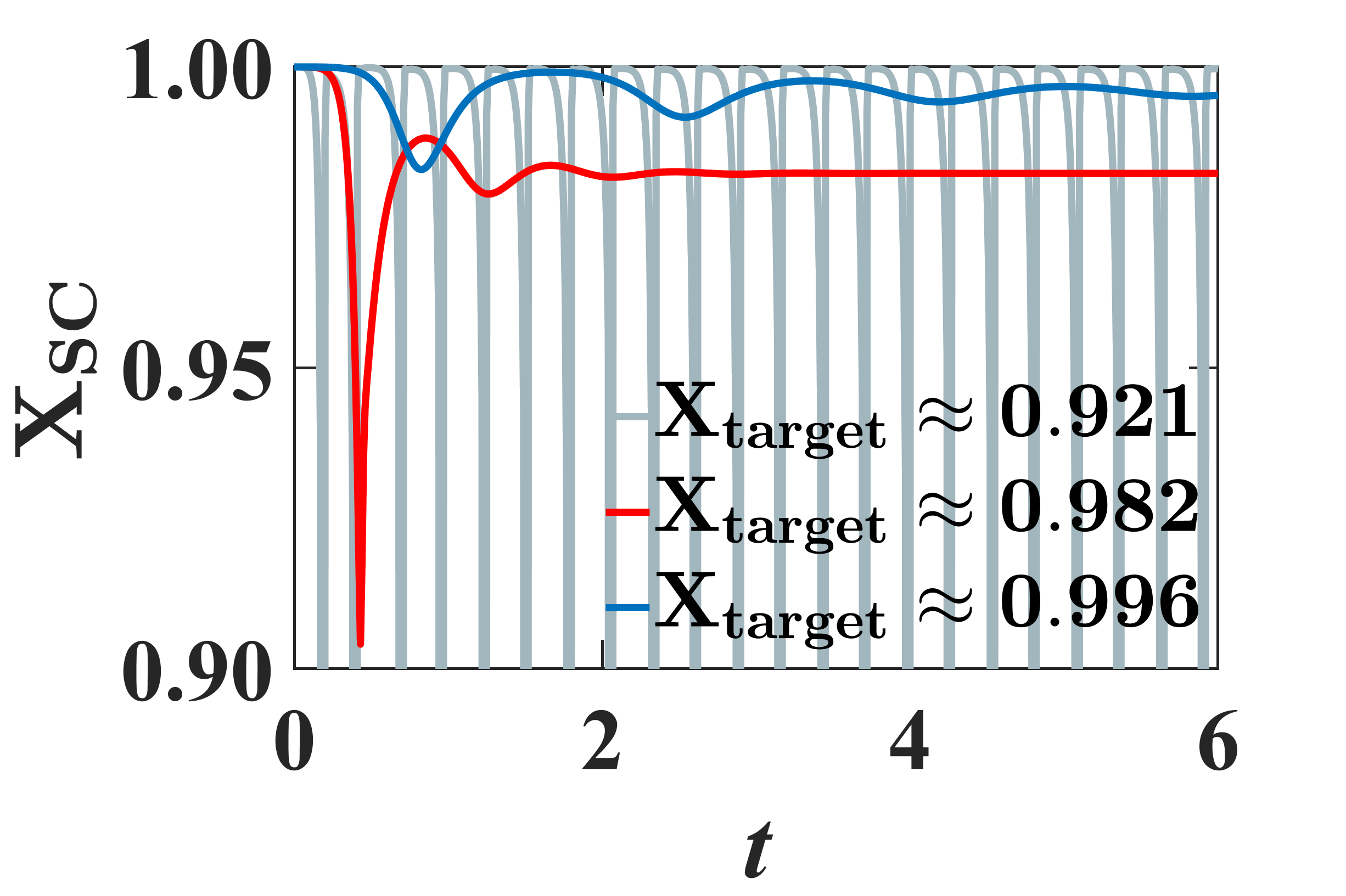}\hspace{-5mm}\label{fig_PID_CV_Xtarget}
     }
    \caption{{Parameters of PID feedback control (Bethe lattice, $k=3$). 
    \subref{fig_PID_CV_T0},\subref{fig_PID_DV_T0}~Different activation times $T_0$ on CV- and DV-based QN.
    \subref{fig_PID_CV_alpha},\subref{fig_PID_DV_alpha}~Different control gains $\alpha$ on CV- and DV-based QN.
    \subref{fig_PID_CV_Xtarget}~Different target entanglement $\mathrm{X}_{\text{SC}}^{\text{target}}$ for CV-based QN. 
    }
    }\label{fig_SI_PID}
\end{figure*}

\subsection{Comparison between CV- and DV-based QN} 
\label{sec_5b}

We present a comparative analysis of feedback control performance between the G-G DET scheme on CV-based QNs and the DET scheme~\cite{XJS2021} on DV-based QNs. 
We will show that the characteristic discontinuity---inherent to mixed-order phase transition---necessitates a more demanding feedback design for reliable entanglement transmission in CV-based QNs.

\emph{CV results.---}Again, we focus on an infinite-size Bethe lattice (as in Sec.~\ref{sec-Bethe}) with degree $k=3$ as the QN topology. 
We consider the following PID control scenario for our CV-based QNs where the critical threshold and the saturation point of $\chi$ (defined in Secs.~\ref{sec-pl}~and~\ref{sec-saturation}) are $\chi_{\text{th}}=\sqrt{3}/2\approx 0.867$ and $\chi_{\text{sat}}=1$ [Fig.~\ref{fig_percolation_CV}], respectively. Starting at $t=0$, all resource states' entanglement $\chi(t)$ undergoes a decay with $\tau=1$ 
from the saturation point ($\chi(0)=\chi_{\text{sat}}$). At $t=T_0$, the PID control [Eq.~\eqref{eq-u}] is activated with $K_P=2.2$, $K_I=100$, and $K_D=0.01$, targeting $95\%$ of perfect entanglement ($\mathrm{X}^\text{target}_\text{SC}=0.95$). We observe that when $\chi(t)$ falls below the threshold during this control process, the CV-based QN displays pronounced oscillations in entanglement percolation process, preventing stabilization. For instance, under the setting of $T_0=0.02$, the ratio negativity $\chi(t)$ approximates $0.845$ at $t=0.13$, falling below the threshold $\chi_{\text{th}}$ and resulting in the network's inability to maintain stable operation [Fig.~4(b) in the main text]. This is since due to the discontinuous nature of the response curve near the critical threshold, even a small variation in $\chi$ can trigger a violent dynamical response in $\mathrm{X}_\text{SC}$.

{We find that the instability of $\mathrm{X}_{\text{SC}}$ is not a result of any sudden change in $\chi$, which only undergoes small perturbations near the critical threshold $\chi_{\text{th}}\approx 0.866$ (Fig.~\ref{fig_PID_CV_t_chi}). 
In nonlinear control theory~\cite{L_Gain1966,L_Gain1975}, any physically realizable controller must satisfy the local Lipschitz condition with finite gain $g$: $\|\chi(t_1)-\chi(t_2)\|\leq g\|u(t_1)-u(t_2)\|$. This condition guarantees that the FOPTD model can produce only small variations in $\chi$, indicating that the instability of $\mathrm{X}_\text{SC}$ does not arise from the FOPTD dynamics. Instead, it results directly from the discontinuous nature of the mixed-order phase transition in $\mathrm{X}_\text{SC}$ vs.~$\chi$. }

\emph{DV results.---}{For comparison, now we  substitute all initial resource states $|\psi^r\rangle$ with two-qubit pure states $|\psi_{\lambda}\rangle:=\sqrt{\lambda}|00\rangle+\sqrt{1-\lambda}|11\rangle$ and implement the DET scheme~\cite{XJS2021} for such DV states. 
The final concurrence, $C_{\text{SC}}$, established between $S$ and $T$ exhibits a second-order phase transition~\cite{XJS2021,OXSGBJ2022} with respect to the concurrence $c=2\sqrt{\lambda(1-\lambda)}$ of $|\psi_{\lambda}\rangle$ [Fig.~\ref{fig_percolation_DV}]. The critical threshold and the saturation point of $c$ are given by $c_{\text{th}}=2^{-1/2}\approx 0.707$ and $c_{\text{sat}}=\sqrt{2-2^{2/3}}/\sqrt{2^{2/3}-1}\approx 0.838$~\cite{XJS2021,OXSGBJ2022}.}

{We change $\chi(t)$ to $c(t)$ in Eq.~\eqref{eq_foptd_si} and substitute the initial condition $\chi(0)=\chi_{\text{sat}}$ and the target value $\mathrm{X}^\text{target}_\text{SC}=0.95$ with $c(0)=c_{\text{sat}}$ and $C^\text{target}_\text{SC}=0.95$, respectively, while maintaining identical other operational conditions and parameter settings. We find that the DV-based QN (under the DET scheme in Ref.~\cite{XJS2021}) demonstrates a smoother and more robust behavior [Fig.~4(a) {in the main text}]. 
The nature of the continuous second-order phase transition in DV-based QNs~\cite{XJS2021} inherently limits the range of fluctuations compared to CV systems, which contributes to the stability and consistency of the control process. This stability suggests the potential effectiveness of control schemes (such as the PID control) for DV systems.}

\subsection{Control parameter analysis}
\label{sec_5c}

To translate the insights gained from our comparative analysis into specific designs, this subsection systematically investigates how variations in key control parameters affect stabilization performance in our CV‑based QN. In particular, we consider three parameters: activation time \(T_0\), PID gain scaling \(\alpha\), and target entanglement \(\mathrm{X}_{\mathrm{SC}}^{\mathrm{target}}\), to chart stability boundaries. For the studies, we hold the PID gains and decay timescale fixed at $K_P=2$, $K_I=30$, $K_d=0.007$, and $\tau=1$.

{\begin{enumerate}
    \item \emph{Activation time $T_0$.---}The control activation time $T_0$ exhibits a sharp threshold-dependent behavior. Numerical simulations shown in Fig.~\ref{fig_PID_CV_T0} reveal that when $T_0$ is large enough (e.g.,~$T_0=0.32$), the value of $\chi$ decays below the critical threshold $\chi_{\text{th}}$ before the activation of PID control, causing the controller to fail in restoring the system to a stable target state. This suggests an operational constraint of $T_0$ to maintain stability.
    The analysis also implies that the precise determination of the critical threshold in mixed-order phase transitions within the CV-based QN can be achieved through parametric tuning of $T_0$, leveraging abrupt drops observed during control processes.
    In contrast, for DV-based QN [Fig.~\ref{fig_PID_DV_T0}], 
    even for a relatively large activation time ($T_0=0.32$), the feedback control still quickly stabilizes the QN system. These results demonstrate that CV-based QNs may need earlier activation of control (i.e.,~smaller $T_0$) compared to DV-based QNs.
    \item \emph{Control gain scaling $\alpha$.---}We introduce a scaling factor $\alpha$ to $u(t)$, such that all PID parameters ($K_P$, $K_I$, and $K_D$) are rescaled by a factor of $\alpha$. Simulation results (under $T_0=0.05$) on the CV-based QN [Fig.~\ref{fig_PID_CV_alpha}] show that excessive control strength induces destabilizing overshoot and persistent oscillations ($\alpha=1.2$), while insufficient strength fails to counteract decay ($\alpha=0.2$). In contrast, for the DV-based QN [Fig.~\ref{fig_PID_DV_alpha}], changes in $\alpha$ do not induce sudden collapse or significant oscillations. This hints that the control stability is more sensitive to the control gain for CV-based QNs than their DV counterparts. 
    \item \emph{Target entanglement $\mathrm{X}_{\text{SC}}^{\text{target}}$.---}We examine the performance of the PID control with $K_P=5$, $K_I=70$, $K_D=0.007$, $\tau=20$, and $T_0=0.02$. As shown in Fig.~\ref{fig_PID_CV_Xtarget}, different target values of $\mathrm{X}_{\text{SC}}^{\text{target}}$ yield noticeably different control behaviors. When the target is set near $\mathrm{X}_{\text{SC}}^{+}$ [namely $\mathrm{X}_{\text{SC}}^{\text{target}}\approx 0.921$, cf.~Fig.~\ref{Bethe_infty}], such that $\chi(0)$ is close to $\chi_{\text{th}}$, the system exhibits frequent and strong oscillations, indicating that the QN fails to operate properly. When the target is set too high ($\mathrm{X}_{\text{SC}}^{\text{target}}\approx 0.996$), the response becomes overly slow, causing substantial resource waste.  We may even estimate the resource waste for each initial resource state $\chi$ by calculating the integral
    \begin{equation}
        \int_{\chi(t)>\chi^{\text{target}}}\chi(t)dt,
        \vspace{1mm}
    \end{equation}
    where $\chi^{\text{target}}$ is the entanglement value when $\mathrm{X}_{\text{SC}}$ matches $\mathrm{X}_{\text{SC}}^{\text{target}}$. 
    The resource waste for $\mathrm{X}_{\text{SC}}^{\text{target}}\approx 0.982$ and $\mathrm{X}_{\text{SC}}^{\text{target}}\approx 0.996$ amounts to values of approximately $0.027$ and $0.057$, respectively. 
    This comparison highlights the importance of target selection in feedback control, where overly aggressive or overly conservative targets can degrade performance, while an appropriately tuned target leads to stable and efficient control.
\end{enumerate}}

In summary, by introducing an FOPTD model paired with a PID controller to stabilize decaying entanglement, we can directly compare CV and DV feedback controls under identical QN topology and control settings. Our results demonstrate not only how the mixed‑order phase transition
in CV links can undermine standard control loops, but also how the control parameters like activation time, gain scaling, and target entanglement can affect stability and provide practical guidelines for deploying robust feedback in CV‑based QN.

\bibliography{NegPT}